\documentclass[preprint]{new_tlp}

\usepackage[table]{xcolor}
\usepackage{xspace}
\usepackage{epsfig}
\usepackage{latexsym,amssymb,stmaryrd}
\usepackage{fancybox,epic}
\usepackage{algorithm}
\usepackage{amsmath}
\usepackage{colortbl}
\usepackage{color}
\usepackage{graphics}
\usepackage{epsfig}
\usepackage{graphicx}
\usepackage[draft]{fixme}
\usepackage[color=green!20]{todonotes}
\usepackage{tikz}
\usepackage{wrapfig}
\usepackage{multicol}
\usepackage{multirow}
\usepackage{pifont}
\usepackage{hyperref}
\usepackage{xparse}
\usepackage{booktabs}
\usepackage{csvsimple}
\usepackage{arydshln}
\usepackage{pgfplots}
\usepackage{array}

\usepackage[xcolor]{changebar}
\nochangebars

\usepackage{listings}
\include{costa-defs}
\usepackage{ulem}

\def\dotuline{\bgroup
  \ifdim\ULdepth=\maxdimen  
   \settodepth\ULdepth{(j}\advance\ULdepth.4pt\fi
  \markoverwith{\begingroup
  \advance\ULdepth0.08ex
  \lower\ULdepth\hbox{\kern.15em .\kern.1em}%
  \endgroup}\ULon}

\def\dashuline{\bgroup
  \ifdim\ULdepth=\maxdimen  
   \settodepth\ULdepth{(j}\advance\ULdepth.4pt\fi
  \markoverwith{\kern.15em
  \vtop{\kern\ULdepth \hrule width .3em}%
  \kern.15em}\ULon}

\newcommand{\TYPE}[1]{\mbox{\sffamily #1}}

\renewcommand\emph{\textit}
\renewcommand\em{\textit}

\newcommand{\tuple}[1]{\langle #1 \rangle}

\newcommand{\lst}[1]{\lstinline!#1!}

\newcommand{\dom}[0]{\ensuremath{\mathit{dom}}}

\newcommand{\match}{{\it match}}
\newcommand{\notmatch}{{\it nonmatch}}

\newcommand{\DC}{\ensuremath{\mathit{Co}}}
\newcommand{\DT}{\ensuremath{\mathit{D}}}

\newcommand{\method}{{\tt m}}

\newcommand{\returnval}[1]{\mbox{\lstinline!return #1!}\xspace}

\newcommand{\false}{\ensuremath{{\mathord{\mathit{false}}}}}

\newcommand{\types}[0]{{\ensuremath{\textit{Types}}}}

\newcommand{\rulemap}[0]{\ensuremath{\mu}}
\newcommand{\rulemaps}[0]{\ensuremath{M}}
\newcommand{\prmap}[0]{\ensuremath{\sigma}}
\newcommand{\prmaps}[1]{\ensuremath{\Sigma^{#1}}}
\newcommand{\addr}[0]{\ensuremath{\oplus}}
\newcommand{\addp}[0]{\ensuremath{\oplus}}
\newcommand{\baddr}[0]{\ensuremath{\bigoplus}}
\newcommand{\baddp}[0]{\ensuremath{\bigoplus}}
\newcommand{\restrmap}[2]{\ensuremath{{#1}|_{#2}}}
\newcommand{\renmap}[3]{\ensuremath{{#1}_{{#2}\sim{#3}}}}
\newcommand{\genf}[1]{\ensuremath{\mathit{gen}^{#1}}}
\newcommand{\nrules}[1]{\ensuremath{|{#1}|}}
\newcommand{\getrule}[2]{\ensuremath{{#2}^{#1}}}
\newcommand{\genG}[1]{\ensuremath{\mathit{genG}^{#1}}}
\newcommand{\genS}[1]{\ensuremath{\mathit{genS}^{#1}}}

\newcommand{\extrmap}[3]{\langle{#1}\rangle^{#2}_{#3}}
\newcommand{\rvars}[1]{\ensuremath{\mathit{rvars(#1)}}}
\newcommand{\variation}[3]{\ensuremath{{#1}\triangleright^{#3}{#2}}}

\newcommand{\traces}[0]{\ensuremath{\mathit{Tr}}}
\newcommand{\trace}[0]{\ensuremath{\mathcal{T}}}
\newcommand{\useful}[3]{\ensuremath{\mathit{useful}_{#1}^{#2}(#3)}}
\newcommand{\valdep}[3]{\ensuremath{{#1}\Rightarrow_{#2}{#3}}}

\newcommand{\trsteps}[1]{\ensuremath{\mathit{steps}({#1})}}
\newcommand{\trstepsnolink}[0]{\ensuremath{\mathit{steps}}}
\newcommand{\depvars}[2]{\ensuremath{\mathit{dep^{#1}({#2})}}}

\newcommand{\semrule}[2]{
\begin{array}{c}
#1 \\
\hline
#2
\end{array}
}

\newcommand{\bc}{b}
\newcommand{\bcs}{bs}
\newcommand{\fevalt}{eval_t}
\newcommand{\fevalg}{eval_g}
\newcommand{\tv}{\it lv}
\newcommand{\TV}{\mathcal{LV}}
\newcommand{\atv}{\it lv^\alpha}
\newcommand{\atvl}[1]{\it {#1}^\alpha}
\newcommand{\stkbc}{{\it bs}}

\NewDocumentCommand\astkbc{O{}}{\it bs^\alpha\IfNoValueTF{#1}{}{_#1}}
\NewDocumentCommand{\code}{m}{\lstinline[basicstyle=\sffamily]|#1|}
\NewDocumentCommand{\mi}{m}{\mathit{#1}}
\NewDocumentCommand{\cons}{}{\mi{Alt}}
\NewDocumentCommand{\dcon}{}{\mi{Co}}
\NewDocumentCommand{\intname}{}{Int}

\NewDocumentCommand{\tint}{}{\TYPE{\intname}}

\NewDocumentCommand{\tintl}{}{\TYPE{IntList}}
\NewDocumentCommand{\rbrprog}{}{P}
\NewDocumentCommand{\typedproc}{}{\mi{Proc}}
\NewDocumentCommand{\emptymap}{}{\epsilon}
\NewDocumentCommand{\emptyseq}{}{\epsilon}

\NewDocumentCommand{\ar}{}{\irstate}

\newcommand{\rrderiv}{\leadsto}
\newcommand{\arb}{ar^\alpha}
\newcommand{\abstractSep}{|}
\newcommand{\p}{\it P}
\newcommand{\palpha}{\it P^\alpha}
\newcommand{\bca}{\bc^{\alpha}}
\newcommand{\bsa}{\bcs^{\alpha}}
\newcommand{\rrabsderiv}{\leadsto_{\alpha}}
\newcommand{\irstate}{\it C}
\newcommand{\airstate}{\ensuremath{\mathit{AC}}}

\newcommand{\type}[1]{\ensuremath{\mathit{type}(#1)}}
\newcommand{\typer}[2]{\ensuremath{\mathit{type}_{#2}(#1)}}

\newcommand{\rulearrow}{\leftarrow}
\newcommand{\rrassigns}[2]{{#1}{:}{=}{#2}}
\newcommand{\expr}{{\it e}}

\newcommand{\typednorm}[2]{\ensuremath{{\parallel}{#1}{\parallel_{#2}}}}
\newcommand{\typednormann}[3]{\ensuremath{{\parallel}{#1}{\parallel^{#3}_{#2}}}}
\newcommand{\typednormorig}[2]{\ensuremath{{\parallel}{#1}{\parallel^{+}_{#2}}}}
\newcommand{\tsnorm}[1]{\ensuremath{{\parallel}{#1}{\parallel_{\mathit{ts}}}}}
\newcommand{\inttype}{\ensuremath{\mathtt{Int}}}
\newcommand{\nattype}{\ensuremath{\mathtt{Int}^+}}
\newcommand{\strtype}{\ensuremath{\mathtt{String}}}
\newcommand{\deptypes}[1]{\ensuremath{\mi{Constituents}({#1})}}
\newcommand{\tttype}[1]{\ensuremath{\mathtt{#1}}}
\newcommand{\typednorms}[0]{\ensuremath{\mi{rtypes}}}
\newcommand{\vars}[0]{\ensuremath{\mathit{vars}}}

\newcommand{\nonneg}[1]{\ensuremath{non\_neg({#1})}}

\newcommand{\rSec}[1]{Section~\ref{#1}}

\newcommand{\rDef}[1]{Definition~\ref{#1}}

\newcommand{\rFig}[1]{Figure~\ref{#1}}

\newcommand{\nmodels}{\not\models}
\newcommand{\slot}{\underline{\hspace{.5em}}}

\newtheorem{definition}{Definition}
\newtheorem{theorem}{Theorem}
\newtheorem{lemma}{Lemma}
\newtheorem{prop}{Proposition}
\newtheorem{example}{Example}

\newcommand*\circled[1]{\tikz[baseline=(char.base)]{
            \node[shape=circle,draw,inner sep=1pt, fill=black!10] (char) {\sffamily\textnormal#1};}}

\newcolumntype{L}[1]{>{\raggedright\let\newline\\\arraybackslash\hspace{0pt}}m{#1}}

\title[A Transformational Approach to Resource Analysis with Typed-norms
Inference]{A Transformational Approach to Resource Analysis with Typed-norms
Inference\thanks{This work was funded partially by the 
Spanish MICINN/FEDER, UE projects
RTI2018-094403-B-C31 and 
RTI2018-094403-B-C32,
MINECO projects 
TIN2015-69175-C4-2-R and 
TIN2015-69175-C4-1-R, 
by the CM project 
S2018/TCS-4314,
the GV project 
PROMETEO/2019/098, 
and the UPV project 
SP20180225.
Ra\'ul Guti\'errez was also supported by INCIBE program ``\emph{Ayudas para la excelencia de los equipos de investigaci\'on avanzada en ciberseguridad}''.}}

\author[E. Albert, S. Genaim, R. Guti\'errez and E. Martin-Martin]
{Elvira Albert, Samir Genaim\\
Dep. Sistemas Inform\'aticos y Computaci\'on, Universidad Complutense de Madrid \\ C/ Prof. Jos\'e Garc\'ia Santesmases 9, 28040 Madrid, Spain\\
\email{elvira@fdi.ucm.es, samir.genaim@fdi.ucm.es}
\and Ra\'ul Guti\'errez\\
Dep. Sistemes Inform\`atics i Computaci\'o, Universitat Polit\`ecnica de Val\`encia \\ Camino de Vera S/N, 46022 Val\`encia, Spain\\
\email{rgutierrez@dsic.upv.es}
\and Enrique Martin-Martin\\
Dep. Sistemas Inform\'aticos y Computaci\'on, Universidad Complutense de Madrid \\ C/ Prof. Jos\'e Garc\'ia Santesmases 9, 28040 Madrid, Spain\\
\email{emartinm@ucm.es}}

\begin{document}

\label{firstpage}

\maketitle

%
%
%
%
%
%
%
%

\begin{abstract}

  In order to automatically infer the resource consumption of
  programs, analyzers track how \emph{data sizes} change along
  program's execution.  Typically, analyzers measure the sizes of data
  by applying \emph{norms} which are mappings from data to natural
  numbers that represent the sizes of the corresponding data. When
  norms are defined by taking type information into account, they are
  named \emph{typed-norms}.  This article presents a transformational
  approach to resource analysis with typed-norms that are inferred by a
  data-flow analysis. The analysis is based on a transformation of the
  program into an \emph{intermediate abstract program} in which each
  variable is abstracted with respect to all considered norms which
  are valid for its type.  We also present the data-flow analysis to
  automatically infer the required, useful, typed-norms from programs.
  Our analysis is formalized on a simple rule-based representation to
  which programs written in different programming paradigms (e.g.,
  functional, logic, imperative) can be automatically
  translated. Experimental results on standard benchmarks used by
  other type-based analyzers show that our approach is both efficient
  and accurate in practice.
  
  \emph{Under consideration in Theory and Practice of Logic Programming (TPLP).}

\end{abstract}

\begin{keywords}
resource analysis, typed-norms, data-flow analysis, program transformation


\end{keywords}



\section{Introduction}
\label{sec:intro}

Automated resource analysis \cite{DBLP:journals/cacm/Wegbreit75} needs
to infer how the sizes of data are modified along program's
execution. Size is measured using so-called norms
\cite{DBLP:conf/tapsoft/BossiCF91} which define how the size of a term is
computed. Examples of norms are \emph{list-length} which counts the
number of elements of a list, \emph{tree-depth} which counts the depth
of a tree, \emph{term-size} which counts the number of constructors,
etc.  Basically, in order to infer the resource consumption of
executing a loop that traverses a data-structure, the analyzer tries
to infer how the size of such data-structure decreases at each
iteration w.r.t.\ the chosen norm. Given a tree \lst{t}, using a
term-size norm, we infer that a loop like ``\textsf{while (t!=leaf)
  t=t.right;}'' performs at most \lst{nodes(t)} iterations, where
function \lst{nodes} returns the number of nodes in the tree. This is
because size analysis infers that  at each iteration the instruction \lst{t=t.right}
decreases \lst{nodes(t)}.  However, by using the tree-depth norm, we will
infer that \lst{depth(t)} is an upper bound on the number of
iterations. The latter is obviously more precise than the former
bound as \lst{depth(t)}$\leq$\lst{nodes(t)}.

The last two decades have witnessed a wealth of research on using
norms in termination analysis, especially in the context of logic
programming
\cite{DBLP:conf/tapsoft/BossiCF91,BCGGV:TOPLAS:2007,GenaimCGL02}. Early
work pointed out that the choice of norm affects the precision such
that the analyzer may only succeed to prove termination if a certain
norm is used, while it cannot prove it with others. Later on, there
has been further investigation on applying multiple norms, i.e., using
two or more norms by applying them simultaneously
\cite{DBLP:conf/tapsoft/BossiCF91}. This means that the same data in the original
program is replaced by two or more abstract data each one specifying
its size information w.r.t.\ the corresponding
norm.
%
%
Even a further step has been taken on using \emph{typed-norms} which
allow defining norms based on type information (namely on recursive
types)~\cite{BCGGV:TOPLAS:2007}.  Inferring norms from type
information makes sense as recursive types represent recursive
data-structures and thus, in termination analysis, they identify some
potential sources of infinite recursion and, in resource analysis,
they might influence the number of iterations that the loops
perform. Besides, typed-norms allow that the same term can be measured
differently depending on its type. As pointed out
in~\cite{GenaimCGL02}, this is particularly useful when the same
function symbol may occur in different type contexts.

In the context of resource analysis, we found early work that already
pointed out that the combination of norms affects the precision of
lower-bound time analysis \cite{lowb-time-andy-ilps97}. Sized-types
provide a way to consider more than one norm for each type. They have
been used in the context of functional
\cite{phd:pedro_vasconcelos,vh-03} and recently in logic
programming \cite{sized-types-iclp2013,SerranoLH14}.
In the former case, they are
inferred by a type analysis and in the latter via abstract
interpretation. 
In contrast, we propose a transformational approach which provides a
simple and accurate way to use multiple typed-norms in resource
analysis as follows: (1) we first transform the program into an
\emph{intermediate abstract program} in which each variable is
abstracted with respect to all considered norms valid for its type,
(2) such intermediate program is then transformed into upper and lower
resource bounds automatically. As regards the first phase, we
formalize the transformation assuming that the input programs are
given in a simple rule-based representation. The rule-based
representation contains program rules, pattern matching and assignment
using a compact syntax.
Programs written in \cbstart first-order \cbend functional or imperative programming languages can
be represented by means of this representation in a straightforward
way (since this representation can model control-flow graphs with
procedure calls). Logic programs can be represented as well by
replacing matching (and assignment) by unification, without any
further change in our analysis.
As regards the second phase, 
\cbstart note that we are interested in relying on
existing techniques and using them as a black-box without modifying
them. This is important since they receive abstract programs that come
from different sources, and we do not want to make any change that is
particular to our transformation that might break the functionality of
other parts. \cbend
Thus, formalizing our framework focuses only on the first
step.

While allowing multiple norms might lead to more accurate bounds than
adopting one norm, the efficiency of the analysis can be degraded
considerably. This is because the process of finding resource bounds
from abstractions that have more arguments (due to the use of multiple
norms) is more costly. Thus, an essential aspect for the practical
applicability of our method is to obtain the smallest sets for the
relevant typed-norms, i.e., eliminate those abstractions that will not
lead to further precision. For this purpose, we present a new
algorithm for the inference of typed-norms which, by inspecting the
program, can detect which norms are useful to later infer the resource
consumption, and discard norms that are useless for this purpose. Our
inference is formalized as a data-flow analysis which is applied as a
pre-process, such that once the relevant norms are inferred, the
transformation into the abstract program is carried out w.r.t.\ the
inferred norms.

\subsection{Summary of Contributions}

The main contributions of this article can be summarized as follows:
\begin{enumerate}
\item We introduce a transformation from the rule-based representation
  to an abstract representation in which each variable is abstracted
  with respect to all considered norms valid for its type, and prove
  soundness of the process.

\item We present to the best of our knowledge the first algorithm for
  the inference of typed-norms that are relevant to infer the resource
  consumption, and prove soundness of the type inference step.

\item We extend our approach to handle polymorphic types and context-sensitive norms.

\item We perform an experimental evaluation and compare the results
  with those obtained using other systems \cite{HAH12,SerranoLH14}.
\end{enumerate}
This article is an extended and revised version of a conference paper
that was published in the proceedings of LOPSTR
2013~\cite{AlbertGG13}. The main extensions w.r.t.\ the conference
paper affect all points above. As regards (1), we now provide a
semantics for the rule-based representation and for the abstract
representation and prove soundness of the transformation process,
while~\cite{AlbertGG13} did not have soundness results. \cbstart (2) The
formalization of the algorithm for the inference of typed-norms and its soundness are new
contributions of this article. In~\cite{AlbertGG13} the inference algorithm was informally presented without any theoretical result, but in this extended revised version we present a completely formalized data-flow algorithm for inferring typed-norms, prove its termination and also prove that the detected typed-norms cover those that may affect the program executions, i.e., the inference algorithm is correct. \cbend
(3) Also, in~\cite{AlbertGG13}, we had
considered only monomorphic types. (4) The experiments
of~\cite{AlbertGG13} have also been improved to deal with the same
benchmarks as in related work \cite{HAH12,SerranoLH14} and a
comparison with these systems has been included. We also have analyzed
an industrial case study to show the performance of our approach when
handling larger programs.

\subsection{Organization of the Article}

The article is organized as follows. In
Section~\ref{sec:intermediate-representation} we describe the syntax
and the semantics of the rule-based representation.
Section~\ref{sec:abstraction} presents our transformational approach
to resource analysis with typed-norms. We start by reviewing the
concept of typed-norm in Section~\ref{sec:prel-typed-norms}. It is
then extended to symbolic typed-norm and used to define the program
abstraction in Section~\ref{sec:our-transf-appr}. Soundness of the
transformation is proven in Section~\ref{sec:soundness}.
 Section~\ref{sec:inference} presents a typed-norms inference
 algorithm that is essential for the scalability of our approach. It
 infers the smallest sets for the typed-norms that are relevant for
 the inference of upper
 bounds. Section~\ref{sec:inference_formalization} formalizes the
 inference process and Section~\ref{sec:inference_soundness} proves
 its soundness. In Section~\ref{sec:polymorphic} we describe the
 extension of our approach to handle polymorphic
 types. Section~\ref{sec:experiments} contains our experimental
 evaluation, Section~\ref{related} compares our approach to related
 work, and Section~\ref{conclusions} concludes. \cbstart Finally, \ref{sec:proofs} contains the proofs of the theoretical results. \cbend




\section{A Rule-based Language}\label{sec:intermediate-representation}

To simplify the presentation, we formalize our approach
on a compact program syntax called \emph{rule-based representation}
(RBR) that contains program rules, pattern matching, and assignment.
%
%
It already 
incorporates \emph{static single
	assignment}~\cite{DBLP:journals/toplas/CytronFRWZ91} (each variable
is assigned exactly once).
\emph{Recursion} is the
only iterative mechanism and \emph{rule guards} are the only
conditional constructions in the RBR.
Although simple, the RBR syntax can represent programs from different programming languages by means of an intermediate translation.
For example, the RBR can be obtained from Java programs~\cite{AlbertAGPZ12}, from the functional part of
\emph{Abstract Behavioral Specification}~\cite{johnsen10fmco} (ABS) programs~\cite{AlbertACGGPR15}, and from the imperative part
of ABS~\cite{AlbertACGGPR15}.
This RBR can handle core-Prolog programs as well simply by
interpreting the pattern matching as unification. Interestingly this
does not require any further change in our size abstraction since our
abstract programs are actually constraint logic programs. However, we
note that analyzing abstract programs that originate from logic
programs for cost should be done by an analyzer that takes failure
into account~\cite{SerranoLH14}.

\subsection{Syntax of the rule-based language}

In order to present typed-norms and its impact on termination and resource analyses in a clear way, for now we will consider only \emph{monomorphic types}, although in Section~\ref{sec:polymorphic} we will present the extension to \emph{polymorphic types}.

\begin{definition}[Monomorphic types]\label{def:monomorphicTypes} A \emph{monomorphic type} $T$ can be a built-in data type as \tint{} or an algebraic data type $\DT$ defined as:
	\begin{flushleft}
		$\begin{array}{l@{~~~~~}lll}
		& \mi{Dd} & ::= & \mi{data}~\DT = \cons ~[\overline{\mid \cons{}}]\\
		& \cons   & ::= & \dcon[(\overline{T})] \\
		\end{array}$
	\end{flushleft}
where $\DC$ represents a data constructor and the notation $[\overline{X}]$ represents an \emph{optional} sequence of elements $X$. For simplicity, we assume that recursive types are in \emph{direct recursive form}, otherwise, we could consider mutually recursive types to be the same type.
\end{definition}

\begin{example}[List of integer numbers]~\label{ex:monomorphicTypes}
	Using the syntax presented in Def.~\ref{def:monomorphicTypes} we can define the data type of integer lists (\TYPE{IntList}) as follows:
	\begin{flushleft}
      ~~~~\lstinline@data IntList = Nil | Cons($\tint$, IntList)@
	\end{flushleft}
  In this case, the type of the nullary data constructor \code{Nil} is \TYPE{IntList}, and the type of the binary data constructor \code{Cons} is \tint{} $\times$ \TYPE{IntList} $\to$ \TYPE{IntList}.
\end{example}


We define programs in \emph{rule-based representation} (RBR programs in the sequel) as a set of data declarations followed by
typed procedures:

\begin{definition}[RBR syntax]\label{def:syntax}
	\[
	\begin{array}{lll}
	\rbrprog & ::= & [\overline{\mi{Dd}}] ~ \overline{\typedproc} \\
	\typedproc & ::= & p :: T_1 \times \cdots \times T_k ~ \overline{r}\\
	r & ::= &  p(\bar{x},\bar{y}) \rulearrow g, b_1, \ldots , b_n\\
	b & ::= & \rrassigns{x}{t} \mid\ p(\bar{x},\bar{y})\\
	g & ::= & {\it true}  \mid  g \wedge g  \mid e > e \mid e = e \mid e \ge e \mid \match(x,p)  \mid  \notmatch(x,p) \\
	p & ::= & \DC(\bar{x}) \\
	t & ::= & \expr \mid \DC(\bar{t}) \\
	\expr & ::= & x \mid  n \mid \expr {+} \expr \mid \expr {-} \expr \\
	\end{array}\]
	
RBR programs \rbrprog{} are formed by an optional set of data declarations ($\mi{Dd}$) followed by a set of typed procedures ($\typedproc{}$).
A \emph{typed procedure} begins with a type declaration $p :: T_1 \times \cdots \times T_n \times T_{n+1} \times \cdots \times
T_{n + m}$ stating the types $T_1, \ldots, T_n$ of its $n$ input arguments $\bar{x}$ ($n \ge 0$) and the types $T_{n+1}, \ldots T_{n + m}$ of its $m$ output arguments $\bar{y}$ ($m \ge 0$). After the type declarations there is a set of \emph{guarded rules} ($r$), where $p(\bar{x},\bar{y})$ is the
\emph{head} of the rule, the guard $g$ specifies the conditions for the rule to
be applicable and $b_1,\dots,b_n$ are the statements in the rule's \emph{body}.
For clarity, we sometimes enclose input and output arguments with angles
``$\langle$'' and ``$\rangle$'', i.e., $p :: \langle T_1 \times \cdots \times
T_n\rangle \times \langle T_{n+1} \times \cdots \times T_{n + m}\rangle$
and $p(\langle\bar{x}\rangle,\langle\bar{y}\rangle)$.
If a program \rbrprog{} has $n$ rules, we say that
$\nrules{\rbrprog{}} = n$ and $\getrule{i}{\rbrprog{}}$ represents the
i-th rule of $\rbrprog{}$.
%
Guards $\match(x,p)$ and $\notmatch(x,p)$, where $x\not\in {\it vars}(p)$ and $x$ and $p$ are of the same type, check if the value stored in variable $x$ matches with pattern $p$. Patterns ($p$) are data constructors $\DC$ applied to variables.
%
Terms ($t$) can be expressions (variables, integer numbers or arithmetic operations over expressions) 
or data constructors $\DC$ applied to properly typed terms
(e.g., \lst{Cons}$(6,y)$, where $6$ has type \tint{} and $y$ has type \TYPE{IntList}).
Terms not containing variables are called \emph{closed terms} (a.k.a. \emph{ground} terms).
\cbstart We assume that RBR programs are well-typed, i.e., every term and subterm in the program (including variables) have a type assigned that is coherent with procedure type declarations and data constructor types, considering a standard monomorphic type system~\cite{Pierce:types}\cbend.
%
%

%
%

\end{definition}


\begin{figure}[tbp]
%
%
%
\begin{lstlisting}[frame=none, multicols=2]
data IntList = Nil | Cons($\tint$, IntList) (*\label{abs:intlist}*)

fact :: $\langle\tint\rangle \times \langle\tint\rangle$
(*\hspace{-0.85cm}$\circled{1}\hspace{0.45cm}$*)fact((*$\langle$*)n(*$\rangle$*), (*$\langle$*)prod'(*$\rangle$*)) $\rulearrow$ true,
  prod := 1,
  while_0((*$\langle$*)n, prod(*$\rangle$*), (*$\langle$*)prod'(*$\rangle$*)) (*\label{rbr:fact_while0_invocation}*)

while_0 :: $\langle\tint \times \tint\rangle \times \langle\tint\rangle$
(*\hspace{-0.85cm}$\circled{2}\hspace{0.45cm}$*)while_0((*$\langle$*)n, prod(*$\rangle$*), (*$\langle$*)prod'(*$\rangle$*)) $\rulearrow$ 0 >= n (*\label{rbr:fact_while_1}*)
  prod' := prod
(*\hspace{-0.85cm}$\circled{3}\hspace{0.45cm}$*)while_0((*$\langle$*)n, prod(*$\rangle$*), (*$\langle$*)prod'(*$\rangle$*)) $\rulearrow$ 0 < n,  (*\label{rbr:fact_while_2b}*)
  prod1 := prod * n,
  n1 := n - 1,
  while_0((*$\langle$*)n1, prod1(*$\rangle$*), (*$\langle$*)prod'(*$\rangle$*)) (*\label{rbr:fact_while_2e}*)

factSum :: $\langle\TYPE{IntList}\rangle \times \langle\tint\rangle$
(*\hspace{-0.85cm}$\circled{4}\hspace{0.45cm}$*)factSum((*$\langle$*)l(*$\rangle$*), (*$\langle$*)sum'(*$\rangle$*)) $\rulearrow$ true, (*\label{fig:runnning:factSum0}*)
  sum := 0, (*\label{fig:runnning:factSum1}*)
  while_1((*$\langle$*)l, sum(*$\rangle$*), (*$\langle$*)sum'(*$\rangle$*))(*\label{fig:runnning:factSum2}*)

while_1 :: $\langle\TYPE{IntList} \times \tint\rangle \times \langle\tint\rangle$
(*\hspace{-0.85cm}$\circled{5}\hspace{0.45cm}$*)while_1((*$\langle$*)l, sum(*$\rangle$*), (*$\langle$*)sum'(*$\rangle$*)) $\rulearrow$ (*\label{fig:runnning:while1_0}*)
  match(l, Nil),(*\label{fig:runnning:while1_1}*)
  sum' := sum(*\label{fig:runnning:while1_2}*)
(*\hspace{-0.85cm}$\circled{6}\hspace{0.45cm}$*)while_1((*$\langle$*)l, sum(*$\rangle$*), (*$\langle$*)sum'(*$\rangle$*)) $\rulearrow$ (*\label{fig:runnning:while1_3}*)
  match(l, Cons(e,l1)), (*\label{rbr:while1_match}*)(*\label{fig:runnning:while1_4}*)
  fact((*$\langle$*)e(*$\rangle$*), (*$\langle$*)prod(*$\rangle$*)), (*\label{rbr:while1_fact_invocation}*)(*\label{fig:runnning:while1_5}*)
  sum1 := sum + prod,(*\label{fig:runnning:while1_6}*)
  while_1((*$\langle$*)l1, sum1(*$\rangle$*), (*$\langle$*)sum'(*$\rangle$*))(*\label{fig:runnning:while1_7}*)

main :: $\langle\rangle \times \langle\tint\rangle$
(*\hspace{-0.85cm}$\circled{7}\hspace{0.45cm}$*)main((*$\langle$*)(*$\rangle$*) $\times$ (*$\langle$*)r'(*$\rangle$*)) $\rulearrow$ true,(*\label{fig:runnning:main_0}*)
  l := Cons(2, Nil),(*\label{fig:runnning:main_1}*)
  factSum((*$\langle$*)l(*$\rangle$*), (*$\langle$*)r'(*$\rangle$*))(*\label{fig:runnning:main_2}*)
\end{lstlisting}

\caption{RBR program.\label{fig:running}}
\end{figure}

\begin{example}[RBR program]\label{ex:rbr}

Figure~\ref{fig:running} contains the RBR of a program with three principal procedures: \code{fact} computes the factorial of an integer number, \code{factSum} traverses a list of integer numbers and adds the factorial value for each element, and \code{main} is the entry point of the program that invokes \code{factSum} with the unitary list \code{Cons(2,Nil)}.
%
Both \code{fact} and \code{factSum} have only one rule each, which establishes the value of the accumulator and invokes \code{while_n}.
Each loop is a procedure with two rules: one for finishing the loop and other for computing one iteration. For example, the first rule of \code{while_0} has the guard \code{0 >= n} (Line~\ref{rbr:fact_while_1}) to check that the loop has finished, therefore returning the input accumulator as output value. On the other hand, the second rule of \code{while_0} (lines~\ref{rbr:fact_while_2b}--\ref{rbr:fact_while_2e}) contains the guard \code{0 < n} that checks that the loop has not finished yet. In that case, updated values of \code{prod} and \code{n} are stored in \code{prod1} and \code{n1} (recall the use of \emph{static single assignment}) and those values are used in the recursive call. Finally, \code{main} is a procedure with one rule that simply invokes \code{factSum}.
\end{example}

\subsection{Semantics of the rule-based language}


\begin{center}
\begin{figure}[tbp]
\begin{center}
\(
\small
\begin{array}{|lc|}
  \hline 
  (1) & \semrule
  {\bc \equiv \rrassigns{x}{t}~~ \fevalt(t,\tv)=v}
  {
    \tuple{p,\bc{\cdot}\stkbc,\tv}{\cdot} \ar 
    \rrderiv
    \tuple{p,\stkbc,\tv[x\mapsto v]}{\cdot} \ar
  }\\[0.5cm]
  (2) & \semrule
  {
   \bc \equiv m(\bar{x},\bar{y})~~~~ m(\bar{x'},\bar{y'}) \rulearrow g, \bc_1 \cdot \cdot \cdot \bc_k \in \rbrprog{} ~\mi{fresh}\\
   \tv_1 \equiv [\overline{x' \mapsto \tv(x)}] ~~ \fevalg(g,\tv_1) = \tv_2  	
  	}
  {
  	\tuple{p,\bc{\cdot}\stkbc,\tv}{\cdot} \ar 
  	\rrderiv
  	\tuple{m,\bc_1 \cdot \cdot \cdot \bc_k,\tv_1 \uplus \tv_2} {\cdot} \tuple{p[\overline{y' \sim y}],\stkbc,\tv}{\cdot} \ar
  }\\[0.5cm]
  (3) & \semrule
  {}
  {
    \tuple{m,\epsilon,\tv_1}{\cdot}\tuple{p[\overline{y' \sim y}],\stkbc,\tv}{\cdot} \ar 
     \rrderiv
     \tuple{p,\stkbc,\tv[\overline{y \mapsto \tv_1(y')}]}{\cdot} \ar
  }\\[0.3cm]
  \hline
\end{array}
\)
\end{center}
\caption{Operational semantics of rule-based programs}\label{fig:rrsem}
\end{figure}
\end{center}

The rule-based language is evaluated using an operational semantics based on \emph{variable mappings} and \emph{configurations}, which are defined as:

\begin{definition}[Variable mappings]\label{def:variableMapping}
	A \emph{variable mapping} $\tv \in \TV$ is a mapping $[\overline{x \mapsto v}]$ that associates values $\overline{v}$ (namely integer numbers or ground constructed terms) to variables $\overline{x}$. We use the symbol $\emptymap{}$ for empty mappings and $\tv_1 \uplus \tv_2$ for the union of variable mappings with disjoint domain. The notation $\tv[\overline{x \mapsto v}]$ represents the extension of $\tv$ with the new mappings $[\overline{x \mapsto v}]$ (this operation redefines the previous mappings for variables $\overline{x}$ if they appear in the domain of $\tv$). The application $\tv(x)$ returns the value $v$ associated to variable $x$, and $\tv(t)$ returns the term resulting of replacing every variable in $t$ by its value in $\tv$ (similarly for patterns $p$).
\end{definition}

\begin{example}[Variable mappings]\label{ex:variableMapping}
	Consider two variable mappings $\tv_1 \equiv [x \mapsto 3, y \mapsto \TYPE{Nil}]$ and $\tv_2 \equiv [z \mapsto \TYPE{Cons}(3,\TYPE{Nil})]$. Since $\tv_1$ and $\tv_2$ have disjoint domains its union $\tv_1 \uplus \tv_2$ is defined, with result $[x \mapsto 3, y \mapsto \TYPE{Nil}, z \mapsto \TYPE{Cons}(3,\TYPE{Nil})]$. On the other hand, $\tv_1[x \mapsto  \TYPE{Nil}, z \mapsto 0] = [y \mapsto \TYPE{Nil},x \mapsto  \TYPE{Nil}, z \mapsto 0]$  because the mapping $x\mapsto3$ has been redefined in the extension of $\tv_1$.
\end{example}

\begin{definition}[Configurations]\label{def:configurations}
	A \emph{configuration} \cbstart(or \emph{call stack}\cbend), denoted as $\ar$, is a sequence of \emph{activation records}. An activation record is a triple of the form $\tuple{p,\bc{\cdot}\stkbc,\tv}$ where $p$ is a procedure name,\footnote{While
	procedure names are not really needed in the operational semantics, they increase clarity and
	simplify the proofs in~\ref{sec:proofs}.} $\bc$ is the next statement to execute, $\stkbc$ is the sequence of  statements after $\bc$, and
	$\tv$ is the variable mapping that stores the values of the procedure parameters and local variables. Elements in sequences are separated by dots, where the last element of the sequence can represent the rest of the sequence (for example $\bc_1 \cdot \bc_2 \cdot \stkbc$ or $\tuple{p,\bc{\cdot}\stkbc,\tv} \cdot \ar$),
	and we overload the symbol $\emptyseq$ to denote empty sequences.
\end{definition}

Figure~\ref{fig:rrsem} shows the rules of the operational semantics ($\rrderiv$) that evaluates RBR programs. We consider a function $\fevalt(t,\tv)$ that evaluates a term $t$ under a variable mapping $\tv$, returning a value $v$, and a function $\fevalg(g,\tv)$ that checks if a guard $g$ is satisfied under a variable mapping $\tv$, returning a new (possibly empty) mapping that performs pattern matching.\footnote{The straightforward definition of these functions can be found in~\ref{sec:proofs}.}
Figure~\ref{fig:rrsem} contains 3 rules. Rule $(1)$ evaluates activation records where the next statement is an assignment. In that case the value obtained from the right-hand side is introduced in the variable mapping of the activation record. Rule $(2)$ evaluates an activation record where the next statement is a procedure call $m(\bar{x},\bar{y})$. The first step is to obtain a \emph{fresh} version of a rule of $m$ with all its variables renamed to avoid any collision. Then a variable mapping $\tv_1$ is created for parameter passing and this mapping is used to evaluate the rule guard $g$. If the guard is evaluated to $\mi{true}$, a new activation record with the rule body and the variable mapping\footnote{\emph{Static single assignment} guarantees that $x$ does not appears in $\tv_1$, so we could use mapping union instead of extension.} $\tv_1 \uplus	\tv_2$ (where $\tv_2$ is generated during the evaluation of the guard) is inserted in front of the configuration. Note that the activation record of the caller stores the relation $\overline{y' \sim y}$ between the output values of the rule and the output parameters of the call. Finally, rule $(3)$ handles empty activation records, which are removed after the output values of the callee are stored in the variable mapping of the caller.

When needed, we decorate steps with two values $\rrderiv^{a \cdot b}$: the semantic rule $a$ applied---(1), (2) or (3)---and the program rule number $b$ used, considering the whole program. If the step is not a procedure call i.e., it uses semantic rules (1) or (3), the program rule number is set to $\epsilon$. Examples of these decorations can be seen in the next example:

\begin{example}[Evaluation of an RBR program]\label{ex:rbr_evaluation}
	The evaluation of the \code{main} procedure in Figure~\ref{fig:running} proceeds as shown below. For simplicity, when obtaining fresh names for variables---rule (2) in Figure~\ref{fig:rrsem}---we simply use subscripts ($_n$) with the same number as the current configuration $\ar_n$, we have underlined the statement that controls each step, and 
	we write $\rrderiv^*$ for several $\rrderiv$-steps. Note that the program rule of \TYPE{factSum} is the $4^{th}$ rule in the program, and the second rule of \TYPE{while\_1} is the $6^{th}$ rule in the program.
    \[
    \begin{array}{l@{}l}
    \ar_0 \equiv & \tuple{\TYPE{main},\underline{\TYPE{l$_0$:=Cons(2, Nil)}} \cdot \TYPE{factSum($\langle$l$_0$$\rangle$,$\langle$r'$_0$$\rangle$)},\emptymap} \rrderiv^{(1)\cdot \epsilon} \\

    \ar_1 \equiv & \tuple{\TYPE{main},\underline{\TYPE{factSum($\langle$l$_0$$\rangle$,$\langle$r'$_0$$\rangle$)}},[\TYPE{l$_0$} \mapsto \TYPE{Cons(2, Nil)}]} \rrderiv^{(2)\cdot 4} \\

    \ar_2 \equiv & \tuple{\TYPE{factSum},\underline{\TYPE{sum$_2$:=0}} \cdot \TYPE{while\_1($\langle$l$_2$, sum$_2$$\rangle$, $\langle$sum'$_2$$\rangle$)},[\TYPE{l$_2$} \mapsto \TYPE{Cons(2, Nil)}]} \cdot \\
                 & \tuple{\TYPE{main}[\TYPE{sum'$_2$} \sim \TYPE{r'$_0$}],\emptyseq,[\TYPE{l$_0$} \mapsto \TYPE{Cons(2, Nil)}]} \rrderiv^{(1)\cdot \epsilon} \\

	\ar_3 \equiv & \tuple{\TYPE{factSum},\underline{\TYPE{while\_1($\langle$l$_2$, sum$_2$$\rangle$, $\langle$sum'$_2$$\rangle$)}},[\TYPE{l$_2$} \mapsto \TYPE{Cons(2, Nil)},\TYPE{sum$_2$} \mapsto 0]} \cdot \\
	& \tuple{\TYPE{main}[\TYPE{sum'$_2$} \sim \TYPE{r'$_0$}],\emptyseq,[\TYPE{l$_0$} \mapsto \TYPE{Cons(2, Nil)}]} \rrderiv^{(2)\cdot 6} \\

	\ar_4 \equiv & \tuple{\TYPE{while\_1}, \underline{\TYPE{fact($\langle$e$_4$$\rangle$,$\langle$prod$_4$$\rangle$)}} \cdots \TYPE{while\_1($\langle$l1$_4$,sum$_4$$\rangle$,$\langle$sum'$_4$$\rangle$)},[\TYPE{l$_4$} \mapsto \TYPE{Cons(2, Nil)},\\
		         & ~~~\TYPE{sum$_4$} \mapsto 0, \TYPE{e$_4$} \mapsto 2, \TYPE{l1$_4$} \mapsto \TYPE{Nil}]} \cdot \\
	             & \tuple{\TYPE{factSum}[\TYPE{sum'$_4$} \sim \TYPE{sum'$_2$}],\emptyseq,[\TYPE{l$_2$} \mapsto \TYPE{Cons(2, Nil)},\TYPE{sum$_2$} \mapsto 0]} \cdot \\
                 & \tuple{\TYPE{main}[\TYPE{sum'$_2$} \sim \TYPE{r'$_0$}],\emptyseq,[\TYPE{l$_0$} \mapsto \TYPE{Cons(2, Nil)}]} \rrderiv^* \\

	\ar_5 \equiv & \tuple{\TYPE{while\_1}, \underline{\emptyseq},[\TYPE{l$_4$} \mapsto \TYPE{Cons(2, Nil)},\TYPE{sum$_4$} \mapsto 0,\TYPE{e$_4$} \mapsto 2, \TYPE{l1$_4$} \mapsto \TYPE{Nil},\TYPE{sum'$_4$} \mapsto 2]} \cdot \\
	& \tuple{\TYPE{factSum}[\TYPE{sum'$_4$} \sim \TYPE{sum'$_2$}],\emptyseq,[\TYPE{l$_2$} \mapsto \TYPE{Cons(2, Nil)},\TYPE{sum$_2$} \mapsto 0]} \cdot \\
	& \tuple{\TYPE{main}[\TYPE{sum'$_2$} \sim \TYPE{r'$_0$}],\emptyseq,[\TYPE{l$_0$} \mapsto \TYPE{Cons(2, Nil)}]} \rrderiv^{(3)\cdot \epsilon} \\

	\ar_6 \equiv & \tuple{\TYPE{factSum},\underline{\emptyseq},[\TYPE{l$_2$} \mapsto \TYPE{Cons(2, Nil)},\TYPE{sum$_2$} \mapsto 0,\TYPE{sum'$_2$} \mapsto 2]} \cdot \\
	& \tuple{\TYPE{main}[\TYPE{sum'$_2$} \sim \TYPE{r'$_0$}],\emptyseq,[\TYPE{l$_0$} \mapsto \TYPE{Cons(2, Nil)}]} \rrderiv^{(3)\cdot \epsilon} \\

	\ar_7 \equiv & \tuple{\TYPE{main},\emptyseq,[\TYPE{l$_0$} \mapsto \TYPE{Cons(2, Nil)}, \TYPE{r'$_0$} \mapsto 2]} \\
    \end{array}
    \]
\end{example}

\begin{definition}[Traces and sequences of steps]
  A \emph{trace} $\trace = \ar_0 \rrderiv^{r_1} \ar_1 \rrderiv^{r_2} \ldots \rrderiv^{r_n} \ar_n$ is a sequence of $\rrderiv$-steps from an initial configuration $\ar_0$. Given a trace $\trace$, its \emph{steps} are defined as $\trsteps{\trace} = \tuple{r_1, r_2, \ldots, r_n}$, i.e., the sequence of step decorations of the $\rrderiv$-steps in \trace. 
  Finally, the set of \emph{trace steps} combines the steps of all the possible traces starting from a given configuration. Formally, $\traces(\ar_0) =  \{ \trsteps{\ar_0 \rrderiv^{r_1} \ar_1 \rrderiv^{r_2} \ldots \rrderiv^{r_n} \ar_n} \mid n \in \mathbb{N}\}$. 
  \cbstart These notions will play an important role when defining the soundness of the abstraction (Section~\ref{sec:soundness}) and the soundness of the typed-norms inference in Section~\ref{sec:inference_soundness}. \cbend
\end{definition}


\cbstart 
Let us informally discuss how the operational semantics is typically
instrumented to account for cost of traces.
We first assume that our language includes a special instruction
\lst{tick($n$)}, where $n$ is a number, which is used to simulate a
consumption of $n$ resources. Note that $n$ can be negative to
simulate release of resources as well. In practice, we also allow the
use any arithmetic expression instead of $n$, but for simplicity we
assume it is a constant.
Next, we instrument activation records with a \emph{global} resource
consumption counter called $\texttt{cost}$, and add a corresponding
semantic rule for \lst{tick($n$)} that simply sets $\texttt{cost}$ to
$\texttt{cost}-n$.
A trace $\trace$ is said to be valid if $\texttt{cost}$ never takes
negative values, i.e., $\texttt{cost}$ is initialized with an amount
of resources that is enough to carry out this particular execution.
The cost of a trace is the minimal initial value for \texttt{cost}
that makes it valid.
Finally, a function $f$ that maps initial configurations to
nonnegative values is called an upper bound on the worst-case cost if,
for any initial configuration $\ar_0$, setting the initial value of
$\texttt{cost}$ to $f(\ar_0)$ guarantees generating valid traces
only.
It is called a lower bound on the best-case cost if, for any initial
configuration $\ar_0$, setting the initial value of $\texttt{cost}$ to
a value smaller than $f(\ar_0)$ generates an invalid trace.
\cbend




\section{Size Abstraction Using Typed-Norms}\label{sec:abstraction}

The resource analysis framework we rely
on~\cite{AlbertAGPZ07,AlbertACGGPR15} performs two phases: (1) the
program is transformed into an abstract representation where data are
replaced by their sizes, and (2) such intermediate program is then
analyzed to obtain upper/lower bounds on the resource consumption. In
this section we will present how to abstract the size of program data
using typed-norms and we will show that this abstraction is sound
wrt.\ the original semantics of RBR programs. The second phase,
performed by cost relation solvers like PUBS~\cite{AlbertGM12},
CoFloCo~\cite{DBLP:conf/aplas/Flores-MontoyaH14,Flores-Montoya16} or
the solver in CiaoPP~\cite{SerranoLH14}, is independent of the technique
applied to abstract data sizes and therefore will not be covered in
this paper.

\subsection{Preliminaries on Typed-Norms}\label{sec:prel-typed-norms}

In order to obtain the abstract representation of a program, we first replace
data with numbers representing their sizes and then transform each instruction
to linear constraints that reflect how these sizes change.
The mapping from data structures to sizes is done by means of size functions
(usually called \emph{norms}). The most well-known norm used in the literature
is \emph{term-size}~\cite{DBLP:conf/tapsoft/BossiCF91,BCGGV:TOPLAS:2007}, which
counts the number of constructors in a given data structure:

\begin{definition}[Term-size norm]\label{def:term-size} The size of a term $t$ using the term-size norm is defined as:
$\label{eq:normscheme:ts}
\tsnorm{\DC(t_1,\ldots, t_n)} = 1 + \sum_{i=1}^n \tsnorm{t_i}$.
\end{definition}
Notice that term-size is  defined for terms containing only data
constructors and cannot handle integer numbers.

\begin{example}[Term-size norm]\label{ex:term-size}
Consider the following program that uses binary trees (\TYPE{BT}), list of
binary trees (\TYPE{BTL}), and lists of lists of binary trees (\TYPE{BTLL}).
{\normalfont
\begin{lstlisting}[frame=none, multicols=2]
data BT = E | B(BT, BT)
data BTL = N | C(BT, BTL)
data BTLL = NL | CL(BTL,BTLL)

trees :: $\langle\TYPE{BTL}\rangle \times \langle\TYPE{Int}\rangle$
trees($\langle$l$\rangle$,$\langle$n'$\rangle$) $\rulearrow$  match(l, N),
  n' := 0
trees($\langle$l$\rangle$,$\langle$n'$\rangle$) $\rulearrow$ match(l, C(h,t)),
   trees($\langle$t$\rangle$, $\langle$nt$\rangle$),
   n' := 1 + nt

sumtrees :: $\langle\TYPE{BTLL}\rangle \times \langle\TYPE{Int}\rangle$
sumtrees($\langle$l$\rangle$, $\langle$n'$\rangle$) $\rulearrow$  match(l, NL),
  n' := 0
sumtrees($\langle$l$\rangle$,$\langle$n'$\rangle$) $\rulearrow$ match(l, CL(h,t)),
   trees($\langle$h$\rangle$, $\langle$nh$\rangle$),
   sumtrees($\langle$t$\rangle$, $\langle$nt$\rangle$),
   n' := nh + nt
\end{lstlisting}
}
The \code{trees} function counts the number of binary trees in a list of type
\TYPE{BTL}, and \code{sumtrees} counts the number of binary trees
in all the lists inside a list of type \TYPE{BTLL}. Using the term-size norm,
an empty list \code{N} or \code{NL} has size $1$ (one data constructor),
whereas the list
\code{C(E,C(E,N))} has size $5$ (three
list constructors plus two empty binary tree constructors).
The function
\code{sumtrees} traverses every element of the list for computing its
number of trees. Using the term-size norm, a static analysis obtains a
complexity of $O(n^2)$ \cbstart for the function \code{sumtrees}\cbend, where $n$ is the total size of the original list (i.e.,
the number of list, list of lists, and binary tree constructors). The size of
each inner \TYPE{BTL} list is bounded by $n$, so each call to \code{trees} will
contribute
$O(n)$ to the overall complexity. This is an example where term-size is not
very precise, as it does not keep separate the information about the length of
the \TYPE{BTLL} list ($l$), the length of the inner \TYPE{BTL} lists ($s$), and
the size of the binary trees ($b$) to obtain a more accurate complexity of $O(l
\times s)$. Note that the size of the binary trees does not play any role in
the actual complexity of \code{sumtrees} because the binary trees are not
traversed.
\end{example}

In order to overcome the mentioned imprecision we use typed-norms, which distinguish data constructors according to their types. Using this kind of norms we can measure the length of a list and the size of its elements separately, similarly to what has been done in the context of termination analysis of logic programs~\cite{BCGGV:TOPLAS:2007}. Before introducing typed-norms, we present some notation about types that will be used throughout the paper:

\begin{definition}\label{def:types}
	Given two types $T_1$ and $T_2$, we say that $T_2$ depends on $T_1$, written $T_1 \preceq T_2$, if the definition of type
	$T_2$ uses (either directly or transitively) type $T_1$. Relying on this notion, we define the \emph{set of constituent types of T} as $\deptypes{T} = \{ T' \mid T' \preceq T\} \cup \{T\}$, i.e., all the types involved in the definition of $T$ including $T$ itself. If $T \preceq T$ we say that $T$ is a \emph{recursive type}.
	%
	\cbstart We use $\typer{t}{i}$ to refer to the type of a term $t$ in the i-th rule of the program, or simply $\type{t}$ if the rule is clear from the context (note that the same variable can have different types in different rules). As the program is well-typed then every term $t$ has a monomorphic type assigned, so $\typer{t}{i}$ simply returns that type.\cbend
	%
\end{definition}
	By definition of \tttype{IntList}---Line~\ref{abs:intlist}) in Figure~\ref{fig:running}---we have that $\inttype \preceq \tttype{IntList}$ and $\tttype{IntList} \preceq \tttype{IntList}$, therefore \tttype{IntList} is a recursive type and $\deptypes{\tttype{IntList}} = \{\inttype, \tttype{IntList}\}$.
\begin{definition}[Typed-norms for closed terms]\label{def:typed-norms}
We consider two typed-norms for computing the size of a closed term $t$
regarding a type $T$: \typednormorig{t}{T} and \typednorm{t}{T}. We will not
contemplate integer values for the first norm but non-negative integer numbers
(denoted as \nattype) as we explain later.
\[
\typednormorig{t}{T} =
\begin{cases}
t & \mbox{ if } T=\nattype \mbox{ and } \type{t} = \nattype\\
0 & \mbox{ if } T \neq \nattype \mbox{ and } \type{t} = \nattype\\
%
%
1 + \sum_{i=1}^n \typednormorig{t_i}{T} &
\mbox{ if } t = \DC(t_1,\ldots, t_n) \mbox{ and } \type{t} = T\\
\sum_{i=1}^n \typednormorig{t_i}{T} &
\mbox{ if } t = \DC(t_1,\ldots, t_n) \mbox{ and } \type{t} \neq T\\
\end{cases}
\end{equation*}

\begin{equation*}
\typednorm{t}{T} =
\begin{cases}
t & \mbox{ if } T=\inttype \mbox{ and } \type{t} = \inttype\\
0 & \mbox{ if } T \neq \inttype \mbox{ and } \type{t} = \inttype\\
1 + \sum_{i=1}^n \typednorm{t_i}{T} &
  \mbox{ if } t = \DC(t_1,\ldots, t_n) \mbox{ and } \type{t} = T\\
\mi{max}_{i=1}^n \typednorm{t_i}{T} &
  \mbox{ if } t = \DC(t_1,\ldots, t_n) \mbox{ and } \type{t} \neq T\\
\end{cases}
\]
%
\end{definition}
Note that, in the last case of the definition of $\typednorm{t}{T}$, the
$\mi{max}$ of
an empty set is $-\infty$ if $T$ is \inttype, and $0$ otherwise. In
principle, this depends on the domain of the elements whose
$\mi{max}$ we are taking.
The intuition behind these typed-norms is to count the number of data constructors of type $T$ that appear in term $t$. For integer numbers 
they simply return their value, 
whereas for constructed terms they check whether the term has type $T$ or not
in order to count the constructor in the head. The difference is how to handle
those nested subterms of type $T$ that occur in a bigger term of type different
from $T$: \typednormorig{t}{T} sums the sizes of those subterms, whereas
$\typednorm{t}{T}$ just keeps the size of the maximal subterm.
The reason to constrain \typednormorig{t}{T} to \nattype{} instead of
\inttype{} is
that adding the integers inside a data structure provides a size that is not
sound when negative values are involved. For example, $\typednormorig{\TYPE{Cons(3,
    Cons(-3, Nil))}}{\inttype}$ would be $0$.
On the other hand,
$\typednorm{\TYPE{Cons(3,
Cons(-3, Nil))}}{\inttype} = 3$, so we know that any integer inside the list is
smaller than or equal to $3$.
We claim that $\typednorm{t}{T}$ suits better in the static analysis framework
we consider, as we explain in Example~\ref{ex:typed-norms}.

In our rule-based language we only consider the basic type of integer numbers, so the rest of values are constructed using data constructors $\DC$. However, typed-norms could be easily extended to support more basic types by providing a suitable case in the definition of \typednormorig{\cdot}{T} and \typednorm{\cdot}{T}. For example, if \strtype{} were a basic type, we could define the size of strings as their length: $\typednormorig{\cdot}{T} = \mathit{length}(t) \mbox{ if } T=\strtype \mbox{ and } \type{t} = \strtype$ (similarly for \typednorm{\cdot}{T}).
\begin{example}[Typed-norms for closed terms]\label{ex:typed-norms}
%
%
%
\cbstart
Consider the program involving lists of binary trees shown in
Example~\ref{ex:term-size}. The program does not contain integer values but only constructed data, so in this case both typed-norms $\typednorm{\cdot}{T}$ and $\typednormorig{\cdot}{T}$ are sound. \cbend
Regarding the
list types \TYPE{BTL} and \TYPE{BTLL}, both typed-norms obtain the same values:
$\typednormorig{\tttype{N}}{\tttype{BTL}} =
\typednorm{\tttype{N}}{\tttype{BTL}} = 1$ and
$\typednormorig{\tttype{NL}}{\tttype{BTLL}} =
\typednorm{\tttype{NL}}{\tttype{BTLL}} = 1$ because they contain exactly one
list constructor.
They also coincide in the list of type \code{BTLL} with two
elements $t \equiv$
\code{CL( C(E,C(E,N)), CL(N,NL) )} regarding the type
\TYPE{BTLL}---$\typednormorig{t}{\tttype{BTLL}} =
\typednorm{t}{\tttype{BTLL}} = 3$---as both count the number of list
constructors: two constructors \code{CL} and one \code{NL}. However, the two
norms differ wrt.\ type \TYPE{BTL}:
$\typednormorig{t}{\TYPE{BTL}} = 4$ because it is the sum of all the \code{BTL}
constructors in the list, 3 in the first element and one in the second element,
but $\typednorm{\cdot}{\TYPE{BTL}} = 3$ as it is the
maximum number of \code{BTL} constructors that appear in some element of the
list. This difference has an impact on the
concrete upper bounds
obtained by static analysis. For example, the function \code{sumtrees} defined
in Example~\ref{ex:term-size} has an asymptotic complexity of $O(l \times s)$
in both cases, where $l$ is the size of the list of lists \TYPE{BTLL} and $s$
is the size of its inner lists \TYPE{BTL}. The size of any list $l$ is the same
using both typed-norms
($\typednormorig{l}{\tttype{BTLL}} = \typednorm{l}{\tttype{BTLL}}$),
however, the size of its elements of type \TYPE{BTL} will differ
(\typednormorig{l}{\tttype{BTL}}
$\geq$ \typednorm{l}{\tttype{BTL}}). Since the static analysis framework we are
considering is compositional and assumes worst-case scenarios for each
iteration, the upper bound obtained using \typednorm{\cdot}{T}
\cbstart will in general be tighter.\cbend
\end{example}

As a final comment, note that we could also define a typed-norm analogous to \typednorm{\cdot}{T} that
estimates the minimum value, by replacing ``$\max$'' with ``$\min$''. This
would be useful in situations where the upper bounds depend on the minimum
value that an inner element of type $T$ can take, for example in recursive
definitions where the value increments in every invocation.
Similarly, if we replace the sum in the definition of
\typednorm{\cdot}{T} by ``$\max$'', then we estimate the depth of
terms instead of the number of their constructs. This is useful in
cases like the  example in Section~\ref{sec:intro}.
Note that all these norms can be used at the same time, so the size of
some elements could be measured using minimum, maximum, depth, etc.,
if these values are relevant for the cost.
%

\subsection{Our Transformational Approach}\label{sec:our-transf-appr}

Next we describe how to use typed-norms to translate RBR programs to abstract programs that only contain procedure calls and constraints between sizes. From this information the resource analysis framework can produce cost relations~\cite{AlbertAGP11a} to obtain the desired bounds.
The main feature of our approach is that we allow the use of several abstractions for the same
variable at the same time, as in~\cite{BCGGV:TOPLAS:2007}. Thus, we can estimate the size of a term using different measures, and even
relations between sizes of different measures which might be crucial for precision as claimed in~\cite{BCGGV:TOPLAS:2007}.

This is important since two different parts of the program might traverse two different parts of the same data structure, so having both
measures allows us to provide tighter bounds.
%
%
%
Since rules will usually contain variables, we need \emph{symbolic} versions of typed-norms. These versions extend the typed-norms presented in Definition~\ref{def:typed-norms} with new cases for handling variables. In the following definition we only show the symbolic typed-norm \typednorm{\cdot}{T}, but it is analogous for \typednormorig{\cdot}{T}. Notice that we use the same notation to represent typed-norms for closed terms and their symbolic counterparts because the version used is clear from the context.

\clearpage
\begin{definition}[Symbolic typed-norms]\label{def:symb-typed-norms}
The symbolic typed-norm to compute the size of a term $t$ (possibly with variables) regarding a type $T$ is defined as:
	\[
	\typednorm{t}{T} =
	\begin{cases}
	X_T & \mbox{ if } t \equiv x \mbox{ and } T \in \deptypes{\type{x}}\\
    - \infty  & \mbox{ if } t \equiv x, T = \tint \mbox{ and } \tint \notin \deptypes{\type{x}}\\

	n & \mbox{ if } t \equiv n \mbox{ and } T = \inttype\\
	\typednorm{e_1}{T} \pm \typednorm{e_2}{T} & \mbox{ if } t \equiv e_1 \pm e_2 \mbox{ and } T=\inttype 
	\\
%
%
	1 + \sum_{i=1}^n \typednorm{t_i}{T} & \mbox{ if } t = \DC(t_1,\ldots, t_n) \mbox{ and } \type{t} = T\\
	\mi{max}_{i=1}^n \typednorm{t_i}{T} & \mbox{ if } t = \DC(t_1,\ldots, t_n) \mbox{ and } \type{t} \neq T\\
	0 & \mbox{ in other case} \\
	\end{cases}
	\]

\end{definition}

    Note that, as in the closed case, the max of an empty set is $-\infty$ if
    $T = \inttype$, and $0$ otherwise.
	When a variable $x$ is abstracted using a type $T \in \deptypes{\type{x}}$, it generates a variable $X_T$.
	On the other hand, abstracting variables not related to \inttype{} using the type $T = \inttype{}$ generates the size $-\infty$. This special case is useful to obtain precise sizes in terms containing integer numbers and variables like pairs $(\inttype, \TYPE{Seq})$. For example, we get  $\typednorm{\TYPE{Pair}(-5, x)}{\inttype} = max(\typednorm{-5}{\inttype}, \typednorm{x}{\inttype}) = max(-5, -\infty) = -5$ because $\tint \notin \deptypes{\type{x}} = \{ \TYPE{Seq} \}$, which is more precise than a size of $0$ obtained if $\typednorm{x}{\inttype} = 0$.
	Integer numbers are abstracted to themselves if $T$ is $\inttype$, and
	in integer expressions with $T = \inttype$ both subexpressions are abstracted recursively.
	The cases for data constructors are similar to Definition~\ref{def:typed-norms}.
	Finally, if the type $T \neq \inttype$ is not valid wrt.\ the type of $t$, i.e., $t$ cannot include any subterm of type $T$, then the symbolic typed-norm simply returns $0$.
	Notice that the symbolic version is equivalent to the typed-norm defined in Definition~\ref{def:typed-norms} for closed terms.

	As an example, consider the term $3 + x$ where both $3$ and $x$ have type \inttype. The abstractions wrt.\ $\inttype$ and $\tttype{IntList}$ are 
	$\typednorm{3+x}{\inttype} = 3 + X_\inttype$  and $\typednorm{3+x}{\tttype{IntList}} = 0$. If we consider a more complex term like $t \equiv$ \cbstart \code{Cons(z,Cons(5,Nil))} \cbend we have that $\typednorm{t}{\inttype} = \max(Z_\inttype, max(5, 0))$  and $\typednorm{t}{\tttype{IntList}} = 1+1+1 =3$.

\begin{figure}[tbp]
\begin{center}
	\begin{tabular}{l@{~~~~}l}
		\toprule
		Guard ($g$) & Guard abstraction ($g^\alpha$) \\
		\midrule
		true & $\top$\\
		$g_1 \wedge g_2$ & $g_1^\alpha \wedge g_2^\alpha$ \\
		$e_1~\mathit{op}~e_2$ & \cbstart $(e_1~\mathit{op}~e_2)[\overline{y\mapsto Y_\inttype}]$, \cbend where $\mi{op} \in\{>, =, \geq\}$ \\
		                      & ~~~and $\{\overline y\} = \mi{vars}(e_1~\mathit{op}~e_2)$\\
		$\match(x,p)$ & $\bigwedge\{ X_T = \typednorm{p}{T} \mid T\in\typednorms(x) \}$ \\
		$\notmatch(x,p)$ & $\top$ \\
		~\\
		\toprule
		Statement ($b$) & Statement abstraction ($b^\alpha$) \\
		\midrule
		$p(\bar{x},\bar{y})$ & $p(\bar{X},\bar{Y})$ \\
		$x := t$ & $\bigwedge\{ X_T = \typednorm{t}{T} \mid T\in\typednorms(x) \}$ \\
		~\\
		\toprule
		Rule ($r$) & Rule abstraction ($r^\alpha$) \\
		\midrule
		$p(\bar{x},\bar{y}) \rulearrow g, b_1, \ldots , b_n$ & $p(\bar{X},\bar{Y}) \rulearrow g^\alpha \wedge \nonneg{S}, b_1^\alpha, \ldots, b_n^\alpha$\\
				                        &  ~~~where $S = \bigcup_{i=1}^n \vars(b_i) \cup \vars(g) \cup \{\bar{y}\} \cup \{\bar{x}\}$   \\
		~\\
		\toprule
		Procedure ($\mi{Proc}$) &  Procedure abstraction ($\mi{Proc}^\alpha$) \\
		\midrule
		$p :: T_1 \times \cdots \times T_k ~ r_1 \ldots r_m$ & $r_1^\alpha \ldots r_m^\alpha$\\
	\end{tabular}
\end{center}
	\caption{Size abstraction for guards, statements, rules, and procedures}
	\label{fig:sizeabst}
\end{figure}

In some cases it is inevitable to measure the size of a variable $x$ wrt.\ all its dependent types---$\deptypes{\type{x}}$---to bound the resource consumption of a program. However, \cbstart often many types \cbend do not play any role for termination or resource consumption and therefore can be safely ignored. In order to consider only those important types we use the notion of \emph{relevant types}.
For every variable $x$ and rule $i$ in the program, we write $\typednorms_i(x)$ to refer to the set of types wrt.\ which we want to measure the size of that variable. In Section~\ref{sec:inference} we explain how to automatically infer these types.
If the rule is clear from the context, we usually omit the subscript.
Note that the set $\typednorms(x)$ must be a subset of all dependent types of a variable $x$ in a program rule, i.e., $\typednorms_i(x) \subseteq \deptypes{\typer{x}{i}}$.

Using the above symbolic typed-norms and the notion of relevant type-norms we can define our transformation of RBR programs.
Figure~\ref{fig:sizeabst} contains the definition of the function $\cdot^\alpha$ that abstracts guards, statements, rules, and procedures into \cbstart procedure calls and conjunctions of arithmetic constraints between sizes\footnote{Relations ($=, >, \geq$) between linear expressions containing integer variables representing sizes. }. \cbend In all these cases we assume that given a type $T\in\typednorms(x)$ the variable $X_T$ is an
integer valued variable representing the size of (the value of) $x$ w.r.t.\ the typed-norm $\typednorm{.}{T}$.
If $T\neq\inttype$, then we implicitly assume $X_T\geq 0$, as constructed terms cannot have negative size.
For a sequence of variables $\bar{x}$, we consider $\bar{X}$ is a sequence
that results from replacing each variable $x_i$ by
$X_{T_1}^i,\ldots,X_{T_n}^i$, where $\typednorms(x_i) =
\{T_1,\ldots,T_n\}$.
%
%

As shown in Figure~\ref{fig:sizeabst}, the trivial guard \code{true} is transformed into the truth value $\top$, the identity element of conjunction. Conjunction of guards ($g_1 \wedge g_2$) is abstracted into the conjunction of their size abstractions. Arithmetic guards are abstracted by replacing their variables $y$ with $Y_\inttype$.
For example, the expression $5$ is directly abstracted as $5$, whereas $x + y + 8$ is abstracted as $X_\inttype + Y_\inttype + 8$.
Match guards---$\match(x,t)$---generate a conjunction with as many elements as types in $\typednorms(x)$. Each element in the conjunction will be an equality between the variable $X_T$ (the size of $x$ wrt.\ $T$) and $\typednorm{t}{T}$ (the size of the term $t$ wrt.\ $T$, computed using the symbolic typed-norm). For example, if $\typednorms(z) = \{\inttype, \tttype{IntList}\}$, the abstraction of $\match(z,\mi{Cons(6, Nil)})$ will be the conjunction $Z_\inttype = max(6,0) \wedge Z_\tttype{IntList} = 2$. On the other hand, the $\notmatch(x,t)$ guard is abstracted to the truth value $\top$, as it does not provide any information relating the sizes of $x$ and $t$. Regarding statements, the abstraction of a procedure call $p(\bar{x},\bar{y})$ simply replaces each variable with variables related to its different sizes\footnote{When needed, we use angles ``$\langle$'' and ``$\rangle$'' to enclose input and output variables, as we do in Section~\ref{sec:intermediate-representation}.} (for example, $p(x,y)$ is abstracted to $p(X_\inttype, X_\tttype{IntList}, Y_\inttype)$ considering $\typednorms(x) = \{\inttype, \tttype{IntList}\}$ and $\typednorms(y) = \{\inttype\}$). Likewise, an assignment $x := t$ is abstracted
by generating a conjunction of equalities between the different
variables $X_T$ and their sizes. The abstraction of a rule $p(\bar{x},\bar{y}) \rulearrow g, b_1, \ldots , b_n$ proceeds compositionally by abstracting its head, the guard and all the statements.  As mentioned before, non-integer variables in a rule are assumed to be non-negative. The abstraction inserts these constraints explicitly using function $\nonneg{S}$, where $S$ is the set of variables occurring in the rule. The definition of $\nonneg{S}$ is the following:
\[
\nonneg{S} = \bigwedge \{ X_T \geq 0 ~|~ x \in S, T \in \typednorms(x), T \neq \inttype \}
\]
Finally, the abstraction of a procedure is the abstraction of all its rules.

\begin{definition}[Program abstraction]\label{def:programAbstraction}
	Given a program $\rbrprog \equiv [\overline{\mi{Dd}}] ~ \overline{\typedproc}$, its size abstraction $\rbrprog^\alpha$ is obtained by abstracting each procedure, i.e., $\rbrprog^\alpha \equiv \overline{\typedproc^\alpha}$.
\end{definition}
 Figure~\ref{fig:abstraction} shows
the abstraction for a fragment of the running example \cbstart given in Figure~\ref{fig:running}\cbend.
When using the typed-norm \typednorm{\cdot}{T} in Definition~\ref{def:symb-typed-norms},
$P^\alpha$ might include constraints of the form $X_T=E$ where
$E$ is an arithmetic expression that involves ``$\max$''.
If the second phase of the resource analysis framework does not support this kind of constraints, they can be approximated using linear constraints as follows:
(1) replace the sub-expression $\max(B_1,\ldots,B_n)$ by a new auxiliary variable $A$, (2) add the constraints $A \geq B_1\wedge\cdots\wedge A \geq B_n$; and (3) if $B_i \geq 0$ for all $1\le i\le n $, then add the constraint $B_1 + \ldots + B_n \geq A$ as well.
%
This process might be applied repeatedly in case of nested or multiple
occurrences of $\max$. In the case of nested ``$\max$'' expression,
we try to flatten them first.
Notice that this approximation is not the only
approach to handle $\max$ constraints, as they could also be
approximated as in~\cite{AlonsoAG11}. Moreover, in practice, the
constant $-\infty$ (that we used in the definition of norms) can
safely be replaced by any other constant --- replacing it by the
minimum value that \emph{syntactically} appears in the program (or by
$0$) works well in practice.

\begin{figure}[tbp]
\begin{lstlisting}[frame=none,breaklines=false]
(*\hspace{-0.85cm}$\circled{4}\hspace{0.32cm}$*) factSum((*$\langle$*)L$_\inttype$, L$_\tttype{IntList}$$\rangle$, $\langle$Sum'$_\inttype$$\rangle$) $\rulearrow$ $\top$ (*\emph{\textcolor{gray}{// RBR line \ref{fig:runnning:factSum0}}}*)
  $\wedge$ L$_\tttype{IntList} \geq 0$, (*\cbstart \cbend*) (*\emph{\textcolor{gray}{// non-negativity constraints}}*)
  Sum$_\inttype$ = 0, (*\emph{\textcolor{gray}{// RBR line \ref{fig:runnning:factSum1}}}*)
  while_1($\langle$L$_\inttype$, L$_\tttype{IntList}$, Sum$_\inttype$$\rangle$, $\langle$Sum'$_\inttype$$\rangle$) (*\emph{\textcolor{gray}{// RBR line \ref{fig:runnning:factSum2}}}*)

(*\hspace{-0.85cm}$\circled{5}\hspace{0.32cm}$*) while_1($\langle$L$_\inttype$, L$_\tttype{IntList}$, Sum$_\inttype$$\rangle$,(*\label{fig:abstraction:while_1:1}*) $\langle$Sum'$_\inttype$$\rangle$) $\rulearrow$ (*\emph{\textcolor{gray}{// RBR line \ref{fig:runnning:while1_0}}}*)
  L$_\inttype$ = $-\infty$ $\wedge$ L$_\tttype{IntList}$ = 1 (*\emph{\textcolor{gray}{// RBR line \ref{fig:runnning:while1_1}}}*)
  $\wedge$ L$_\tttype{IntList} \geq 0$, (*\emph{\textcolor{gray}{// non-negativity constraints}}*)
  Sum'$_\inttype$ = Sum$_\inttype$ (*\emph{\textcolor{gray}{// RBR line \ref{fig:runnning:while1_2}}}*)

(*\hspace{-0.85cm}$\circled{6}\hspace{0.32cm}$*) while_1($\langle$L$_\inttype$, L$_\tttype{IntList}$, Sum$_\inttype$$\rangle$,(*\label{fig:abstraction:while_1:2}*)$\langle$Sum'$_\inttype$$\rangle$) $\rulearrow$ (*\emph{\textcolor{gray}{// RBR line \ref{fig:runnning:while1_3}}}*)
	L$_\inttype$ = max(E$_\inttype$, L1$_\inttype$) $\wedge$ L$_\tttype{IntList}$ = 1 + L1$_\tttype{IntList}$ (*\emph{\textcolor{gray}{// RBR line \ref{fig:runnning:while1_4}}}*)
	$\wedge$ L$_\tttype{IntList} \geq 0$	$\wedge$ $L1_\tttype{IntList} \geq 0$, (*\emph{\textcolor{gray}{// non-negativity constraints}}*)
	fact($\langle$E$_\inttype$$\rangle$, $\langle$Prod$_\inttype$$\rangle$), (*\emph{\textcolor{gray}{// RBR line \ref{fig:runnning:while1_5}}}*)
	Sum1$_\inttype$ = Sum$_\inttype$ + Prod$_\inttype$, (*\emph{\textcolor{gray}{// RBR line \ref{fig:runnning:while1_6}}}*)
	while_1($\langle$L1$_\inttype$, L1$_\tttype{IntList}$, Sum1$_\inttype$$\rangle$,  $\langle$Sum'$_\inttype$$\rangle$) (*\emph{\textcolor{gray}{// RBR line \ref{fig:runnning:while1_7}}}*)

(*\hspace{-0.85cm}$\circled{7}\hspace{0.32cm}$*) main($\langle$$\rangle$, $\langle$R'$_\inttype$$\rangle$) $\rulearrow$ $\top$ (*\emph{\textcolor{gray}{// RBR line \ref{fig:runnning:main_0}}}*)
	$\wedge$ L$_\tttype{IntList} \geq 0$,(*\emph{\textcolor{gray}{// non-negativity constraints}}*)
	L$_\inttype$ = 2 $\wedge$ L$_\tttype{IntList}$ = 2,(*\emph{\textcolor{gray}{// RBR line \ref{fig:runnning:main_1}}}*)
	factSum($\langle$L$_\inttype$, L$_\tttype{IntList}$$\rangle$, $\langle$R'$_\inttype$$\rangle$)(*\emph{\textcolor{gray}{// RBR line \ref{fig:runnning:main_2}}}*)
\end{lstlisting}
\caption{Abstraction of a fragment of the RBR program from Figure~\ref{fig:running}.\label{fig:abstraction}}
\end{figure}

\begin{example}[Program abstraction]\label{ex:programAbstraction}
  Figure~\ref{fig:abstraction} shows the abstraction using the typed-norm $\typednorm{\cdot}{T}$ of the procedures \code{factSum}, \code{while_1} and \code{main} from Figure~\ref{fig:running}. Each line contains a comment indicating if it involves non-negativity constraints or is the abstraction of a particular line of the original RBR program.
  We assume that list variables ($x \in \{$ \code{l}, \code{l1}$\}$) have $\typednorms(x) = \{\inttype, \tttype{IntList}\}$, whereas integer variables ($x \in \{$ \code{sum}, \code{sum'}, \code{e}, \code{prod}, \code{sum1}, \code{r'}$\}$) have $\typednorms(x) = \{\inttype\}$. The abstraction proceeds rule by rule, abstracting rule heads, guards, and statements; \cbstart so rule numbers (gray circles in the left margin) does not change from the concrete program. \cbend Notice that the constraints from the guard are combined with the non-negativity of all the non-integer variables of each rule.
  %
  The most interesting part is the abstraction of the guard \code{match(l, Cons(e,l1))} in the second rule of \code{while_1}. For the type \inttype, the guard generates L$_\inttype$ = max(E$_\inttype$, L1$_\inttype$)---which could be approximated by linear constraints---and for the type \tttype{IntList} the guard generates the constraint $L_\tttype{IntList} = 1 + L1_\tttype{IntList}$, stating that the new list \code{l1} is one constructor smaller than \code{l}.

\end{example}

\cbstart
Note that deterministic programs can be abstracted to non-deterministic programs if different terms of the same type have the same size. For example, consider the simple data type \code{data Dir = Up \| Down}, and the following \lstinline@price :: $\langle \mathsf{Dir}\rangle \times \langle\tint\rangle$@ procedure:
\vspace*{-0.35cm}
\begin{lstlisting}[frame=none, multicols=2, numbers=none]
price($\langle$x$\rangle$,$\langle$y$\rangle$) $\leftarrow$ match(x, Up), y := 10
price($\langle$x$\rangle$,$\langle$y$\rangle$) $\leftarrow$ match(x, Down), y := 0
\end{lstlisting}
\vspace*{-0.35cm}
Since $\typednorms(x) = \mathsf{Dir}$ and $\typednorms(y) = \mathsf{Int}$ in both rules, and $\typednorm{\mathsf{Up}}{\mathsf{Dir}} = \typednorm{\mathsf{Down}}{\mathsf{Dir}} = 1$, the program abstraction results in the nondeterministic version:
\vspace*{-0.35cm}
\begin{lstlisting}[frame=none, multicols=2, numbers=none]
price($\langle X_{Dir}\rangle$,$\langle Y_\inttype\rangle$) $\leftarrow$
    $X_{Dir}$ = 1, $Y_\inttype$ = 10
price($\langle X_{Dir}\rangle$,$\langle Y_\inttype\rangle$) $\leftarrow$
    $X_{Dir}$ = 1, $Y_\inttype$ = 0
\end{lstlisting}
\vspace*{-0.35cm}
However, the possible nondeterminism introduced by the abstraction is not a problem from the point of view of soundness, as any trace in the deterministic program can be performed in the nondeterministic abstraction. This result is presented in the next section.
\cbend

\subsection{Soundness} \label{sec:soundness}
The abstraction presented in the previous section transforms RBR
programs into abstract programs, which the resource analysis framework
can take as input for computing bounds. Therefore, we ensure that the size of variables that are observed in $\rrderiv$ traces of RBR programs can be observed in traces of their abstractions. In order to prove this soundness result, we need an abstract operational semantics for abstract programs, which is an adaptation of the operational semantics for RBR programs in Figure~\ref{fig:rrsem}. This abstract operational semantics evaluates \emph{abstract configurations} step by step:

\begin{definition}[Abstract configurations]\label{def:abstractConfigurations}
An \emph{abstract configuration}, denoted as $\airstate \equiv \arb \abstractSep \psi$, is a sequence $\arb$ of
\emph{abstract activation records} followed by a conjunction $\psi$ of constraints. Similarly, an abstract activation record, denoted as $\tuple{\bsa}$, is a sequence of constraint conjunctions and abstract procedure calls $p(\bar{X}, \bar{Y})$. We use Greek letters ($\psi$, $\varphi$, \ldots) to denote constraint conjunctions, $\bca$ to denote one constraint conjunction or one abstract procedure call, and $\bsa$ to denote a sequence $\bca_1 \cdots \bca_n$.
\end{definition}


\begin{center}
\begin{figure}[t]
\begin{center}
\(
\begin{array}{|lc|}
  \hline 
  (1) & \semrule
  {\psi{\wedge}\varphi \not\models {\it false}
  }
  {
  	\tuple{\varphi \cdot \bsa} \cdot \arb \abstractSep \psi
  	\rrabsderiv
  	\tuple{\bsa} \cdot \arb \abstractSep \psi{\wedge}\varphi\\
  }\\[0.5cm]
  (2) & \semrule
  {
  	\cbstart p(\bar{X'},\bar{Y'})  \rulearrow \varphi, \bca_1,\ldots, \bca_n \cbend 
  	\in \palpha \mi{~fresh}
    ~~~~~~
     \overline{X = X'} \wedge\varphi{\wedge}\psi \not\models {\it false}
  }
  {
  	\tuple{p(\bar{X},\bar{Y}) \cdot \bsa} \cdot \arb \abstractSep \psi
    \rrabsderiv
    \tuple{\bca_1 \cdots \bca_n}^{\overline{Y = Y'}} \cdot \tuple{\bsa} \cdot \arb \abstractSep \overline{X = X'} \wedge\varphi{\wedge}\psi
  }\\[0.5cm]  
  (3) & \semrule
    {\psi{\wedge}\overline{Y = Y'} \not\models {\it false}
    }
    {
    	\tuple{\epsilon}^{\overline{Y = Y'}} \cdot \arb \abstractSep \psi
    	\rrabsderiv
    	\arb \abstractSep \psi{\wedge}\overline{Y=Y'}
    }\\[0.3cm]
    \hline
\end{array}
\)
\end{center}
\caption{Operational semantics of abstract rule-based programs}\label{fig:absrrsem}
\end{figure}
\end{center}

Figure~\ref{fig:absrrsem} contains the operational semantics $\rrabsderiv$ that evaluates abstract configurations regarding an abstract program $\rbrprog^\alpha$. Rule $(1)$ handles abstracted assignments, which are translated into conjunctive constraints $\varphi$. If $\varphi$ is consistent wrt.\ the global set of constraints ($\psi\wedge\varphi \not\models {\it false}$) then it is added to the global constraints and the evaluation continues with the next element of the abstract activation record. Rule $(2)$ handles procedure invocation. First, we obtain a fresh rule
\cbstart from the abstract program $P^\alpha$\cbend.
Then, if the parameter passing, the guard and the global constraints are consistent \cbstart ($\overline{X = X'} \wedge\varphi{\wedge}\psi \not\models {\it false}$),\footnote{We use the notation $\overline{X = X'}$ to denote the conjunction of constraints $\bigwedge \overline{X= X'}$.}
\cbend a new abstract activation record containing the body of the procedure is inserted in the abstract configuration.

Similarly to the $\rrderiv$ semantics, the relation $\overline{Y = Y'}$ between the output variables and the parameters is stored as a mark. Finally, rule $(3)$ removes an empty abstract activation record if the output variables are consistent with the global constraints. Notice that global constraints accumulate variables from all the abstract activation records, even those whose execution has finished and therefore have been removed.
\cbstart
Similarly to $\rrderiv$, we use the notion of \emph{abstract trace} $\trace^\alpha = \airstate_0 \rrabsderiv^{r_1} \airstate_1 \rrabsderiv^{r_2} \ldots \rrabsderiv^{r_n} \airstate_n$ where its \emph{steps} are defined as $\trsteps{\trace^\alpha} = \tuple{r_1, r_2, \ldots, r_n}$, and the set of \emph{abstract trace steps} as the combination of the steps of all the possible abstract traces starting from a given configuration ($\traces(\airstate_0) =  \{ \trsteps{\airstate_0 \rrderiv^{r_1} \airstate_1 \rrderiv^{r_2} \ldots \rrderiv^{r_n} \airstate_n} \mid n \in \mathbb{N}\}$.
\cbend

\begin{example}[Evaluation of an abstract program]\label{ex:abs_evaluation}
	The abstract evaluation of the \code{main} procedure in Figure~\ref{fig:running} proceeds similarly to Example~\ref{ex:rbr_evaluation}. When obtaining fresh names for variables---rule ($2$) in Figure~\ref{fig:absrrsem}---we use superscripts with the same number as the current abstract configuration $\airstate_n$. For simplicity, the statement or constraint that controls each step is underlined, the constraints in each abstract configuration $\airstate_i$ are denoted as $\psi_i$ and $\top$ is omitted. \cbstart Each step $\rrabsderiv^{a\cdot b}$ is decorated with the semantic rule $a$ and abstract program rule $b$ used, \cbend and $\rrabsderiv^*$ refers to many $\rrabsderiv$-steps.

\[
	\begin{array}{l@{}l}
	\airstate_0 \equiv & \tuple{\underline{L^0_\inttype = 2 \wedge L^0_\tttype{IntList} = 2} \cdot \TYPE{factSum}(\langle L^0_\inttype, L^0_\tttype{IntList}\rangle, \langle R'^0_\inttype\rangle)} ~|~ L^0_\tttype{IntList} \geq 0 \rrabsderiv^{(1)\cdot \epsilon} \\

	\airstate_1 \equiv & \tuple{\underline{\TYPE{factSum}(\langle L^0_\inttype, L^0_\tttype{IntList}\rangle, \langle R'^0_\inttype\rangle)}} ~|~ L^0_\tttype{IntList} \geq 0 \land L^0_\inttype = 2 \wedge L^0_\tttype{IntList} = 2 \rrabsderiv^{(2)\cdot 4} \\

	\airstate_2 \equiv & \tuple{\underline{\mi{Sum}^2_\inttype = 0} \cdot \TYPE{while}\_1(\langle L^2_\inttype, L^2_\tttype{IntList}, \mi{Sum}^2_\inttype\rangle, \langle\mi{Sum}'^2_\inttype\rangle)}^{R'^0_\inttype = \mi{Sum}'^2_\inttype} ~\cdot \\
	& \tuple{\epsilon} ~|~ \psi_1 \land L^0_\inttype = L^2_\inttype \wedge L^0_\tttype{IntList} = L^2_\tttype{IntList} \wedge L^2_\tttype{IntList} \geq 0 \cbstart \rrabsderiv^{(1)\cdot \epsilon} \cbend \\

	\airstate_3 \equiv & \tuple{\underline{\TYPE{while}\_1(\langle L^2_\inttype, L^2_\tttype{IntList}, \mi{Sum}^2_\inttype\rangle,  \langle\mi{Sum}'^2_\inttype\rangle)}}^{R'^0_\inttype = \mi{Sum}'^2_\inttype} ~\cdot \\
	& \tuple{\epsilon} ~|~ \psi_2 \land \mi{Sum}^2_\inttype = 0
	\rrabsderiv^{(2)\cdot 6} \\

	\airstate_4 \equiv & \tuple{\underline{\TYPE{fact}(\langle E^4_\inttype\rangle, \langle Prod^4_\inttype\rangle)} \cdots}^{\mi{Sum}'^2_\inttype = \mi{Sum}'^4_\inttype} ~\cdot \\
	& \tuple{\epsilon }^{R'^0_\inttype = \mi{Sum}'^2_\inttype} ~\cdot \\
    & \tuple{\epsilon} ~|~ \psi_3 \land L^2_\inttype = L^4_\inttype \wedge L^2_\tttype{IntList} = L^4_\tttype{IntList} \wedge \mi{Sum}^2_\inttype = \mi{Sum}^4_\inttype \wedge \\
	& L^4_\inttype = \max(E^4_\inttype, L1^4_\inttype) \wedge L^4_\tttype{IntList} = 1 + L1^4_\tttype{IntList} \wedge \\
	& L^4_\tttype{IntList} \geq 0 \wedge L1^4_\tttype{IntList} \geq 0 \rrabsderiv^* \\

	\airstate_5 \equiv & \tuple{\underline{\epsilon}}^{\mi{Sum}'^2_\inttype = \mi{Sum}'^4_\inttype} \cdot \\
	& \tuple{\epsilon }^{R'^0_\inttype = \mi{Sum}'^2_\inttype} ~\cdot \\
	& \tuple{\epsilon} ~|~ \psi_5 \rrabsderiv^{(3)\cdot \epsilon} \\

	\airstate_6 \equiv & \tuple{\underline{\epsilon}}^{R'^0_\inttype = \mi{Sum}'^2_\inttype} ~\cdot \\
	& \tuple{\epsilon} ~|~ \psi_5 \land \mi{Sum}'^2_\inttype = \mi{Sum}'^4_\inttype \rrabsderiv^{(3)\cdot \epsilon} \\

	\airstate_7 \equiv & \tuple{\epsilon} ~|~ \psi_6 \land R'^0_\inttype = \mi{Sum}'^2_\inttype \\

	\end{array}
	\]

	\noindent In $\airstate_3$ the first program rule of \code{while_1} (Line~\ref{fig:abstraction:while_1:1} in Figure~\ref{fig:abstraction}) cannot be used to evaluate the call because the guard is not compatible with the global set of constraints, namely:
    $L^2_\tttype{IntList} = L^4_\tttype{IntList} \wedge
    L^4_\tttype{IntList} = 1 \land
    L^0_\tttype{IntList} = L^2_\tttype{IntList} \land
    L^0_\tttype{IntList} = 2
    \models \mi{false}$.
	Therefore, the evaluation can only proceed with the second rule of \code{while_1}
	(Rule~6, Line~\ref{fig:abstraction:while_1:2}
	in Figure~\ref{fig:abstraction}).
\end{example}

The soundness result relates $\rrderiv$-traces starting from a configuration $\irstate$ with abstract $\rrabsderiv$-traces starting from the abstraction of $\irstate$. The following definitions present the abstraction of configurations, which is based on \typednorm{\cdot}{T} and the abstraction of statements defined in Figure~\ref{fig:sizeabst}.

\begin{definition}[Variable mapping abstraction]\label{def:varmap_abstr}
	The abstraction of a variable mapping $\tv$	is defined as $\atv =
	\bigwedge \{ X_T = \typednorm{\tv(x)}{T} ~|~ x\in \dom(\tv), T \in \typednorms(x) \}$.
	Notice that $\atv$ is a conjunction of equalities between distinct $X_T$ variables and integer values, since $\tv(x)$ are concrete  values and therefore
	\typednorm{\tv(x)}{T} generates integer numbers.
\end{definition}

\begin{definition}[Configuration abstraction]\label{def:state_abstr}
	Let $\irstate = \tuple{p_1,\stkbc_1,\tv_1}{\cdot}{\cdot}{\cdot} \tuple{p_n,\stkbc_n,\tv_n}$ be a configuration. Its abstraction is defined as:
	$\irstate^\alpha = \tuple{\astkbc{1}}{\cdot}{\cdot}{\cdot} \tuple{\astkbc{n}} \abstractSep \psi$ where $\psi = \atv_1 \wedge \ldots \wedge \atv_n$.
	Notice that $\psi$ is a conjunction of equalities between distinct $X_T$ variables and integer values, since every activation record uses
	fresh variables and $\psi$ is the conjunction of abstracted variable mappings.
\end{definition}

The following \cbstart theorem establishes the soundness of our translation using
typed-norms: for any trace $\irstate_0 \rrderiv^* \irstate_n$ it is
possible to create an abstract trace $\irstate_0^\alpha \rrabsderiv^* \arb
\abstractSep \psi$ with the same steps.\cbend

\cbstart
\begin{theorem}[Soundness]\label{teo:soundness}
	If $\trace \equiv \irstate_0 \rrderiv^*
	\irstate_n$ then there is an abstract trace
	$\trace^\alpha \equiv \irstate_0^\alpha \rrabsderiv^* \arb \abstractSep \psi$
	such that $\trsteps{\trace} = \trsteps{\trace^\alpha}$, $\irstate_n^\alpha = \arb \abstractSep \widetilde{\psi}$ and $\psi \wedge \widetilde{\psi} \not\models \false$.
\end{theorem}
\cbend

Intuitively, the above theorem states that the sizes of the variables
of the concrete configuration $\irstate_n$, w.r.t.\ the corresponding
norm, define a model of the abstract state configuration.


\cbstart 

Let us now informally explain how abstract programs preserve cost. The
idea is to simulate the same process that we have described at the end
of Section~\ref{sec:intermediate-representation}.
First, during the abstraction phase the instruction \lst{tick($n$)} is
kept in the abstract program.
Next, we instrument abstract activation records with a cost counter
$\mathtt{cost}^\alpha$, and add a corresponding abstract semantic
rule for \lst{tick($n$)} that simply sets $\mathtt{cost}^\alpha$ to
$\mathtt{cost}^\alpha-n$.
An abstract trace $\trace^\alpha$ is said to be valid if
$\texttt{cost}^\alpha$ never takes negative values, i.e.,
$\texttt{cost}$ is initialized with an amount of resources that is
enough to carry out this particular abstract execution.
Finally, a function $f^\alpha$ that maps initial abstract
configurations to nonnegative values is called an upper bound on the
(abstract) worst-case cost if, for and $\airstate_0$, setting the
initial value of $\mathit{cost}^\alpha$ to $f^\alpha(\airstate_0)$
guarantees generating valid abstract traces only.
It is called a lower bound on the best-case cost if, for any
$\airstate_0$, setting the initial value of $\texttt{cost}^\alpha$ to a
value smaller than $f^\alpha(\airstate_0)$ generates an invalid
abstract trace.
The function $f^\alpha$ is typically given in terms of the input abstract
variables, i.e., in terms of the sizes of the corresponding data.
The black-box component that infer the cost of the abstract program,
infers such functions.

Given the statement of Theorem~\ref{teo:soundness}, it is easy to see
that if $f^\alpha$ is is an upper (resp. lower) bound on the abstract
worst-case (resp. best-case) cost, then it is also an upper
(resp. lower) bound on the concrete cost (up to rewriting it in terms
of typed-norms instead of corresponding abstract variables).

As a final remark, let us mention that although the paper and the experiments
focus on upper bounds inference, our transformation is valid also to
infer lower bounds as it ensures that the cost of every trace is
preserved by the transformed program.
\cbend 


\section{Inference of Relevant Types}\label{sec:inference}

As explained in the previous section, in order to abstract the size of a variable $x$ in the rule number $i$ we consider a set $\typednorms_i(x)$ containing its \emph{relevant types}. It is safe to assume that $\typednorms_i(x)$ contains all the constituent types of $x$. However, the \cbstart complexity \cbend of the solving phase that obtains bounds grows exponentially with the number of variables involved, so it is very important to obtain the smallest sets for the relevant typed-norms. In this section we will present the inference algorithm for relevant types as well as its soundness result.

As it was observed in~\cite{AlbertAGPZ08}, variables that do not
affect the cost can be removed from the abstract program and the
bounds obtained do not change. In our setting, the cost of a program
depends primarily on the number of recursive calls performed, which is
affected by the guards in the rules. Therefore, any variable that does
not affect directly or indirectly the value of a guard can be
ignored. We push this idea further \cbstart and detect, from each variable, those types that do not affect \cbend a guard evaluation directly or
indirectly. These types are useless from the point of view of resource
analysis, so they should be discarded from the set of relevant types
of the variable.

The main intuition behind the algorithm for inferring relevant types
is detecting those constituent types of a guard variable that are
involved in the guard evaluation. For example, in \code{match(l,
  Cons(x,xs))} we say that the type \TYPE{IntList} of variable
\code{l} is involved in the guard evaluation because the pattern
matched is of type \TYPE{IntList}. On the other hand, the type \tint{}
of the same variable is not involved in the guard evaluation because
it can succeed or fail regardless of the possible \tint{} values stored in \code{l}.
Once we have this information from the guards, we propagate it backwards to include those relevant types in the rest of variables that can affect the value of the guard variable. 
\cbstart Relevant types will be propagated to the rules formal parameters, where they will be combined \cbend with the relevant types from the rest of rules of the same procedure. 
Similarly, when invoking a procedure, the relevant types of the formal parameters will be included in the variables of the actual arguments.

\begin{example}[Intuition of relevant types inference]
\label{ex:inference_intuition}
Consider the RBR program in Figure~\ref{fig:running}.
From Line~\ref{rbr:fact_while_1} (Rule \circled{2}) we discover that \tint{} is involved in the arithmetic guard \code{0 >= n} for procedure \code{while_0}, so $\tint \in \typednorms_2($\code{n}$)$.
The procedure \code{while_0} is invoked in
Line~\ref{rbr:fact_while0_invocation} of Rule \circled{1}, so that
this relevant type is propagated to the \code{fact} rule, i.e., $\tint \in \typednorms_1($\code{n}$)$. Similarly, \code{fact} is invoked in Line~\ref{rbr:while1_fact_invocation} of the Rule \circled{6} corresponding to the \code{while_1} predicate, so \tint{} will be a relevant type for variable \code{e}. That rule contains a guard \code{match(l, Cons(e,l1))} in Line~\ref{rbr:while1_match}, so \tint{} will be propagated from variable \code{e} to \code{l}, i.e., $\tint \in \typednorms_6($\code{l}$)$. After another step of propagation in Rule \circled{4} (procedure \code{factSum}) we will obtain that $\tint \in \typednorms_4(\mathsf{l})$. In summary, the relevant type \tint{} detected in Rule \circled{2} has been propagated to Rule \circled{4} following the path \circled{2} $\to$ \circled{1} $\to$ \circled{6} $\to$ \circled{4}. 
A similar but shorter process would produce that $\TYPE{IntList} \in
\typednorms_4($\code{l}$)$ by propagating it in the path \circled{5} $\to$ \circled{4}. In this case all the constituent types of \TYPE{IntList} are relevant types for the parameter \code{l} in \code{factSum} (Rule \circled{4}), as the resource usage depends both in the length of the list and its stored numbers.
\end{example}

\subsection{Formalization of relevant type inference}\label{sec:inference_formalization}

We formalize the relevant types inference algorithms as a data-flow analysis that constructs a mapping for the complete program. This mapping relates a set of relevant types to every variable in the program, and is defined as follows:

\begin{definition}[Rule and program mappings]\label{def:normsSet}
	A \emph{rule mapping} $\rulemaps \ni \rulemap = [\overline{x \mapsto P(\types)}]$ maps variables $\bar{x}$ in a rule to sets of types. The domain of $\rulemap$ is denoted as $\dom(\rulemap) = \{\bar{x}\}$, and $\rulemap(x)$ returns the set of relevant types related to variable $x$. We use $\epsilon$ to denote an empty rule mapping. 
	
	A \emph{program mapping} $\prmaps{\rbrprog{}} \ni \prmap = \langle \rulemap_1, \rulemap_2, \ldots, \rulemap_n\rangle$ aggregates the mappings of all the rules of a program $\rbrprog{}$, considering that $P$ has $n$ rules. We use the notation $\prmap(i) = \rulemap_i$ to refer to the i$^{th}$ rule mapping in $\prmap$, and $\prmap(i)(x)$ to refer to the set of relevant types of variable $x$ in the $i^{th}$ rule.
\end{definition}	

Rule and program mappings support the following set of standard operations (extension, ordering, combination, restriction and renaming) that we use in the formalization of the inference algorithm.

\begin{definition}[Operations on rule and program mappings]\label{def:normsSet_operations}	
	A rule mapping can be \emph{extended} to a program mapping by assigning it to the $i^{th}$ rule and considering the empty mapping for the rest of rules. The extension of a rule mapping is denoted as $\extrmap{\rulemap}{n}{i} = \langle \epsilon_1, \ldots, \epsilon_{i-1}, \rulemap, \epsilon_{i+1}, \ldots, \epsilon_n \rangle$.
	
	We consider the natural order of rule mappings based on subset inclusion. We say that $\rulemap \sqsubseteq \rulemap'$ iff.\ $\forall x \in \dom(\rulemap).~\rulemap(x) \subseteq \rulemap'(x)$. This order is extended to program mappings as follows: $\prmap \sqsubseteq \prmap'$ iff.\ $\prmap(i) \sqsubseteq \prmap'(i)$ for all rule $i$.
	
	Rule mappings can be combined using the commutative and associative operator $\addr$, defined as: 
	\[
	(\rulemap \addr \rulemap')(x) = \left\{
	  \begin{array}{rl}
        \rulemap(x) \cup \rulemap'(x) & \text{if } x \in \dom(\rulemap) \cap \dom(\rulemap') \\
	    \rulemap(x)                   & \text{if } x \in \dom(\rulemap) \smallsetminus \dom(\rulemap') \\
        \rulemap'(x)                  & \text{if } x \in \dom(\rulemap') \smallsetminus \dom(\rulemap) \\
        \text{undefined}              & i.o.c.
	\end{array} \right.
	\]
	If a variable $x$ appears in both rule mappings, its typed-norms are combined, otherwise it takes the typed-norms from the mapping where \cbstart it appears. \cbend Notice that variables not appearing in $\dom(\rulemap_1)$ \cbstart nor \cbend $\dom(\rulemap_2)$ are undefined in $\rulemap_1 \addr \rulemap_2$. 
	We also consider the combination of program mappings of the same length by proceeding element-wise. We overload the symbol $\addp$:
	\[
	\langle \rulemap_1, \ldots, \rulemap_n\rangle \addp \langle \rulemap'_1, \ldots, \rulemap'_n\rangle = \langle \rulemap_1 \addr \rulemap'_1,  \ldots, \rulemap_n \addr \rulemap'_n\rangle
	\]
	When combining sequences of mappings we use the notation $\baddr_{i=1}^n \rulemap_i = \rulemap_1 \addr \ldots \addr \rulemap_n$, and $\baddr_{\prmap \in C} \prmap = \prmap_1 \addp \ldots \addp \prmap_n$ for combining sets $C = \{\prmap_1, \ldots, \prmap_n\}$ of program mappings.
	
	Rule mappings can be restricted to a set of variables $C$, formally: 
	\[
	(\restrmap{\rulemap}{C})(x) = \left\{
	\begin{array}{rl}
	\rulemap(x) & \text{if } x \in C \\
	\text{undefined}              & i.o.c.
	\end{array} 
	\right.
	\]
	
	Finally, we use $\renmap{\rulemap}{\bar{x}}{\bar{y}}$ to denote the renaming of a rule mapping $\rulemap$, where the variables $\bar{x} = x_1, \ldots, x_n$ are renamed to $\bar{y} = y_1, \ldots, y_n$ (we consider that $\{\bar{x}\} \subseteq \dom(\rulemap)$ and $\{\bar{y}\} \cap \dom(\rulemap) = \emptyset$):
	\[
	(\renmap{\rulemap}{\bar{x}}{\bar{y}})(z) = \left\{
	\begin{array}{rl}
	\rulemap(x_i) & \text{if } z = y_i\\
	\rulemap(z) & \text{if } z \in \dom(\rulemap) \smallsetminus \{\bar{x}\} \\
	\text{undefined}              & i.o.c.
	\end{array} 
	\right.
	\]
\end{definition}

\begin{example}[Operations on rule and program mappings]
Consider the following rule mappings:
\[
\begin{array}{rcl}
\mu_1  & = & [x \mapsto \{\tint\}, y \mapsto \{\tint, \tintl\} ] \\ 
\mu_2  & = & [y \mapsto \{\tintl\} ] \\ 
\mu_3  & = & [y \mapsto \{\tint\} ]
\end{array}
\]
We have that $\mu_2 \sqsubseteq \mu_1$ ($\{\tintl\} \subseteq \{\tint, \tintl\}$ for variable $y$) and $\mu_3 \sqsubseteq \mu_1$ ($\{\tint\} \subseteq \{\tint, \tintl\}$ for variable $y$) but $\mu_1 \not \sqsubseteq \mu_2$ ($x \notin \dom(\mu_2)$) and $\mu_2 \not\sqsubseteq \mu_3$ ($\{\tintl\} \not\subseteq \{\tint\}$ for variable $y$). Regarding combination, we have that $\mu_2 \addr \mu_3 = [y \mapsto \{\tint, \tintl\}]$ and $\mu_1 \addr \mu_2 = \mu_1 \addr \mu_3 = \mu_1$. The restriction of rule mappings simply reduces the domain, so $\restrmap{\mu_1}{\{x\}} = [x \mapsto \{\tint\} ]$ and \cbstart $\restrmap{\mu_1}{\{y\}} = [y \mapsto \{\tint, \tintl\} ]$. \cbend Finally, the renaming changes the variables in the domain, for example $\renmap{\rulemap_1}{x,y}{a,b} = [a \mapsto \{\tint\}, b \mapsto \{\tint, \tintl\} ]$ because $x$ is renamed by $a$ and $y$ is renamed by $b$.
\end{example}

\begin{figure}[tb]
\[
  \begin{array}{rrl}
  \genf{\rbrprog{}} & : & \prmaps{\rbrprog{}} \mapsto \prmaps{\rbrprog{}} \\
  \genf{\rbrprog{}}(\prmap) & = & \prmap \addp \baddp_{i=1}^n \genf{\rbrprog{}}_i(\prmap) \\
  \end{array}
  \]
  
  \[
  \begin{array}{rrl}
  \genf{\rbrprog{}}_i & : & \prmaps{\rbrprog{}} \mapsto \prmaps{\rbrprog{}} \\
  \genf{\rbrprog{}}_i(\prmap) & = & \extrmap{\genG{\rbrprog{}}_i(g, \prmap)}{n}{i} \addp \genS{\rbrprog{}}_i(\{\bar{b}\}, \prmap) \addr \extrmap{\rulemap_i}{n}{i}, \text{ where } \\  
  & & ~~\getrule{i}{\rbrprog{}}=p(\bar{x},\bar{y}) \rulearrow g, \bar{b},~~~\rulemap_i = 
  \baddr_{\rulemap \in S_i} \rulemap \text{~and} \\
  & & ~~S_i =  \{ \renmap{(\restrmap{\prmap(j)}{\bar{w}\bar{z}})}{\bar{w}\bar{z}}{\bar{x}\bar{y}} \mid j \in [1..n], \getrule{j}{\rbrprog{}} = p(\bar{w},\bar{z})\rulearrow g', \bar{b'}\}
  \end{array}
  \]
    
  \[
  \begin{array}{rrl}
  \genG{\rbrprog{}}_i                     & : &  \text{Guard} \times \prmaps{\rbrprog{}} \mapsto \rulemaps\\
  \genG{\rbrprog{}}_i({\it true}, \prmap) & = & \epsilon \\
  \genG{\rbrprog{}}_i(g_1 \wedge g_2, \prmap) & = &\genG{\rbrprog{}}_i(g_1, \prmap) \addp \genG{\rbrprog{}}_i(g_2, \prmap)\\
  \genG{\rbrprog{}}_i(e_1~op~e_2, \prmap) & = &[ x \mapsto \{\inttype\} \mid x \in \vars(e_1) \cup \vars(e_2)] \\
  \genG{\rbrprog{}}_i(\match(x,p), \prmap) & = & [ x \mapsto \{ T \} \mid T = \typer{x}{i}, T \text{ is recursive}] ~\addp \\
                               &   & [ x \mapsto \{ T \mid y \in \vars(p),~ T \in \prmap(i)(y) \} ] \\
  \genG{\rbrprog{}}_i(\notmatch(x,p), \prmap) & = & [ x \mapsto \{ T \mid y \in \vars(p),~ T \in \prmap(i)(y) \} ] \\
  \end{array}
  \]
  
  \[
  \begin{array}{rrl}
  \genS{\rbrprog{}}_i         & : & P(\text{Statement}) \times \prmaps{\rbrprog{}} \mapsto \prmaps{\rbrprog{}} \\
  \genS{\rbrprog{}}_i(\{\bar{b}\}, \prmap) & = & \baddp_{b \in \{\bar{b}\}} \genS{\rbrprog{}}_i(b, \prmap) \\
  \end{array}
  \]  
  
  \[
  \begin{array}{rrl}
  \genS{\rbrprog{}}_i           & : & \text{Statement} \times \prmaps{\rbrprog{}} \mapsto \prmaps{\rbrprog{}} \\
  \genS{\rbrprog{}}_i(\rrassigns{x}{t}, \prmap) & = & \extrmap{[ y \mapsto T_{x,y,i} \mid y \in \vars(t)]}{n}{i}, \text{ where } \\
  	                            & & \cbstart ~~T_{x,y,i} = \{ T \mid T \in \prmap(i)(x) 
  	                            \cap  \deptypes{\typer{y}{i}} \} \cbend \\
  \genS{\rbrprog{}}_i(p(\bar{x},\bar{y}), \prmap) & = & \baddp_{\prmap \in A \cup B} \prmap, \text{ where }\\ 
    & & ~~A = \{ \extrmap{\renmap{(\restrmap{\prmap(j)}{\bar{w}})}{\bar{w}}{\bar{x}}}{n}{i} \mid j \in [1..n], \getrule{j}{\rbrprog{}} = p(\bar{w},\bar{z})\rulearrow g, \bar{b} \}\\
   & & ~~B = \{ \extrmap{\renmap{(\restrmap{\prmap(i)}{\bar{y}})}{\bar{y}}{\bar{z}}}{n}{j} \mid j \in [1..n], \getrule{j}{\rbrprog{}} = p(\bar{w},\bar{z})\rulearrow g, \bar{b}\}
  \end{array}
  \]
\caption{Functions for inferring typed-norms.\label{fig:infer}}
\end{figure}

%
Next we explain the inference algorithm, which is based on the definition
of $\genf{\rbrprog{}}$ that is given in Figure~\ref{fig:infer}.
  Let $\rbrprog{}$ be a program with $n$ rules ($\nrules{\rbrprog{}} = n$). Then the relevant types inferred for the program is the least fixed point of the function $\genf{\rbrprog{}}$,
  i.e., $\mi{lfp}(\genf{\rbrprog{}}) \in \prmaps{\rbrprog}$. 
  \cbstart Note that $\genf{\rbrprog{}}$ is a monotone function and $\prmaps{\rbrprog}$ is a finite complete lattice, so the least fixed point can be computed exactly. \cbend
  Considering this fixed point, the notation $\typednorms_i(x)$ 
  is defined as $\mi{lfp}(\genf{\rbrprog{}})(i)(x)$.
The function $\genf{\rbrprog{}}$ takes a program mapping $\prmap$ and extends it with the new information from the different rules of the program. For each program rule, it computes $\genf{\rbrprog{}}_i(\prmap)$ and combines them with the current program mapping.

  The function $\genf{\rbrprog{}}_i$ takes the current program mapping and extends it using the information in rule number $i$. This function processes a rule $p(\bar{x},\bar{y}) \rulearrow g, \bar{b}$ by combining the information from the guard ($\genG{\rbrprog{}}_i$), the body ($\genS{\rbrprog{}}_i$) and also collecting the relevant types from the parameters of all the rules of the same procedure (mapping $\rulemap_i$). This last step, which requires a renaming of the parameters $\bar{w}\bar{z}$ to $\bar{x}\bar{y}$, forces rules of the same predicate to have the same relevant types and therefore be abstracted to rules with the same name and number of abstracted parameters.

  The function \cbstart $\genG{\rbrprog{}}_i$ \cbend takes a guard $g$ and a program mapping $\prmap$ and generates new relevant types in a rule mapping. This function proceeds by combining the new relevant types obtained in every fragment of the guard. A \emph{true} guard always succeeds, so it does not impose any relevant type. Arithmetic guards $e_1~op~e_2$ requires \tint{} as a relevant type for any variable in the expression. If there is a $\match(x,p)$ guard and $x$ has a recursive type $T$ then 
  this type is included as a relevant type for $x$.
  Non-recursive types are ignored because they cannot directly affect the number of recursions. However, if it contains an inner recursive type it will be detected and propagated when computing the fixed point. 
  Finally, in $\match(x,p)$ and $\notmatch(x,p)$ guards all relevant types already detected in the variables of the pattern $p$ are propagated to the matching variable $x$.
    \cbstart 
    Note that in $\match(x,p)$ guards, the recursive type $T$ of $x$ will be only relevant if the procedure is recursive, otherwise it will not affect the execution. Therefore, in practice we can infer relevant types for recursive procedures and then transfer that information to non-recursive ones.
    \cbend
  
  The function $\genS{\rbrprog{}}_i$ is overloaded. When it takes a set of statements $\{\bar{b}\}$ and a program mapping \prmap{} then it traverses all the statements, invoking $\genS{\rbrprog{}}_i$ for each one and combining the results. When $\genS{\rbrprog{}}_i$ is invoked with only one statement $b$ then it propagates the information from $\prmap$ according to the statement processed. If $b \equiv \rrassigns{x}{t}$ then all the relevant types detected for $x$ that are constituents of the type of some variable $y$ of the right-hand side $t$ (set $T_{x,y,i}$) are propagated to that variable $y$. Finally, if the statement is a procedure call $p(\bar{x},\bar{y})$ then the function propagates the relevant types of the parameters. The set $A$ propagates all the relevant types for the input parameters $\bar{w}$ of the rules $j$ for the same procedure to the input parameters $\bar{x}$ of the current rule $i$. Similarly, the set $B$ propagates the relevant types of the output parameters $\bar{y}$ of the 
  current rule $i$ to the output parameters $\bar{z}$ of the rules $j$ of the same predicate. Sets $A$ and $B$ contain program mappings that are combined.

The result of the relevant types inference always exists, and it can be effectively computed by iterating $\genf{\rbrprog{}}$ starting from the empty program typed-norms mapping, as we state below.

\begin{theorem}\label{teo:lfp}
Consider a program $\rbrprog{}$ such that $\nrules{\rbrprog{}} = n$. 
Then $\mi{lfp}(\genf{\rbrprog{}})$ exists and is the supremum of the ascending Kleene chain starting from 
$\langle \epsilon_1, \ldots, \epsilon_n \rangle$.
\end{theorem}


\begin{example}[Relevant types inference]
Here we explain how the functions in Figure~\ref{fig:infer} produce the result intuitively explained in Example~\ref{ex:inference_intuition}. We consider $\rbrprog$ as the RBR program in Figure~\ref{fig:running}.
%
The process starts with an empty program mapping $\prmap_0 = \langle \epsilon,\epsilon,\epsilon,\epsilon,\epsilon,\epsilon,\epsilon\rangle$. In the first iteration $\prmap_1 = \genf{\rbrprog{}}(\prmap_0)$, when processing Rule \circled{2} we obtain that $\tint \in \prmap_1(2)(n)$ because of the call $\genG{\rbrprog{}}_2($\code{0 >= n}$, \prmap_0)$. 
In the next iteration, $\prmap_2 = \genf{\rbrprog{}}(\prmap_1)$, so that information is propagated to Rule \circled{1}, i.e., $\tint \in \prmap_2(1)(n)$, thanks to the invocation
\cbstart $\genS{\rbrprog{}}_1($\code{while_0($\langle$n, prod$\rangle$, $\langle$prod'$\rangle$)}$, \prmap_1)$. \cbend This information is collected in the set $A$ from Rule \circled{2}. 
Similarly, that information is propagated to \cbstart $\tint \in \prmap_3(6)(e)$ from the procedure call \code{fact($\langle$e$\rangle$, $\langle$prod$\rangle$)} \cbend in Line~\ref{rbr:while1_fact_invocation} by invoking $\genS{\rbrprog{}}_6($\code{fact($\langle$e$\rangle$, $\langle$prod$\rangle$)}$, \prmap_2)$. 
In the next step, the relevant type of \cbstart \code{e} \cbend is propagated to \code{l}---$\tint \in \prmap_4(6)(l)$---when processing the guard with $\genG{\rbrprog{}}_6($\code{match(l, Cons(e,l1))}$, \prmap_3)$. 
Finally, the relevant type in the input parameter \cbstart \code{l} \cbend of Rule \circled{6} is propagated to Rule \circled{4} when processing the procedure call with \cbstart $\genS{\rbrprog{}}_4($\code{while_1($\langle$l, sum$\rangle$, $\langle$sum'$\rangle$)}$, \prmap_4)$, \cbend so $\tint \in \prmap_5(4)(l)$. As the function \genf{\rbrprog{}} is monotone, all this information will be kept in the least fixed point. 

\end{example}


%

\subsection{Soundness of relevant types inference}\label{sec:inference_soundness}

The inference of relevant types previously presented obtains a set of \emph{interesting} types for every variable in every rule in the program. In order to state that those inferred types are enough to guide the traces, we need to introduce a new notion: \emph{value variation}.

\begin{definition}[Value variation]\label{def:valueVariation}
	We say that $v'$ is a \emph{variation} of $v$ wrt.\ a type $T$ (written $\variation{v}{v'}{T}$) if $v'$ results from replacing some components of $v$ with type $T$ for other components of the same type.
    For example, 
      $\variation{\tttype{Cons(1, Nil)}}{\tttype{Cons(4, Nil)}}{\inttype}$---replacing
      $1$ by $4$---and 
      $\variation{\tttype{Cons(0, Cons(1,Nil))}}{\tttype{Cons(1, Nil)}}{\tttype{IntList}}$---replacing the list completely.
\end{definition}

Using the notion of variation, we can define when a type is \emph{useful} for a variable: $T$ will be useful for a variable $x$ in the i-th rule of a program if changing the value of that variable changes the set of trace steps starting from a configuration of rule $i$. Formally:

\begin{definition}[Useful type]\label{def:usefulTypedNorm} We say that a type $T$ is \emph{useful} for variable $x$ in the i-th rule of a program $\rbrprog{}$---written $\useful{i}{T}{x}$---if there are configurations $\ar_0$ and $\ar'_0$ such that:
\begin{enumerate}
	\item $\ar_0 = \tuple{p,\bc{\cdot}\stkbc,\tv} \cdot \ar_1$, where $\bc{\cdot}\stkbc$ are statements from the i-th rule of $\rbrprog{}$
\item $x \in \dom(\tv)$
\item $\ar'_0 = \tuple{p,\bc{\cdot}\stkbc,\tv'} \cdot \ar_1$ with $\tv' = \tv[x \mapsto v]$ and $\variation{\tv(x)}{v}{T}$
\item $\traces(\ar_0) \neq \traces(\ar'_0)$
\end{enumerate}
\end{definition}

Finally, we can state the soundness result of the inference of relevant types: if changing some components of type $T$ of a variable $x$ in rule $i$ affects the possible trace steps---i.e. $\useful{i}{T}{x}$---then that type $T$ will be inferred by our process---i.e. $T \in \typednorms_i(x)$. 

\begin{theorem}[Soundness]\label{teo:soundnessInference}
If $\useful{i}{T}{x}$ then $T \in \typednorms_i(x)$. 
\end{theorem}

\section{Extensions}\label{sec:extensions}

The purpose of this section is to propose two extensions, which are very
useful in practice, to the previous analysis. This, moreover,
demonstrates that extending our framework is very easy.

\subsection{Extension to Polymorphic Types}\label{sec:polymorphic}

Let us now describe the extension of our approach to handle
polymorphic types. First, we extend the data type definition in
\rSec{sec:intermediate-representation} by adding \emph{polymorphic
  types}.

\begin{definition}[Polymorphic types]\label{def:polymorphicTypes}
  Given a countable set of type variables $V_T$, a \emph{polymorphic
    type} $T(V_T)$ can be either a monomorphic type $T$, a type
  variable $\alpha \in V_T$, or an algebraic data type
  $\DT\langle \Gamma\rangle$ defined as:
  \begin{flushleft}
    $\begin{array}{lll}
       \mi{pDd} & ::= & \mi{data}~\DT\langle \Gamma\rangle = p\cons ~[\overline{\mid p\cons{}}]\\
       p\cons   & ::= & \dcon[(\overline{T(\Gamma)})] \\
     \end{array}$
   \end{flushleft}
   where $\Gamma \subseteq V_T$ is a finite subset of variables and
   any type variable in the right-hand side of a data type definition
   must be in the finite subset of variables in the left-hand side.
\end{definition}

\begin{example}[Polymorphic list]~\label{ex:polymorphicTypes} Using
  the syntax presented in \rDef{def:polymorphicTypes} we can
  define the data type of a polymorphic list
  (\TYPE{List$\langle$A$\rangle$}) as follows:
	\begin{flushleft}
      ~~~~\lstinline@data List$\langle$A$\rangle$ = Nil | Cons(A, List$\langle$A$\rangle$)@
	\end{flushleft}
        The type \TYPE{IntList} in \rDef{def:monomorphicTypes} can
        be represented by the type \TYPE{List$\langle$Int$\rangle$},
        where the parametric type \TYPE{A} is instantiated to \tint{}.
\end{example}

The RBR program syntax in \rDef{def:syntax} can also be
extended by adding \emph{polymorphic typed procedures} in the
following way:

	\[\begin{array}{lll}
            \typedproc & ::= & 
                               p\langle\Gamma\rangle :: T(\Gamma) \times \cdots \times T(\Gamma) ~
                               \overline{r}\\
	\end{array}\]
\newcommand{\sepl}[0]{{\tiny<}}
\newcommand{\sepr}[0]{{\tiny>}}

In this context, properly typed terms must be understood as the
natural extension to polymorphic types. \rFig{fig:polrunning} contains a RBR
program with polymorphic types. In order to avoid confusion with the angles used
to represent polymorphic types, in the examples of this section we use $\sepl$
and $\sepr$ to separate input and output arguments.



\begin{figure}[tb]
\begin{minipage}{5.4cm}



\begin{lstlisting}[frame=none]
data List$\langle$A$\rangle$ =
  Nil
  | Cons(A, List$\langle$A$\rangle$)

head$\langle$A$\rangle$ :: $\sepl\TYPE{List}\langle\TYPE{A}\rangle\sepr \times \sepl\TYPE{A}\sepr$
(*\hspace{-0.85cm}$\circled{1}\hspace{0.45cm}$*)head($\sepl$l$\sepr$,$\sepl$e$\sepr$) $\rulearrow$
  match(l,Cons(e,l'))

tail$\langle$A$\rangle$ :: $\sepl\TYPE{List}\langle\TYPE{A}\rangle\sepr \times \sepl\TYPE{List}\langle\TYPE{A}\rangle\sepr$
(*\hspace{-0.85cm}$\circled{2}\hspace{0.45cm}$*)tail($\sepl$l$\sepr$,$\sepl$l'$\sepr$) $\rulearrow$
   match(l,Cons(e,l'))
\end{lstlisting}
\end{minipage}
\begin{minipage}{7cm}
\begin{lstlisting}[frame=none,firstnumber=12]
factSum :: $\sepl\TYPE{List}\langle\tint\rangle\sepr \times \sepl\tint\sepr$
(*\hspace{-0.85cm}$\circled{3}\hspace{0.45cm}$*)factSum($\sepl$l$\sepr$, $\sepl$sum'$\sepr$) $\rulearrow$ true,
  sum := 0,
  while_1($\sepl$l, sum$\sepr$, $\sepl$sum'$\sepr$)

while_1 :: $\sepl\TYPE{List}\langle\tint\rangle \times \tint\sepr \times \sepl\tint\sepr$
(*\hspace{-0.85cm}$\circled{4}\hspace{0.45cm}$*)while_1($\sepl$l, sum$\sepr$,  $\sepl$sum'$\sepr$) $\rulearrow$ match(l, Nil)
  sum' := sum
(*\hspace{-0.85cm}$\circled{5}\hspace{0.45cm}$*)while_1($\sepl$l, sum$\sepr$, $\sepl$sum'$\sepr$) $\rulearrow$
  nonmatch(l, Nil),
  head($\sepl$l$\sepr$, $\sepl$e$\sepr$),
  fact($\sepl$e$\sepr$, $\sepl$prod$\sepr$),
  sum1 := sum + prod,
  tail($\sepl$l$\sepr$, $\sepl$l1$\sepr$),
  while_1($\sepl$l1, sum1$\sepr$, $\sepl$sum'$\sepr$)
\end{lstlisting}
\end{minipage}
\caption{RBR program with polymorphic types.\label{fig:polrunning}}
\end{figure}

Analyzing RBR programs with polymorphic types can be done by
transforming them into equivalent programs with monomorphic types
only. This basically creates monomorphic versions of procedures (and
types) for every use of the polymorphic types.
Note that all possible uses (or instantiations) of the polymorphic
types can be inferred using standard type inference algorithms.
The following list of steps describes how such a transformation can be
carried out:

\begin{itemize}

\item For every polymorphic type definition $\DT\langle \Gamma\rangle$
  and every needed instantiation of the polymorphic type, we replace
  the polymorphic type definition by new copies using fresh
  monomorphic types and constructors.

\item For every polymorphic procedure $p\langle\Gamma\rangle$ and
  every needed instantiation of its polymorphic type declarations, we
  replace the polymorphic procedure by new copies using fresh
  monomorphic types. Calls to procedures are also replaced by their
  properly typed monomorphic translations.

\end{itemize}

Note that, in addition, in some situations the user might be
interested in analyzing a polymorphic procedure without providing any
calling context, i.e., the polymorphic type cannot be instantiated to
a monomorphic one in such cases. To handle these cases, for every type
variable \TYPE{A} in the program, we create a corresponding
\emph{fresh} monomorphic type \TYPE{A} together with a fresh constant
\TYPE{a}, and then use those types to instantiate the corresponding
procedures.

\begin{example}
  Consider the RBR program in \rFig{fig:polrunning}. We can analyze
  the program and extract that $\TYPE{head}$ and $\TYPE{tail}$ are
  called in the body of $\TYPE{while\_1}$ with the type $\TYPE{A}$
  instantiated to $\inttype$. Therefore, we obtain the transformed
  program presented in \rFig{fig:polrunningtr} and its abstraction in
  \rFig{fig:polabstraction}.
  Note also that we have instantiated procedures $\TYPE{head}$ and
  $\TYPE{tail}$ with respect to the auxiliary monomorphic type
  $\TYPE{A}$. This allows analyzing these procedures directly without
  considering the calling context of $\TYPE{while\_1}$.

\begin{figure}[htb]
\begin{minipage}{5.4cm}
\begin{lstlisting}[frame=none]
data A = a

data ListA =
  NilA
  (*\cbstart*)| ConsA(A, ListA)(*\cbend*)

data ListInt =
  NilInt
  | ConsInt(Int, ListInt)

headA :: $\sepl$$\TYPE{ListA}\sepr \times \sepl\TYPE{A}\sepr$
(*\hspace{-0.85cm}$\circled{1}\hspace{0.45cm}$*)headA($\sepl$l$\sepr$,$\sepl$e$\sepr$) $\rulearrow$
  (*\cbstart*)match(l,ConsA(e,l'))(*\cbend*)

headInt :: $\sepl\TYPE{ListInt}\sepr \times \sepl\TYPE{Int}\sepr$
(*\hspace{-0.85cm}$\circled{2}\hspace{0.45cm}$*)headInt($\sepl$l$\sepr$,$\sepl$e$\sepr$) $\rulearrow$
  (*\cbstart*)match(l,ConsInt(e,l'))(*\cbend*)

tailA :: $\sepl\TYPE{ListA}\sepr \times \sepl\TYPE{ListA}\sepr$
(*\hspace{-0.85cm}$\circled{3}\hspace{0.45cm}$*)tailA($\sepl$l$\sepr$,$\sepl$l'$\sepr$) $\rulearrow$
   (*\cbstart*)match(l,ConsA(e,l'))(*\cbend*)
\end{lstlisting}
\end{minipage}
\begin{minipage}{7cm}
\begin{lstlisting}[frame=none,firstnumber=22]
tailInt :: $\sepl\TYPE{ListInt}\sepr \times \sepl\TYPE{ListInt}\sepr$
(*\hspace{-0.85cm}$\circled{4}\hspace{0.45cm}$*)tailInt($\sepl$l$\sepr$,$\sepl$l'$\sepr$) $\rulearrow$
   (*\cbstart*)match(l,ConsInt(e,l'))(*\cbend*)

factSum :: $\sepl\TYPE{ListInt}\sepr \times \sepl\tint\sepr$
(*\hspace{-0.85cm}$\circled{3}\hspace{0.45cm}$*)factSum($\sepl$l$\sepr$, $\sepl$sum'$\sepr$) $\rulearrow$ true,
  sum := 0,
  while_1($\sepl$l, sum$\sepr$, $\sepl$sum'$\sepr$)

while_1 :: $\sepl\TYPE{ListInt} \times \tint\sepr \times \sepl\tint\sepr$
(*\hspace{-0.85cm}$\circled{5}\hspace{0.45cm}$*)while_1($\sepl$l, sum$\sepr$, $\sepl$sum'$\sepr$) $\rulearrow$ match(l, NilInt)
  sum' := sum
(*\hspace{-0.85cm}$\circled{6}\hspace{0.45cm}$*)while_1($\sepl$l, sum$\sepr$, $\sepl$sum'$\sepr$) $\rulearrow$
  (*\cbstart*)nonmatch(l, NilInt),(*\cbend*)
  headInt($\sepl$l$\sepr$, $\sepl$e$\sepr$),
  fact($\sepl$e$\sepr$, $\sepl$prod$\sepr$),
  sum1 := sum + prod,
  tailInt($\sepl$l$\sepr$, $\sepl$l1$\sepr$),
  while_1($\sepl$l1, sum1$\sepr$, $\sepl$sum'$\sepr$)
\end{lstlisting}
\end{minipage}
\caption{Transformed RBR program.\label{fig:polrunningtr}}
\end{figure}

\begin{figure}[tbh]
\begin{minipage}{6.25cm}
{\small
\begin{lstlisting}[frame=none]
headA($\sepl$L$_{\tttype{ListA}}$, L$_{\tttype{A}}$$\sepr$, $\sepl$E$_{\tttype{A}}$$\sepr$) $\rulearrow$
  L$_{\tttype{ListA}} = ~$L'$_{\tttype{ListA}} + 1$
  $\wedge~$L'$_{\tttype{ListA}} \geq 0~\wedge~$L$_{\tttype{A}} = $max(E$_{\tttype{A}}$,L'$_{\tttype{A}}$)
  $\wedge~$E$_{\tttype{A}} \geq 0~\wedge~$L'$_{\tttype{A}} \geq 0$,

headInt($\sepl$L$_{\tttype{ListInt}}$, L$_{\tttype{Int}}$$\sepr$, $\sepl$E$_{\tttype{Int}}$$\sepr$) $\rulearrow$
  L$_{\tttype{ListInt}} = ~$L'$_{\tttype{ListInt}} + 1$
  $\wedge~$L'$_{\tttype{ListInt}} \geq 0$
  $\wedge~$L$_{\tttype{Int}} = $max(E$_{\tttype{Int}}$,L'$_{\tttype{Int}}$),

tailA($\sepl$L$_{\tttype{ListA}}$, L$_{\tttype{A}}$$\sepr$, $\sepl$L'$_{\tttype{ListA}}$,L'$_{\tttype{A}}$$\sepr$) $\rulearrow$
  L$_{\tttype{ListA}} = ~$L'$_{\tttype{ListA}} + 1$
  $\wedge~$L'$_{\tttype{ListA}} \geq 0~\wedge~$L$_{\tttype{A}} = $max(E$_{\tttype{A}}$,L'$_{\tttype{A}}$)
  $\wedge~$E$_{\tttype{A}} \geq 0~\wedge~$L'$_{\tttype{A}} \geq 0$,

tailInt($\sepl$L$_{\tttype{ListInt}}$, L$_{\tttype{Int}}$$\sepr$, $\sepl$L'$_{\tttype{ListInt}}$,L'$_{\inttype}$$\sepr$)
  $\rulearrow$ L$_{\tttype{ListInt}} = ~$L'$_{\tttype{ListInt}} + 1$
  $\wedge~$L'$_{\tttype{ListInt}} \geq 0$
  $\wedge~$L$_{\tttype{Int}} = $max(E$_{\tttype{Int}}$,L'$_{\tttype{Int}}$),
\end{lstlisting}
}
\end{minipage}
\begin{minipage}{6.25cm}
{\small
\begin{lstlisting}[frame=none,firstnumber=20]
factSum($\sepl$L$_{\tttype{ListInt}}$, L$_{\tttype{Int}}$$\sepr$, $\sepl$Sum'$_\inttype$$\sepr$) $\rulearrow$
  $\top~\wedge$ L$_{\tttype{ListInt}} \geq 0~\wedge~$L'$_{\tttype{ListInt}} \geq 0$,
  Sum$_\inttype$ = 0,
  while_1($\sepl$L$_{\tttype{ListInt}}$, L$_{\tttype{Int}}$, Sum$_\inttype$$\sepr$,
            $\sepl$Sum'$_\inttype$$\sepr$)

while_1($\sepl$L$_{\tttype{ListInt}}$, L$_{\tttype{Int}}$, Sum$_\inttype$$\sepr$, (*\label{fig:polabstraction:while_1:1}*)
          $\sepl$Sum'$_\inttype$$\sepr$) $\rulearrow$
  L$_\inttype$ = $-\infty$ $\wedge$ L$_{\tttype{ListInt}}$ = 1
  $\wedge$ L$_{\tttype{ListInt}} \geq 0$
  Sum'$_\inttype$ = Sum$_\inttype$

while_1($\sepl$L$_{\tttype{ListInt}}$, L$_{\tttype{Int}}$, Sum$_\inttype$$\sepr$, (*\label{fig:polabstraction:while_1:2}*)
          $\sepl$Sum'$_\inttype$$\sepr$) $\rulearrow$ $\top$,
  headInt($\sepl$L$_{\tttype{ListInt}}$, L$_{\tttype{Int}}$$\sepr$,$\sepl$E$_{\tttype{Int}}$$\sepr$),
  fact($\sepl$E$_\inttype$$\sepr$, $\sepl$Prod$_\inttype$$\sepr$),
  Sum1$_\inttype$ = Sum$_\inttype$ + Prod$_\inttype$,
  tailInt($\sepl$L$_{\tttype{ListInt}}$, L$_{\tttype{Int}}$$\sepr$,$\sepl$L1$_{\tttype{ListInt}}$,L1$_{\tttype{Int}}$$\sepr$),
  while_1($\sepl$L1$_{\tttype{ListInt}}$, L1$_\inttype$, Sum1$_\inttype$$\sepr$,
            $\sepl$Sum'$_\inttype$$\sepr$)
          \end{lstlisting}
}
\end{minipage}
	\caption{Abstraction of the polymorphic RBR program.\label{fig:polabstraction}}
\end{figure}

\cbstart Notice that, in this example, the number of arguments of $\TYPE{head}$
and $\TYPE{tail}$ procedures in the different instantiated versions of
Figure~\ref{fig:polabstraction} are the same, but they can differ depending on
the instantiations of the type variables. In this way, we can keep the size
relations between input and output instantiated abstractions (this is not
possible if we have only one version of $\TYPE{head}$ and $\TYPE{tail}$ and a
fixed number of arguments). \cbend
\end{example}

It is easy to check that any evaluation step given in the original RBR
program is translated into an evaluation step in the transformed RBR
program. The rest of the analysis follows from previous sections.

Another possibility for handling polymorphism could be to analyze every method once and instantiate that information in every particular method invocation. However, this is not an straightforward approach because different uses of the same polymorphic method can involve a different number of norms whose sizes must be related. Our approach is not completely modular and cannot be applied in every scenario (for example, polymorphic recursion cannot be handled), but provides a simple and effective way of supporting standard parametric polymorphism in our setting.

\subsection{Context-Sensitive Norms}\label{sec:cs-norms}

\cbstart In this section, we propose a way to improve the precision of
typed-norms using annotations.
%
When we define a data type, we can use the same type in different positions and
for different purposes. \cbstart \cbend In such cases, if the type is not part
of the recursive data structure, we can distinguish the different uses of the
same type. The corresponding upper bounds are typically more precise and provide
insights on the complexity of processing each part of the data.
Technically, to achieve this, we \cbstart \cbend need to annotate
non-recursive types with their positions in the data type structures.

\begin{definition}[Annotated Monomorphic types]\label{def:AnnMonomorphicTypes} An
\emph{annotated monomorphic type} $T$ can be a built-in data type as $\tint{}$ or an
algebraic data type $\DT$ defined as: \begin{flushleft}
$\begin{array}{l@{~~~~~}lll} & \mi{Dd} & ::= & \mi{data}~\DT = \cons_{D}
~[\overline{\mid \cons_{D}}]\\ & \cons_{D}   & ::= & \dcon[(T_{D,\mi{Co},1},\ldots,T_{D,\mi{Co},n})] \\
& T_{D,\mi{Co},1} & ::= & \mbox{$
\begin{cases}
T & \mbox{ if $T$ is recursive in $D$}\\
T_{\mi{Co},i} & \mbox{ if $T$ is non-recursive in $D$}
\end{cases}$}
\end{array}$
\end{flushleft}
\end{definition}

By unrolling annotated type definitions as trees
we can extend \deptypes{T} to
consider the annotations in the path from the root to the types.
Figure~\ref{fig:unrolledTerm} shows the type definition $\TYPE{
IntListPair = Pair(IntList, IntList)}$ as a tree.
From these annotations we obtain $\deptypes{\tttype{IntListPair}} =
\{(\tttype{IntListPair}, \Lambda), (\tttype{IntList}, (\tttype{Pair},1)), (\tttype{Int},(\tttype{Pair},1) \cdot
(\tttype{Cons},1)), (\tttype{IntList}, (\tttype{Pair},2)),\allowbreak
(\tttype{Int},(\tttype{Pair},2) \cdot (\tttype{Cons},1))\}$, where
each element of the set is formed by a pair (type, chain of positions) that
represents its annotations in the tree, and the empty chain is represented by
$\Lambda$. Notice that, if we do not have duplicated non-recursive types, we
obtain the same number of norms (but annotated).

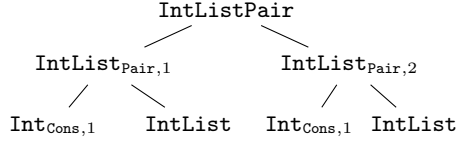
\begin{figure}
\begin{center}
\begin{tikzpicture}[node distance = 6cm, auto]

\node [] (tH) {\tttype{IntListPair}};
\node [below left = 0.5cm and 0.6cm] (tH1) {\tttype{IntList}$_{\tttype{Pair},1}$};
\node [below left = 1.3cm and 1.6cm] (tH11) {\tttype{Int}$_{\tttype{Cons},1}$};
\node [below right = 1.3cm and -1.2cm] (tH12) {\tttype{IntList}};
\node [below right = 0.5cm and 0.6cm] (tH2) {\tttype{IntList}$_{\tttype{Pair},2}$};
\node [below right = 1.3cm and 0.4cm] (tH21) {\tttype{Int}$_{\tttype{Cons},1}$};
\node [below right = 1.3cm and 1.8cm] (tH22) {\tttype{IntList}};

\draw [] (tH) [bend right=0] to (tH1) ;
\draw [] (tH) [bend right=0] to (tH2) ;
\draw [] (tH1) [bend right=0] to (tH11) ;
\draw [] (tH1) [bend right=0] to (tH12) ;
\draw [] (tH2) [bend right=0] to (tH21) ;
\draw [] (tH2) [bend right=0] to (tH22) ;
\end{tikzpicture}\\[0.2cm]
\end{center}
\caption{Unrolled type definition tree view of type \tttype{IntListList}}
\label{fig:unrolledTerm}
\end{figure}

Symbolic typed norms are extended to deal with annotated types. The
intuition in $\typednorm{t}{(T,A)}$ is that we traverse the annotated data
structure using $A$ and when the annotation is empty ($\Lambda$), we have
reached the type-norm we want to compute and we can use a definition similar to
Definition~\ref{def:symb-typed-norms}.

\begin{definition}[Symbolic annotated typed-norms]\label{def:ann-symb-typed-norms}
The symbolic typed-norm to compute the size of a term $t$ (possibly with variables) regarding a type $T$ is defined as $\typednorm{t}{(T,A)} = \typednormann{t}{(T,A)}{A}$ where:
	\[
	\typednormann{t}{(T,A)}{\Lambda} =
	\begin{cases}
	X_{(T,A)} & \mbox{ if } t \equiv x \mbox{ and } {(T,A)} \in \deptypes{\type{x}}\\
    - \infty  & \mbox{ if } t \equiv x, T = \tint \mbox{ and } (\tint,A) \notin \deptypes{\type{x}}\\

	n & \mbox{ if } t \equiv n \mbox{ and } T = \inttype\\
	\typednormann{e_1}{(T,A)}{\Lambda} \pm \typednormann{e_2}{(T,A)}{\Lambda} & \mbox{ if } t \equiv e_1 \pm e_2 \mbox{ and } T=\inttype 
	\\
%
%
	1 + \sum_{i=1}^n \typednormann{t_i}{(T,A)}{\Lambda} & \mbox{ if } t = \DC(t_1,\ldots, t_n) \mbox{ and } \type{t} = T\\
	\mi{max}_{i=1}^n \typednormann{t_i}{(T,A)}{\Lambda} & \mbox{ if } t = \DC(t_1,\ldots, t_n) \mbox{ and } \type{t} \neq T\\
	0 & \mbox{ in other case} \\
	\end{cases}
	\]
  \[
	\typednormann{t}{(T,A)}{(\DC,i) \cdot \mi{rest}} =
	\begin{cases}
	\typednormann{t_i}{(T,A)}{\mi{rest}} & \mbox{ if } t = \DC(t_1,\ldots, t_n)\\
	- \infty & \mbox{ if } t \neq \DC(t_1,\ldots, t_n) \mbox{ and } T = \tint\\
	0 & \mbox{ in other case} \\
	\end{cases}
	\]

\end{definition}

By using $\typednorm{t}{(T,A)}$, we obtain a sound abstraction in the sense of
Theorem~\ref{teo:soundness}, and the relevant type inference presented in Section~\ref{sec:inference} could be easily adapted to consider
annotated type-norms.
In the following, we present an example that shows how this technique is applied
in practice.

\begin{example}
  Consider the program in \rFig{fig:cspolrunning} (left)
  where the rule $\TYPE{lengthP}$ takes a non-recursive
  $\TYPE{IntListPair}$ structure and computes the length of the second
  list if the first list is empty; otherwise, returns the length of
  the first list. Since both lists are of the same type
  ($\TYPE{IntList}$), we are interested in considering
  their norms separately to obtain a better upper
  bound. \rFig{fig:cspolrunning} \cbstart (right) \cbend shows this abstraction\footnote{For simplicity, we use $\TYPE{IntList}_1$
  and $\TYPE{IntList}_2$ instead of $\TYPE{IntList}_{(\TYPE{Pair},1)}$
  and $\TYPE{IntList}_{(\TYPE{Pair},2)}$, and  $\TYPE{Int}_1$
  and $\TYPE{Int}_2$ instead of $\TYPE{Int}_{(\TYPE{Pair},1) \cdot (\TYPE{Cons},1)}$
  and $\TYPE{Int}_{(\TYPE{Pair},2) \cdot (\TYPE{Cons},1)}$.}.
  It basically distinguishes the uses of $\TYPE{IntList}$ by their position in
  the definition of $\TYPE{IntListPair}$---note the use of $\TYPE{IntList}_1$
  and $\TYPE{IntList}_2$.
  Using
  this improvement, we can obtain the upper bound
  $1+max(6+6*\TYPE{P}_{\TYPE{IntList}_1},5+6*\TYPE{P}_{\TYPE{IntList}_2})$
  for the function
  $\TYPE{lengthP}(\sepl\TYPE{P}_{\tttype{IntListPair}},
  \TYPE{P}_{\tttype{IntList}_1},\TYPE{P}_{\tttype{IntList}_2},\TYPE{P}_{\tttype{Int}_1},\TYPE{P}_{\tttype{Int}_2}\sepr,\sepl\TYPE{Res}_\inttype\sepr)$
  instead of $7+6*\TYPE{P}_{\tttype{IntList}}$
  $\TYPE{lengthP}(\sepl\TYPE{P}_{\tttype{IntListPair}},
  \TYPE{P}_\tttype{IntList},\TYPE{P}_\tttype{Int}\sepr,$ $\sepl\TYPE{Res}_\inttype\sepr)$. Note that in this case
  $\TYPE{P}_{\tttype{IntList}} = \TYPE{P}_{\tttype{IntList}_1} +
  \TYPE{P}_{\tttype{IntList}_2}$.

\begin{figure}[tbp]
\begin{minipage}{6cm}
{\small
    \begin{lstlisting}[frame=none]
data IntListPair = Pair(IntList, IntList)

lengthP :: $\TYPE{IntListPair} \times \tint$
(*\hspace{-0.85cm}$\circled{1}\hspace{0.45cm}$*)lengthP(p, res) $\rulearrow$
  match(p,Pair(Nil,l2)),
  length(l2,res)
(*\hspace{-0.85cm}$\circled{2}\hspace{0.45cm}$*)lengthP(p, res) $\rulearrow$
  nomatch(p,Pair(Nil,l2)),
  match(p,Pair(l1,l2)),
  length(l1,res)
\end{lstlisting}
}
\end{minipage}
\begin{minipage}{6.5cm}
{\small
  \begin{lstlisting}[frame=none,firstnumber=11]
lengthP($\sepl$P$_{\tttype{IntListPair}}$, P$_{\tttype{IntList}_1}$, P$_{\tttype{IntList}_2}$, P$_{\tttype{Int}_1}$, P$_{\tttype{Int}_2}\sepr$,$\sepl$Res$_\inttype\sepr$) $\rulearrow$
  P$_{\tttype{IntList}_1} = 1$
  $\wedge$ P$_{\tttype{IntList}_1} \geq 0$
  $\wedge$ P$_{\tttype{IntList}_2} \geq 0$,
  length($\sepl$P$_{\tttype{IntList}_2}$, P$_{\inttype_2}\sepr$, $\sepl$Res$_\inttype\sepr$)

lengthP($\sepl$P$_{\tttype{IntListPair}}$, P$_{\tttype{IntList}_1}$, P$_{\tttype{IntList}_2}$, P$_{\tttype{Int}_1}$, P$_{\tttype{Int}_2}\sepr$,$\sepl$Res$_\inttype\sepr$) $\rulearrow$
  P$_{\tttype{IntList}_1} > 1$
  $\wedge$ P$_{\tttype{IntList}_1} \geq 0$
  $\wedge$ P$_{\tttype{IntList}_2} \geq 0$,
  length($\sepl$P$_{\tttype{IntList}_1}$, P$_{\inttype_1}\sepr$, $\sepl$Res$_\inttype\sepr$)
\end{lstlisting}
}
\end{minipage}
\caption{Abstraction of RBR program with context-sensitive norms.\label{fig:cspolrunning}}
\end{figure}
\end{example}
\cbend

The proposed approach using annotated types to handle context-sensitive norms  generates abstract programs without any type information where every parameter represents an integer value. Therefore, these abstract programs can be solved automatically to upper and lower resource bounds using the existing techniques with no additional adaptations.


\section{Experiments}\label{sec:experiments}

\begin{figure}[tbp]
\begin{tabular}{L{2.05cm}rrrrr}
\toprule
\multirow{2}{*}{\bf Method} &  \multirow{2}{*}{\bf TS} & \multicolumn{2}{c}{\bf TN sum} & \multicolumn{2}{c}{\bf TN max}\\  \cmidrule(r){3-4}\cmidrule(r){5-6}
& & \small Inference & \small w/o Inference & \small Inference & \small w/o Inference 
\\ \midrule
\begin{filecontents*}{examples/RAML_Ciao_cost.csv}
method,tsMEAN,tsMEDIAN,tnsiMEAN,tnsiMEDIAN,tnsMEAN,tnsMEDIAN,tnmiMEAN,tnmiMEDIAN,tnmMEAN,tnmMEDIAN
append,119.2,121,34,32,34,32,34,32,34,32
appendAll2,41504.5,41809,521.8,535,442.3,451,3082.8,3189,3082.8,3189
coupled,92,92,95.2,84,95.2,84,95.2,84,95.2,84
dyade,3714.6,3601.5,323.4,273,323.4,273,323.4,273,323.4,273
eratos,6634,6424,472,387,472,387,472,387,472,387
factSum,564.6,573,-,-,-,-,67,63,67,63
fib,627,627,1235,467,1235,467,1235,467,1235,467
isort,4048.8,3912.5,341.6,285,341.6,285,341.6,285,341.6,285
isortlist,2392445.7,2100497,6857.1,7214.5,3807.9,3999.5,13816.4,14444,13816.4,14444
listfact,564.6,573,-,-,-,-,67,63,67,63
listnum,48,48,49.4,44.5,49.4,44.5,49.4,44.5,49.4,44.5
minSort,6694.8,6485.5,1009,835,1009,835,546.4,457,546.4,457
nub,1656114.9,1454705,3285,3492,3352.8,3571.5,-,-,-,-
partition,244,247.5,69,65,69,65,69,65,69,65
traverse1,1275,1243,81.4,83,81.4,83,-,-,-,-
traverse2,1266,1234,81.4,83,81.4,83,-,-,-,-
zip3,286.8,291,85.8,81,91.8,87,85.8,81,94.8,90
\end{filecontents*}
\csvreader[head to column names]{examples/RAML_Ciao_cost.csv}{}{\small \method & $\tsMEDIAN$ & $\tnsiMEDIAN$ & $\tnsMEDIAN$ & $\tnmiMEDIAN$ & $\tnmMEDIAN$\\}
\\[-0.45cm]  \bottomrule
\end{tabular}
\caption{Median number of steps when evaluating the upper bounds for the list manipulation examples using 10 random input parameters.\label{fig:table_RAML_Ciao_ub}}
\end{figure}

\begin{figure}[tbp]
\begin{tabular}{L{2.85cm}rrrrr}
\toprule
\multirow{2}{*}{\bf Method} & \multirow{2}{*}{\bf TS} & \multicolumn{2}{c}{\bf TN sum} & \multicolumn{2}{c}{\bf TN max}\\  \cmidrule(r){3-4}\cmidrule(r){5-6}
& & \small Inference & \small w/o Inference & \small Inference & \small w/o Inference 
\\ \midrule
\begin{filecontents*}{examples/Chat_cost.csv}
method,tsMEAN,tsMEDIAN,tnsiMEAN,tnsiMEDIAN,tnsMEAN,tnsMEDIAN,tnmiMEAN,tnmiMEDIAN,tnmMEAN,tnmMEDIAN
ButtonImpl.press,761.6,725,227.5,217,227.5,217,227.5,217,227.5,217
ButtonImpl.registerListener,118.1,113,41.5,40,41.5,40,41.5,40,41.5,40
ClientGUIImpl.init2,151.1,146,74.5,73,74.5,73,74.5,73,74.5,73
ClientImpl.getGUI,172.1,167,95.5,94,95.5,94,95.5,94,95.5,94
ClientImpl.receive,160.3,154,65.8,64,65.8,64,65.8,64,128.8,124
Main.main,20339.8,17831,5006.2,4687,4761.2,4372,5006.2,4687,4003.4,3859
ServerImpl.connect,176.3,170,81.8,80,81.8,80,81.8,80,119.6,116
ServerImpl.sessionClosed,140.3,134,45.8,44,45.8,44,45.8,44,45.8,44
SessionImpl.close,144.3,138,49.8,48,49.8,48,49.8,48,49.8,48
SessionImpl.init2,165.3,159,70.8,69,70.8,69,70.8,69,108.6,105
UserImpl.use,786.6,750,252.5,242,252.5,242,252.5,242,252.5,242
\end{filecontents*}
\csvreader[head to column names]{examples/Chat_cost.csv}{}{\small \method & $\tsMEDIAN$ & $\tnsiMEDIAN$ & $\tnsMEDIAN$ & $\tnmiMEDIAN$ & $\tnmMEDIAN$\\}
\\[-0.45cm] \bottomrule
\end{tabular}
\caption{Median number of steps when evaluating the upper bounds for Chat methods using 10 random input parameters.\label{fig:table_Chat_ub}}
\end{figure}

\begin{figure}[tbp]
\begin{tabular}{L{2.85cm}rrrrr}
\toprule
\multirow{2}{*}{\bf Method} &  \multirow{2}{*}{\bf TS} & \multicolumn{2}{c}{\bf TN sum} & \multicolumn{2}{c}{\bf TN max}\\  \cmidrule(r){3-4}\cmidrule(r){5-6}
& & \small Inference & \small w/o Inference & \small Inference & \small w/o Inference 
\\ \midrule
\begin{filecontents*}{examples/ETICS_cost.csv}
method,tsMEAN,tsMEDIAN,tnsiMEAN,tnsiMEDIAN,tnsMEAN,tnsMEDIAN,tnmiMEAN,tnmiMEDIAN,tnmMEAN,tnmMEDIAN
Solver.bestSolution,2,2,2,2,2,2,2,2,2,2
Solver.contains,850,840.5,126.1,118.5,126.1,118.5,126.1,118.5,126.1,118.5
Solver.randomMap,23812.8,23873.5,549.8,515,549.8,515,549.8,515,549.8,515
Solver.utility,111319.5,95160,1444.4,1250,6850.4,6076,1444.4,1250,\mathit{timeout},\mathit{timeout}
\end{filecontents*}
\csvreader[head to column names]{examples/ETICS_cost.csv}{}{\small \method & $\tsMEDIAN$ & $\tnsiMEDIAN$ & $\tnsMEDIAN$ & $\tnmiMEDIAN$ & $\tnmMEDIAN$\\}
\\[-0.45cm] \bottomrule
\end{tabular}
\caption{Median number of steps when evaluating the upper bounds for ETICS methods using 10 random input parameters.\label{fig:table_ETICS_ub}}
\end{figure}


We have implemented our approach in the SACO
system~\cite{AlbertAGGP11,AlbertAFGGMPR14}. \cbstart SACO is a Static Analyzer
for Concurrent Objects that in addition to the resource analyzer also
includes a deadlock analyzer. Although in this paper we have focused
on sequential programs, concurrent objects \cite{agha93actorspace}
provide a formalism to model concurrent and distributed systems using
the ABS language mentioned in Section~\ref{sec:intermediate-representation}. SACO carries out the
transformation of ABS programs into the intermediate representation
defined in Section~\ref{sec:intermediate-representation} so that new
extensions and analyses can be integrated in SACO at this point in the
implementation.  The whole system is implemented in Prolog and can be
used from an online web interface at
\url{https://costa.fdi.ucm.es/saco/web/}, where in addition to the
examples that will be described in this section, the user can type her
own programs and run the available analyses. \cbend

In this section we perform an experimental comparison of the different norms to evaluate both the precision of the obtained upper bounds and the time needed by our resource analysis framework to obtain them.
In order to create a set of benchmarks as complete as possible, we have collected tests from different sources. First, we have taken some examples from the Resource Aware ML~\cite{Hoffmann2017,HAH12,Hofmann03} (RAML) set of examples, which are also used in the resource usage analysis using sized types presented in~\cite{SerranoLH14}. These examples are mainly focused on list manipulation, and include: 
\begin{itemize}
\item \emph{append} and \emph{appendAll2} for appending lists and compound lists,
\item \emph{coupled}, a pair of mutually recursive functions that generate lists
of a given size,
\item \emph{dyade} for combining and multiplying two lists,
\item \emph{eratos} for detecting prime numbers in a range using the \emph{sieve of Eratosthenes},
\item \emph{fib} for computing the $n^{th}$ number of the Fibonacci sequence,
\item \emph{isort}, \emph{isortlist} and \emph{minSort}, for sorting integer lists and lists of integer lists,
\item \emph{listnum}, that creates a list of decreasing integer numbers,
\item \emph{nub} for removing duplicates in a list of lists,
\item \emph{partition} for splitting a list given a pivot element, and
\item \emph{zip3} for combining 3 lists in a list of triples.
\end{itemize}
Additionally, we have included in this first set of list manipulation examples the \emph{factSum} method presented in Figure~\ref{fig:running}, the \emph{listfact} running example used in~\cite{SerranoLH14} for adding all the factorial numbers in a list, and two extra methods \emph{traverse1} and \emph{traverse2} for traversing lists of lists.

As the second set of examples we have used 11 different methods of the \emph{Chat} program, a well-known distributed ABS~\cite{johnsen10fmco} program used for testing and evaluating static analyses. This program provides a chat server that allows different clients to exchange messages. Finally, we have also considered a third set of 4 methods from the industrial case study ETICS\footnote{\url{http://envisage-project.eu/wp-content/uploads/2015/04/D4_4_1r.pdf}}. This case study models in ABS the process of exploiting on-demand virtual machines to satisfy service requests trying to maximize the profit considering cost and penalty terms specified in service level agreements (SLAs). The complete code of all the examples used in this section
can be found in \url{https://costa.fdi.ucm.es/saco/web/} inside the ``\emph{Typed-Norms}'' folder. 

Figures~\ref{fig:table_RAML_Ciao_ub}--\ref{fig:table_ETICS_ub} contain upper bounds on the number of steps obtained by our resource analysis framework with different norms and enabling/disabling the relevant types inference algorithm presented in Section~\ref{sec:inference}.
The resource analysis framework obtains upper bounds as arithmetic expressions on the input parameters of the methods, which are difficult to compare (when possible). In order to easily compare the upper bounds obtained for the different norms, we have evaluated them using 10 sets of random input parameters.\footnote{Note that we evaluate the obtained upper bounds, so in this comparison we are considering the worst case of every method for different random sizes of its parameters.} This approach generates 10 different numeric values for each method and norm, so in the Figures~\ref{fig:table_RAML_Ciao_ub}--\ref{fig:table_ETICS_ub} we have included the median of those numbers.\footnote{We have used the median value instead of the average to avoid any influence of extreme outliers. These values may appear, as we are evaluating upper bounds using random input parameters.} If the framework is not able to obtain an upper bound for a method using a particular norm we mark it with the symbol ``--'', and if the framework requires more than 2 hours to obtain the upper bound we use the term ``\emph{timeout}''. These figures contain a row for each method, whose upper bounds have been computed using the different norms presented in Definitions~\ref{def:term-size} and~\ref{def:typed-norms}: term-size norm $\tsnorm{\cdot}$ (\textbf{TS} in the figures), typed-norm $\typednormorig{\cdot}{T}$ (\textbf{TN sum} in the figures), and typed-norm $\typednorm{\cdot}{T}$ (\textbf{TN max} in the figures). In the case of typed-norms, the upper bound has been computed inferring the useful types (column \emph{Inference}) and without inferring that information (column \emph{w/o Inference}).

The main result extracted from Figures~\ref{fig:table_RAML_Ciao_ub}--\ref{fig:table_ETICS_ub} is that, for all  methods, the resource analysis framework obtains more precise results using typed-norms ($\typednormorig{\cdot}{T}$ or $\typednorm{\cdot}{T}$) than using term-size ($\tsnorm{\cdot}$). Concretely, for the list manipulation examples in Figure~\ref{fig:table_RAML_Ciao_ub} we observe that the ratio $\frac{\textnormal{steps typed-norm}}{\textnormal{steps term-size}}$ ranges from $0.002$ (\emph{isortlist}, \emph{nub}) to $0.93$ (\emph{listnum}), although the vast majority of the examples show a very small ratio: considering the 17 methods, the $53\%$ have a ratio smaller than $0.1$ and $82\%$  have a ratio smaller than $0.3$.
As shown in Figure~\ref{fig:table_Chat_ub}, the examples of the \emph{Chat} program obtain more precise upper bounds than using term-size, but the gain is not as relevant as in the list manipulation examples. For these methods, the ratio ranges between $0.22$ for \emph{Main.main} and $0.80$ for \emph{ClientImpl.receive}. Finally, Figure~\ref{fig:table_ETICS_ub} shows improved upper bounds with ratio $0.01$ for \emph{Solver.utility}, $0.02$ for \emph{Solver.randomMap} and $0.14$ for \emph{Solver.contains}. In the case of \emph{Solver.bestSolution}, there is no improvement because the method has a constant cost which does no involve traversing any data structure, so in this case the chosen norm is not relevant for the upper bound.

Another important result extracted from Figures~\ref{fig:table_RAML_Ciao_ub}--\ref{fig:table_ETICS_ub} is that, for the majority of the methods, the upper bounds obtained using 
$\typednormorig{\cdot}{T}$ and $\typednorm{\cdot}{T}$ are equal. 
There are cases where one typed-norm definition allows  obtaining an upper bound, whereas the other definition cannot (namely in \emph{factSum}, \emph{listfact}, \emph{nub}, \emph{traverse1}, and \emph{traverse2} from Figure~\ref{fig:table_RAML_Ciao_ub}). This is not a surprising result, as depending on the manipulation performed on a data structure one typed-norm can detect a decreasing of the size whereas the other cannot, which has a dramatic impact on the obtained upper bound. However, this cannot be seen as a practical limitation because our resource analysis framework can support term-size norm and both types of typed-norms, so the user can choose at any moment which norm to use.

A surprising result observed in Figures~\ref{fig:table_RAML_Ciao_ub}--\ref{fig:table_ETICS_ub} is that, in some cases like \emph{appendAll2}, \emph{isortlist}, \emph{nub}, \emph{zip3}, \emph{ClientImpl.receive}, \emph{Main.main}, or \emph{Solver.utility}, the upper bound obtained for a typed-norm is different with or without the inference enabled. From the theoretical point of view both versions have the same upper bounds, but the actual implementation of the cost relation solver can produce different upper bounds. Concretely, we use CoFloCo~\cite{DBLP:conf/aplas/Flores-MontoyaH14,Flores-Montoya16}, a compositional solver 
for programs with complex execution flow and multi-dimensional ranking functions. 
This solver can find different ranking functions for recurrences (loops), and in those cases it selects one of them to generate the final upper bound. For this selection step, it takes into account which ranking functions can be maximized in terms of the input parameters, and it uses some heuristics in case of ties. When we remove some variables thanks to the inference of typed-norms, the set of ranking functions found by CoFloCo for a given recurrence can be different, altering the selected ranking function and therefore the final upper bound.

\begin{figure}[tbp]
\begin{tabular}{lrrrrr}
\toprule
\multirow{2}{*}{\bf Method} &  \multirow{2}{*}{\bf TS} & \multicolumn{2}{c}{\bf TN sum} & \multicolumn{2}{c}{\bf TN max}\\  \cmidrule(r){3-4}\cmidrule(r){5-6}
& & \small Inference & \small w/o Inference & \small Inference & \small w/o Inference 
\\ \midrule
\begin{filecontents*}{examples/RAML_Ciao_times.csv}
method,ts,tnsi,tns,tnmi,tnm
append,31.0,32.0,62.0,46.0,31.0
appendAll2,125.0,282.0,281.0,172.0,187.0
coupled,47.0,47.0,62.0,16.0,47.0
dyade,110.0,62.0,110.0,63.0,110.0
eratos,125.0,109.0,140.0,93.0,141.0
factSum,140.0,-,-,47.0,159.0
fib,31.0,31.0,31.0,16.0,31.0
isort,94.0,94.0,125.0,110.0,94.0
isortlist,375.0,937.0,797.0,359.0,375.0
listfact,46.0,-,-,62.0,109.0
listnum,32.0,31.0,47.0,31.0,46.0
minSort,109.0,109.0,140.0,125.0,172.0
nub,594.0,531.0,647.0,-,-
partition,78.0,78.0,110.0,93.0,125.0
traverse1,62.0,63.0,62.0,-,-
traverse2,31.0,47.0,47.0,-,-
zip3,94.0,78.0,141.0,62.0,125.0
\end{filecontents*}
\csvreader[head to column names]{examples/RAML_Ciao_times.csv}{}{\method & $\ts$ & $\tnsi$ & $\tns$ & $\tnmi$ & $\tnm$\\}
\\[-.45cm] \bottomrule
\end{tabular}
\caption{Time (milliseconds) to obtain upper bounds for the list manipulation examples\label{fig:table_RAML_times}}
\end{figure}

\begin{figure}[tbp]
\begin{tabular}{L{2.85cm}rrrrr}
\toprule
\multirow{2}{*}{\bf Method} &  \multirow{2}{*}{\bf TS} & \multicolumn{2}{c}{\bf TN sum} & \multicolumn{2}{c}{\bf TN max}\\  \cmidrule(r){3-4}\cmidrule(r){5-6}
& & \small Inference & \small w/o Inference & \small Inference & \small w/o Inference 
\\ \midrule
\begin{filecontents*}{examples/Chat_times.csv}
method,ts,tnsi,tns,tnmi,tnm
ButtonImpl.press,63.0,63.0,157.0,63.0,172.0
ButtonImpl.registerListener,31.0,47.0,63.0,31.0,125.0
ClientGUIImpl.init2,62.0,47.0,109.0,47.0,141.0
ClientImpl.getGUI,62.0,62.0,140.0,78.0,157.0
ClientImpl.receive,63.0,94.0,109.0,78.0,140.0
Main.main,765.0,375.0,875.0,484.0,750.0
ServerImpl.connect,141.0,47.0,141.0,63.0,188.0
ServerImpl.sessionClosed,109.0,47.0,94.0,63.0,125.0
SessionImpl.close,47.0,32.0,110.0,63.0,94.0
SessionImpl.init2,94.0,63.0,125.0,63.0,109.0
UserImpl.use,78.0,62.0,172.0,78.0,156.0
\end{filecontents*}
\csvreader[head to column names]{examples/Chat_times.csv}{}{\method & $\ts$ & $\tnsi$ & $\tns$ & $\tnmi$ & $\tnm$\\}
\\[-.45cm] \bottomrule
\end{tabular}
\caption{Time (milliseconds) to obtain upper bounds for Chat methods\label{fig:table_Chat_times}}
\end{figure}

\begin{figure}[tbp]
\begin{tabular}{L{2.87cm}rrrrr}
\toprule
\multirow{2}{*}{\bf Method} &  \multirow{2}{*}{\bf TS} & \multicolumn{2}{c}{\bf TN sum} & \multicolumn{2}{c}{\bf TN max}\\  \cmidrule(r){3-4}\cmidrule(r){5-6}
& & \small Inference & \small w/o Inference & \small Inference & \small w/o Inference 
\\ \midrule
\begin{filecontents*}{examples/ETICS_times.csv}
method,ts,tnsi,tns,tnmi,tnm
Solver.bestSolution,47.0,30.0,31.0,31.0,31.0
Solver.contains,360.0,265.0,469.0,234.0,500.0
Solver.randomMap,359.0,328.0,2000.0,312.0,2063.0
Solver.utility,3719.0,4125.0,444109.0,3937.0,\mathit{timeout}
\end{filecontents*}
\csvreader[head to column names]{examples/ETICS_times.csv}{}{\method & $\ts$ & $\tnsi$ & $\tns$ & $\tnmi$ & $\tnm$\\}
\\[-.45cm] \bottomrule
\end{tabular}
\caption{Time (milliseconds) to obtain upper bounds for ETICS methods\label{fig:table_ETICS_times}. Executions were aborted after 2 hours.}
\end{figure}

Another dimension we have measured in our benchmarks is the time needed by our resource analysis framework to obtain upper bounds using the different norms. Figures~\ref{fig:table_RAML_times}--\ref{fig:table_ETICS_times} show these results for the 3 sets of examples, measuring time in milliseconds in a CPU Intel$^{\small \textregistered}$ i5-7300U CPU with 8 GB of memory. 
As before, 
we use ``--'' to denote that the framework is not able to obtain the upper bound and ``\emph{timeout}'' if the framework does not finish in 2 hours. 
In general, obtaining upper bounds using typed-norms requires more time than using the term-size norm. 
\cbstart 
This fact was expected because of two causes: the additional stage for inferring typed-norms and the increment in variables when using typed-norms. Although the data-flow analysis is not very expensive, for small programs where solving is very fast the inference can require a similar amount of time, therefore significantly increasing the overall time. However, independently of the typed-norm inference stage, using typed-norms implies that one data-structure can be measured using different sizes, which results in more than one variable per structure in the set of cost relations that the solver processes. The time needed by the solver is proportional to the complexity of the cost relations (i.e., the number of relations and the amount of variables), so using typed-norms usually increases the overall time of the analysis. 
\cbend
As shown in Figures~\ref{fig:table_RAML_times}--\ref{fig:table_ETICS_times}, the increase of time wrt.\ term-size norms is not very pronounced and this extra time greatly compensates the gain in precision. However, there are also some situations (for example in \emph{coupled}, \emph{eratos}, \emph{zip3}, \emph{ClientGUIImpl.init2}, \emph{ServerImpl.connect}, \emph{ServerImpl.sessionClosed}, and \emph{Solver.contains}, among others) where the time required to obtain upper bounds using typed-norms with inference is faster than using term-size. The explanation of this behavior is that the type-norm inference can detect that \cbstart a parameter has not relevant types (i.e., its value is not involved in any recurrence), so that parameter disappears from the cost relations, \cbend reducing the number of variables. As explained before, the resulting set of cost relations will be simpler and the solver will require less time to obtain the upper bound. 
Similarly, we also notice that using the typed-norm inference in general produces better times than avoiding this step. Although this requires an extra analysis, the reduction on the complexity (number of variables) of the resulting cost relations makes the resolution stage faster, producing smaller overall times. An extreme example of this situation is the method \emph{Solver.utility}, where not applying inference \cbstart is \cbend 120 times slower (for $\typednormorig{\cdot}{T}$) or reaches the time limit of 2 hours (for $\typednorm{\cdot}{T}$).

\subsection{Comparison to RAML and Sized Types}

\begin{figure}[tb]
\begin{tabular}{llll}
\toprule
\textbf{Method} & RAML & Sized Types & Typed-norms \\
\midrule
\code{append(L1,L2)} & \code{L1}$_{\small\TYPE{il}}$ & \code{L1}$_{\small\TYPE{il}}$ & \code{L1}$_{\small\TYPE{il}}$\\
\code{appendAll2(L)} & \code{L}$_{\small\TYPE{il}}\cdot$\code{L}$_{\small\TYPE{ill}}\cdot$\code{L}$_{\small\TYPE{illl}}$ & \code{L}$_{\small\TYPE{il}}\cdot$\code{L}$_{\small\TYPE{ill}}\cdot$\code{L}$_{\small\TYPE{illl}}$ & \code{L}$_{\small\TYPE{il}}\cdot$\code{L}$_{\small\TYPE{ill}}\cdot$\code{L}$_{\small\TYPE{illl}}$ \\
\code{coupled(N)} & \code{N}$_{\small\TYPE{i}}$ & \code{N}$_{\small\TYPE{i}}$ & \code{N}$_{\small\TYPE{i}}$ \\
\code{dyade(L1,L2)} & \code{L1}$_{\small\TYPE{il}}\cdot$\code{L2}$_{\small\TYPE{il}}$ & \code{L1}$_{\small\TYPE{il}}\cdot$\code{L2}$_{\small\TYPE{il}}$ &
\code{L1}$_{\small\TYPE{il}}\cdot$\code{L2}$_{\small\TYPE{il}}$ \\
\code{eratos(L)} & \code{L}$_{\small\TYPE{il}}\cdot$\code{L}$_{\small\TYPE{il}}$ & \code{L}$_{\small\TYPE{il}}\cdot$\code{L}$_{\small\TYPE{il}}$ &
\code{L}$_{\small\TYPE{il}}\cdot$\code{L}$_{\small\TYPE{il}}$ \\
\code{factSum(L)} & infeasible &   \code{L}$_{\small\TYPE{il}}\cdot$\code{L}$_{\small\TYPE{i}}$ & \code{L}$_{\small\TYPE{il}}\cdot$\code{L}$_{\small\TYPE{i}}$\\
\code{fib(N)} & infeasible & \cbstart $\varphi^{{\small\TYPE{N}}_{\small\TYPE{i}}}$ \cbend & $2^{{\small\TYPE{N}}_{\small\TYPE{i}}}$\\
\code{hanoi(N,A,B,C)} & infeasible & $2^{{\small\TYPE{N}}_{\small\TYPE{i}}}$ & unknown\\
\code{isort(L)} & \code{L}$_{\small\TYPE{il}}\cdot$\code{L}$_{\small\TYPE{il}}$ & 
\code{L}$_{\small\TYPE{il}}\cdot$\code{L}$_{\small\TYPE{il}}$ & \code{L}$_{\small\TYPE{il}}\cdot$\code{L}$_{\small\TYPE{il}}$ \\
\code{isortlist(L)} & \code{L}$_{\small\TYPE{ill}}\cdot$\code{L}$_{\small\TYPE{ill}}\cdot$\code{L}$_{\small\TYPE{il}}$ & \code{L}$_{\small\TYPE{ill}}\cdot$\code{L}$_{\small\TYPE{ill}}\cdot$\code{L}$_{\small\TYPE{il}}$ & \code{L}$_{\small\TYPE{ill}}\cdot$\code{L}$_{\small\TYPE{ill}}\cdot$\code{L}$_{\small\TYPE{il}}$\\
\code{listfact(L)} & unknown & \code{L}$_{\small\TYPE{il}}\cdot$\code{L}$_{\small\TYPE{i}}$ & \code{L}$_{\small\TYPE{il}}\cdot$\code{L}$_{\small\TYPE{i}}$\\
\code{listnum(N)} & unknown & \code{N}$_{\small\TYPE{i}}$ & \code{N}$_{\small\TYPE{i}}$\\
\code{minSort(L)} & \code{L}$_{\small\TYPE{il}}\cdot$\code{L}$_{\small\TYPE{il}}$ & \code{L}$_{\small\TYPE{il}}\cdot$\code{L}$_{\small\TYPE{il}}$ &\code{L}$_{\small\TYPE{il}}\cdot$\code{L}$_{\small\TYPE{il}}$ \\
\code{nub(L)} & \code{L}$_{\small\TYPE{ill}}\cdot$\code{L}$_{\small\TYPE{ill}}\cdot$\code{L}$_{\small\TYPE{il}}$ & \code{L}$_{\small\TYPE{ill}}\cdot$\code{L}$_{\small\TYPE{ill}}\cdot$\code{L}$_{\small\TYPE{il}}$ & \code{L}$_{\small\TYPE{ill}}\cdot$\code{L}$_{\small\TYPE{ill}}\cdot$\code{L}$_{\small\TYPE{il}}$\\
\code{partition(E,L)} & \code{L}$_{\small\TYPE{il}}$ & \code{L}$_{\small\TYPE{il}}$ & \code{L}$_{\small\TYPE{il}}$\\
\code{traverse1(L)} & \code{L}$_{\small\TYPE{ill}} \cdot$ \code{L}$_{\small\TYPE{il}}$ & unknown & \code{L}$_{\small\TYPE{il}}$\\
\code{traverse2(L)} & \code{L}$_{\small\TYPE{ill}} \cdot$ \code{L}$_{\small\TYPE{il}}$ & \code{L}$_{\small\TYPE{il}}$ & \code{L}$_{\small\TYPE{il}}$\\
\code{zip3(L1,L2,L3)} & \code{L3}$_{\small\TYPE{il}}$ & min(\code{L1}$_{\small\TYPE{il}}$,\code{L2}$_{\small\TYPE{il}}$,\code{L3}$_{\small\TYPE{il}}$) & \code{L3}$_{\small\TYPE{il}}$ \\
\bottomrule
\end{tabular}
\caption{
Complexity order of the upper bounds obtained by RAML~\protect\cite{Hoffmann2017}, sized types~\protect\cite{SerranoLH14}, and typed-norms.\label{fig:comparison}
}
\end{figure}

To conclude this section we will compare the results obtained by RAML~\cite{Hoffmann2017,HAH12,Hofmann03}, sized types~\cite{SerranoLH14} and our approach using the set of list manipulation examples. To that end, we will use Table 1 from~\cite{SerranoLH14} extended with a new column containing the upper bounds obtained by our resource analysis framework using \cbstart the best typed-norm definition for each program\cbend. We will also consider a new \cbstart example \cbend \code{hanoi} that computes a list of movements needed to solve the Tower of Hanoi puzzle, as it is included in the comparison of~\cite{SerranoLH14}.
%
The results can be found in Figure~\ref{fig:comparison}, where we show the complexity order of the upper bounds instead of the concrete expression obtained. Parameters \code{L}, \code{L1}, \code{L2}, and \code{L3} are lists of different types, and parameters \code{N}, \code{A}, \code{B}, \code{C}, and \code{E} are integer numbers. We represent the integer type as ``\code{i}'', lists of integers as ``\code{il}'', lists of lists of integers as ``\code{ill}'', and lists of lists of lists of integers as ``\code{illl}''.

As shown in Figure~\ref{fig:comparison}, the three approaches obtain the same results for the majority of the methods and they differ in a small number of cases.
RAML cannot handle \code{factSum}, \code{fib}, \code{listfact}, and \code{listnum}, but both sized types and typed-norms obtain a similar upper bound. In the case of \code{fib} both approaches obtain an exponential upper bound but sized types obtain a slightly more precise bound of \cbstart $\varphi^N$, where $\varphi \approx 1.62$ \cbend is the golden ratio, whereas our approach obtains the upper bound $2^N$. 
Both RAML and typed-norms obtain upper bounds for \code{traverse1} and \code{traverse2}, but the upper bounds \code{L}$_{\small\TYPE{il}}$ computed by typed-norms are more precise than the ones from RAML (\code{L}$_{\small\TYPE{ill}} \cdot$ \code{L}$_{\small\TYPE{il}}$). 
On the other hand, sized types cannot handle \code{traverse1} but generate the same upper bound \code{L}$_{\small\TYPE{il}}$ for \code{traverse2}.
In \code{zip3} both RAML and typed-norms obtain the same upper bound \code{L3}$_{\small\TYPE{il}}$ (the length of the third list), but the upper bound  min(\code{L1}$_{\small\TYPE{il}}$,\code{L2}$_{\small\TYPE{il}}$,\code{L3}$_{\small\TYPE{il}}$) generated by sized types is more precise, as it represents the minimum length of any of the lists. Finally, RAML cannot handle the \code{hanoi} method whereas sized types obtain an upper bound of $2^N$.
Note that the fact that SACO cannot obtain an upper bound for this example is not related to typed-norms abstraction. This program creates a data-structure of exponential size in the input parameter $N$ and then traverses it. SACO is able to infer that the cost of generating the data-structure is $2^N$, and that the cost of traversing it is linear in the size of the data-structure. However, the underlying cost analysis techniques of SACO cannot track exponential input-output dependencies, and thus it fails to conclude that  the cost of traversing the data-structure is actually exponential in $N$.


\section{Related Work}\label{related}

Our work is inspired by~\cite{GenaimCGL02} where the authors introduce
the notion of typed-based norm in the context of termination analysis,
and show how types can be very useful for finding suitable norms even
for untyped languages like Prolog. They also illustrate that
typed-based norms sometimes must be combined to get a termination
proof.

\cbstart
Resource Aware ML~\cite{Hoffmann2017,HAH12,Hofmann03} 
uses automatic amortized resource analysis, where the main idea is to consider potential functions that depend on data structures. For every step the available potential must be sufficient for the cost of the evaluation and the potential of the next state. The information about the potential functions (combinations of a set of base polynomials) are annotated in the types, and the type system collects the relations among the different types, which are finally solved by linear programming. On the other hand, our transformational approach is based on Wegbreit's~\cite{DBLP:journals/cacm/Wegbreit75}: (1) generation of an abstract representation and (2) resolution of cost relations. Although originally presented for Lisp programs, this approach has been applied to other functional languages~\cite{Sands_thesis,grobauer01cost}, imperative languages like Java bytecode~\cite{AlbertAGPZ07} or actor systems~\cite{AlbertACGGPR15}.
\cbend

\cbstart
Sized types~\cite{DBLP:conf/popl/HughesPS96,HP99,Chin2001} are type expressions that incorporates annotations representing lower/upper bounds of the size of the different components of a type. They were originally proposed for guaranteeing various basic properties of reactive systems like productivity, memory leaks, or termination. In~\cite{phd:pedro_vasconcelos}, Vasconcelos proposed to use sized types to track the different sizes involved in a data structure and use them to perform resource
analysis. Unlike our approach, in this approach one can handle
multiple typed-norms on variables only by having parametric
data-structures.
\cbend
The techniques of Vasconcelos have been extended to
the context of logic programs~\cite{sized-types-iclp2013} and applied to resource
usage analysis~\cite{SerranoLH14} inside the PLAI abstract interpretation framework~\cite{ai-jlp,inc-fixp-sas} of CiaoPP~\cite{hermenegildo11:ciao-design-tplp}.

Our approach and the one presented in~\cite{sized-types-iclp2013,SerranoLH14} use the sizes of different inner parts of a data structure in order to obtain more precise bounds, but differ in the destination resource analysis framework: the analysis in~\cite{sized-types-iclp2013,SerranoLH14} must fit in the abstract interpretation framework used by CiaoPP, whereas our analysis must fit in a transformational framework based on Wegbreit's approach~\cite{DBLP:journals/cacm/Wegbreit75}.
From this point of view, the presented approach in this paper
is simpler to define and to implement because we do not need to update
the abstract interpretation theory (defining specific
concretization, abstraction functions, etc.). Instead, we have
to add explicit arguments for the sizes of data structures and define
a size abstraction which is rather straightforward. The implementation
simply requires a pre-process to add the arguments and properly
abstract them. Then, standard size analysis works on the transformed
program.
As regards accuracy, the resource analysis in~\cite{SerranoLH14} is closer to logic programming and takes into account some features like backtracking or failure to obtain more precise upper and lower bounds.
The sized types used in~\cite{SerranoLH14} can track separately the size of different components of a clause, but the only data type they use are lists of elements. Using lists and procedures, our approach could provide a similar level of precision \cbstart \cbend.
Finally, the approach used in~\cite{SerranoLH14} detects relevant Prolog variables and generates constraints between them. Once a variable is detected as relevant, its complete sized type is used to generate constraints, although some of them may not have any impact in the resource analysis. As explained in Section~\ref{sec:inference}, we define an additional step to infer the relevant typed-norms of any variable, which is an extension of the results in~\cite{AlbertAGGP11} to deal with typed-norms in addition to
useless arguments.
Using this analysis, irrelevant variables will have an empty set of relevant typed-norms, and will therefore be ignored in the resolution phase. Moreover, for relevant variables only those relevant typed-norms will be kept in the next phase. As explained in the experiments in Section~\ref{sec:experiments}, removing useless typed-norms is an optimization in the resolution phase with a great impact in the overall time when obtaining upper bounds, so this analysis is
essential to be scalable in practice. Moreover, to the best of our knowledge, it is the first time that
it is applied on norms.


\section{Conclusions and Future Work}\label{conclusions}

We have presented a transformational approach to resource analysis
with typed-norms which has the advantage that its formalization can be
done by only adapting the first phase of cost analysis in which the
program is transformed into an intermediate abstract program.  Besides
its simple formal development, the implementation easily integrates
into the previous system as a pre-phase to the existing analysis.
Additionally, we have presented, to the best of our knowledge, the
first algorithm to automatically infer typed-norms from programs. Our
analysis is formalized on a simple rule-based language and it is
therefore not tied to any particular programming language. Translating
from the standard programming languages to the rule-based form is
rather straightforward (see e.g. \cite{AlbertACGGPR15,AlbertAGPZ12}).
Finally, we have carried out a thorough experimental evaluation of our
proposal and integrated it within the SACO
system~\cite{AlbertAGGP11}.  In future work we plan to extend
our work to a concurrent setting for which the inference of
typed-norms will be more contrived.


\bibliographystyle{acmtrans}

\newpage
\appendix

\section{Proofs}\label{sec:proofs}

\begin{definition}[$\fevalt$ function for terms]\label{def:evalt}
The function $\fevalt : \mi{Terms} \times \TV \hookrightarrow \mi{Terms}$ evaluates a term $t$ based on a variable mapping $\tv$:
\[
\fevalt(t,\tv)=
\begin{cases}
\tv(x) & \mi{if~} t \equiv x, ~x \in \dom(\tv)\\
n      & \mi{if~} t \equiv n\\
\fevalt(e_1,\tv) + \fevalt(e_2,\tv) & \mi{if~} t \equiv e_1 + e_2\\
\fevalt(e_1,\tv) - \fevalt(e_2,\tv) & \mi{if~} t \equiv e_1 - e_2\\
\DC(\overline{\fevalt(t_n,\tv)}) & \mi{if~} t \equiv \DC(\overline{t_n})\\
\mi{undefined} & i.o.c
\end{cases}
\]
\end{definition}

\begin{definition}[$\fevalg$ function for guards]\label{def:evalg}
The function $\fevalg : \mi{Guards} \times \TV \hookrightarrow \TV$ checks if a guard $g$ is satisfied w.r.t.\ a variable mapping $\tv$, returning the variable mapping that instantiates the variables in the guard:
\[
\fevalg(g,\tv)=
\begin{cases}
\emptymap{}  & \mi{if~} g \equiv \mi{true}\\
\tv_1 \uplus \tv_2 & \mi{if~} g \equiv g_1 \wedge g_2, \fevalg(g_1, \tv) = \tv_1, \\
 & ~~~\fevalg(g_2, \tv \uplus \tv_1) = \tv_2 \\
\emptymap{} & \mi{if~} g \equiv e_1 \mathtt{~op~} e_2, ~ \mathtt{op} \in \{>, =, \geq\}, \\
 & ~~~\fevalt(e_1, \tv) = n_1, ~\fevalt(e_2, \tv) = n_2, ~n_1 \mathtt{~op~} n_2 \\
\tv' & \mi{if~} t \equiv \match(x,p),~ \exists \tv' \in \TV ~s.t.~ \tv'(p) = \tv(x), \\
  & ~~~\dom(\tv') = \vars(p)\\
\emptymap{} & \mi{if~} t \equiv \notmatch(x,p),~ \nexists \tv' \in \TV ~s.t.~ \tv'(p) = \tv(x), \\
  & ~~~\dom(\tv') = \vars(p)\\
\mi{undefined} & i.o.c
\end{cases}
\]
\end{definition}

\begin{prop}\label{teo:evalg_disjoint_tv}
Let $\{\overline{x}\}$ \cbstart be \cbend those variables occurring in the patterns $p$ of a guard $g$. If $\fevalg(g, \tv) = \tv'$ and $\dom(\tv) \cap \{\overline{x}\} = \emptyset$ then $\dom(\tv) \cap \dom(\tv') = \emptyset.$ \end{prop}
\begin{proof}
By induction on the structure of the guard $g$.
\end{proof}

\subsection{Results in Section~\ref{sec:abstraction}}

\begin{prop}\label{teo:satisf_rderiv}
Consider a configuration $\irstate$ and its abstraction $\irstate^\alpha = \arb \abstractSep \psi$. Then $\psi \nmodels \false$.
\end{prop}	
\begin{proof}
$\psi$ is a conjunction of equalities between distinct $X_T$ variables,
so there is a trivial model.
\end{proof}

\begin{prop}\label{teo:remove_contraints}
If $\varphi \land \psi \nmodels \false$ then $\varphi \nmodels \false$.
\end{prop}
\begin{proof}
$\varphi \land \psi$ have at least one solution $S$, and it \cbstart is \cbend also a valid solution to $\varphi$.
\end{proof}

\cbstart
In the rest of the appendix we will use the notation $\overline{X_n = z_n}$ to denote the conjunction of constraints $X_1 = z_1 \land X_2 = z_2 \land \ldots \land X_n = z_n$, where $z_i \in \mathbb{Z}$. Similarly, the notation $\overline{X_n = z_n} \in \varphi$ expresses that the conjunction of constraints $\varphi$ syntactically contains the constraints $\overline{X_n = z_n}$.
\cbend

\begin{prop}\label{teo:satisf_param_passing}
If $\psi_1 \wedge \psi_2 \not\models \false$ and $\overline{X_i = z_i} \in \psi_2$, where $z_i \in \mathbb{Z}$, then $\psi_1 \wedge \overline{X'_i = X_i} \wedge \psi_2 \wedge \overline{X'_i = z_i} \nmodels \false$.
\end{prop}
\begin{proof}
Every solution $S$ to the original set of constraints $\psi_1 \wedge \psi_2$ can be extended $S \cup \overline{X'_i = z_i}$ and it is also a solution of $\psi_1 \wedge \overline{X'_i = X_i} \wedge \psi_2 \wedge \overline{X'_i = z_i}$.
\end{proof}

In order to prove the soundness of the translation using typed-norms, we first need to prove that the definition of typed-norms is \emph{consistent} with the evaluation of terms and guards. The next two lemmas prove that consistency with respect the definition of \typednorm{t}{T} in Def.~\ref{def:typed-norms} and Def.~\ref{def:symb-typed-norms}, but the proof will be similar for $\typednormorig{t}{T}$. Moreover, any definition of typed-norms satisfying the following Lemma~\ref{teo:satisf_eval} and Lemma~\ref{teo:satisfy_guard_eval} will produce a sound translation.

\begin{lemma}\label{teo:satisf_eval}
If $\rrassigns{x}{t} \in \rbrprog$, $\fevalt(t,\tv) = v$ and $\psi \wedge \atv \not\models \false$ then $X_T = \typednorm{t}{T} \wedge  X_T = \typednorm{v}{T} \wedge \psi \wedge \atv \not\models \false$ for any type $T \in \typednorms(x)$.
\end{lemma}
\begin{proof}
By induction on the structure of the term $t$. It is important to notice that $\fevalt(t,\tv) = v$ implies that $y \in \dom(\tv)$ for every variable $y \in \vars(t)$, so $\atv$ will contain equalities $Y_T = \typednorm{\tv(y)}{T}$ that are consistent with $\psi$. Therefore,  $X_T = \typednorm{v}{T}$ will assign a concrete value to $X_T$ that is consistent with the constraint $X_T = \typednorm{t}{T}$ that involves those variables $Y_T$.
\end{proof}

\begin{lemma}\label{teo:satisfy_guard_eval}
Let $g$ a guard in a program $P$. If $\fevalg(g, \tv) = \tv_g$ and $\psi \wedge \atv \nmodels \false$ then $\psi \wedge g^\alpha \wedge \atv \wedge \atvl{\tv_g} \nmodels \false$
\end{lemma}
\begin{proof}
By induction on the structure of the guard $g$. The most interesting cases are $\match(x,p)$ and $g_1 \wedge g_2$:
\begin{itemize}
 \item $g \equiv \match(x,p)$, where $p \equiv \DC(y^1, \ldots, y^k)$ and $y^i$ do not appear in $\dom(\tv)$. In this case $g^\alpha$ is $\bigwedge\{ X_T = \typednorm{p}{T} \mid T\in\typednorms(x) \}$. Consider only a type $T'\in\typednorms(x)$ such that $T' = \type{x}$, then $X_{T'} = 1 + Y^1_{T'} \ldots + Y^k_{T'}$. On the other hand, $X_{T'} = z$ appears in $\atv$ and $\atv_g$ contains some constraints $\overline{Y^i_{T'} = z_i}$ such that $1 + z_i \ldots + z_k = z$ (by definition of $\fevalg$). Therefore, $\psi \wedge X_{T'} = 1 + Y^1_{T'} \ldots + Y^k_{T'} \wedge \overline{Y^i_{T'} = z_i} \wedge \atv_g \nmodels \false$. The case is the same if $T' \neq \type{x}$, and it can be repeated for every $T'\in\typednorms(x)$, so $\psi \wedge g^\alpha \wedge \atv \wedge \atvl{\tv_g} \nmodels \false$.
 
 \item $g \equiv g_1 \wedge g_2$. By definition of guard abstraction (Fig.~\ref{fig:sizeabst}) $g^\alpha = g_1^\alpha \land g_2^\alpha$, and by definition of $\fevalg$ we have that (A) $\fevalg(g_1, \tv) = \tv_1$ and (B) $\fevalg(g_2, \tv \uplus \tv_1) = \tv_2$. From (A) and the premises, by IH we obtain that $\psi \wedge g_1^\alpha \wedge \atv \wedge \atv_1 \nmodels \false$. By Prop.~\ref{teo:evalg_disjoint_tv} we know that $\dom(\tv) \cap \dom(\tv_1) = \emptyset$, so 
 $\atv \wedge \atv_1 = \atvl{(\tv \uplus \tv_1)}$. Then, again, by IH and using (B) we have that $\psi \wedge g_1^\alpha \wedge g_2^\alpha \wedge \atv \wedge \atvl{(\tv \uplus \tv_1)} \wedge \atv_2 \nmodels \false$. Finally, as by Prop.~\ref{teo:evalg_disjoint_tv} the domains of the generated variable mappings are disjoint then $\atvl{(\tv \uplus \tv_1)} \wedge \atv_2 = \atv \wedge \atv_1 \wedge \atv_2 = \atv \wedge \atvl{(\tv_1 \uplus \tv_2)}$ and therefore $\psi \wedge g_1^\alpha \wedge g_2^\alpha \wedge \atv \wedge \atvl{(\tv_1 \uplus \tv_2)} \nmodels \false$ .
\end{itemize}
\end{proof}


\noindent\textbf{Theorem~\ref{teo:soundness} (Soundness).}
	If $\trace \equiv \irstate_0 \rrderiv^*
	\irstate_n$ then there is an abstract trace
	$\trace^\alpha \equiv \irstate_0^\alpha \rrabsderiv^* \arb \abstractSep \psi$
	such that $\trsteps{\trace} = \trsteps{\trace^\alpha}$, $\irstate_n^\alpha = \arb \abstractSep \widetilde{\psi}$ and $\psi \wedge \widetilde{\psi} \not\models \false$.

\begin{proof}
By induction on the length of the $\irstate_0 \rrderiv^n \irstate_n$ derivation.

~
\cbstart
\noindent\underline{Base Case:} $n = 0$\\
We have $\trace \equiv \irstate_0 \rrderiv^0 \irstate_n$ and the trivial derivation $\trace^\alpha \equiv \irstate_0^\alpha \rrabsderiv^0 \irstate_0^\alpha$, where $\irstate_0^\alpha \equiv \arb \abstractSep \psi$ and therefore $\psi = \widetilde{\psi}$ in this case. By Prop.~\ref{teo:satisf_rderiv} $\psi \nmodels \false$, so $\psi \wedge \psi \nmodels \false$. Additionally, the sequences of steps in both traces are empty: $\trsteps{\trace} = \langle \rangle = \trsteps{\trace^\alpha}$.\\

\noindent\underline{Inductive Step:} $n > 0$\\
We have a derivation $\irstate_0 \rrderiv^{n-1} \irstate_{n-1} \rrderiv \irstate_n$.
By Induction Hypothesis we have that if $\irstate_0 \rrderiv^{n-1} \irstate_{n-1}$ then $\irstate_0^\alpha \rrabsderiv^{n-1} \arb_{n-1} \abstractSep \psi_{n-1}$ such that $\irstate_{n-1}^\alpha = \arb_{n-1} \abstractSep \widetilde{\psi_{n-1}}$, $\psi_{n-1} \wedge \widetilde{\psi_{n-1}} \nmodels \false$, and $\trsteps{\irstate_0 \rrderiv^{n-1} \irstate_{n-1}} = \trsteps{\irstate_0^\alpha \rrabsderiv^{n-1} \arb_{n-1} \abstractSep \psi_{n-1}}$. Depending on $\irstate_{n-1}$ we can perform the last step using the 3 different rules:

\begin{itemize}
	\item $\irstate_{n-1} \equiv \tuple{p,\rrassigns{x}{t}{\cdot}\stkbc,\tv}{\cdot}\ar$. Then we can only apply rule $(1)$:
	$$(1)~\semrule
	{\bc \equiv \rrassigns{x}{t}~~ \fevalt(t,\tv)=v}
	{
		\ar_{n-1} \equiv \tuple{p,\bc{\cdot}\stkbc,\tv}{\cdot} \ar 
		\rrderiv^{(1)\cdot \epsilon}
		\tuple{p,\stkbc,\tv[x\mapsto v]}{\cdot} \ar \equiv \ar_{n}
	}$$
	By definition of configuration abstraction (Def.~\ref{def:state_abstr}) $\irstate_{n-1}^\alpha = \arb_{n-1} \abstractSep \widetilde{\psi_{n-1}}$, where $\arb_{n-1} = \tuple{\varphi_1 \cdot \astkbc{}} \cdot \arb$
	and $\varphi_1 = \bigwedge \{ X_T = \typednorm{t}{T} | T \in \typednorms(x) \}$.
	Then we can perform the following abstract step from 
	$\arb_{n-1} \abstractSep \psi_{n-1}$:
	\[(1)~\semrule
	{\varphi_1 {\wedge} \psi_{n-1} \not\models \false
	}
	{
		\tuple{\varphi_1 \cdot \bsa} \cdot \arb \abstractSep \psi_{n-1}
		\rrabsderiv^{(1)\cdot \epsilon}
		\tuple{\bsa} \cdot \arb \abstractSep \varphi_1{\wedge}\psi_{n-1}\\
	}\]
	We define $\varphi_2 = \bigwedge \{ X_T = \typednorm{\tv(x)}{T} ~|~T \in \typednorms(x) \}$ as the constraints added by the mapping extension $\tv[x\mapsto v]$, and $\tv_1$, \ldots, $\tv_k$ the variable mappings in the activation records in $\ar$. By IH we have $\psi_{n-1} \wedge \widetilde{\psi_{n-1}} \nmodels \false$, where $\widetilde{\psi_{n-1}} = \atv \land \atv_1 \land \ldots \land \atv_k$ by definition.
	By the repeated application of Lemma~\ref{teo:satisf_eval} using all the variable mappings $\tv \uplus \tv_1 \uplus \ldots \uplus \tv_k$ (their domains are disjoint because variables in every activation record are fresh), we obtain that
	$\varphi_1 \wedge \varphi_2 \wedge \psi_{n-1} \wedge \widetilde{\psi_{n-1}} \nmodels \false$. Therefore the abstract $(1)$ step is correct---$\varphi_1 \wedge \psi_{n-1}\nmodels \false$ by Prop.~\ref{teo:remove_contraints}---and
	$\varphi_1 \wedge \varphi_2 \wedge \psi_{n-1} \wedge \widetilde{\psi_{n-1}} = (\varphi_1 \wedge \psi_{n-1}) \wedge (\varphi_2 \wedge \widetilde{\psi_{n-1}}) = \psi_{n} \wedge \widetilde{\psi_{n}} \nmodels \false$. 
	Finally, since $\trsteps{\irstate_0 \rrderiv^{n-1} \irstate_{n-1}} = \trsteps{\irstate_0^\alpha \rrabsderiv^{n-1} \arb_{n-1} \abstractSep \psi_{n-1}}$ then it is clear that $\trsteps{\irstate_0 \rrderiv^{n-1} \irstate_{n-1} \rrderiv^{(1)\cdot \epsilon} \irstate_n} = \trsteps{\irstate_0^\alpha \rrabsderiv^{n-1} \arb_{n-1} \abstractSep \psi_{n-1} \rrabsderiv^{(1)\cdot \epsilon}\tuple{\bsa} \cdot \arb \abstractSep \varphi_1{\wedge}\psi_{n-1}}$.

	\item $\irstate_{n-1} \equiv \tuple{p,m(\bar{x},\bar{y}){\cdot}\stkbc,\tv}{\cdot}\ar$. Then we can only apply rule $(2)$:
	\[(2)~ \semrule
	{\bc \equiv m(\bar{x},\bar{y})~~ m(\bar{x'},\bar{y'}) \rulearrow g, \bc_1 \cdot \cdot \cdot \bc_k \in \p \mi{~fresh} \\
		\tv_1 \equiv [\overline{x' \mapsto \tv(x)}] ~~ \fevalg(g,\tv_1) = \tv_2
		
	}
	{
		\irstate_{n-1} \equiv \tuple{p,\bc{\cdot}\stkbc,\tv}{\cdot} \ar 
		\rrderiv^{(2)\cdot \mi{rn}}
		\tuple{m,\bc_1 \cdot \cdot \cdot \bc_k,\tv_1 \uplus \tv_2} {\cdot} \tuple{p[\overline{y' \sim y}],\stkbc,\tv}{\cdot} \ar \equiv \irstate_{n}
	}\]
	We assume that the fresh program rule $m(\bar{x'},\bar{y'}) \rulearrow g, \bc_1 \cdot \cdot \cdot \bc_k$ used in the step is the number $\mi{rn}$. By definition of configuration abstraction (Def.~\ref{def:state_abstr}) $\irstate_{n-1}^\alpha = \arb_{n-1} \abstractSep \widetilde{\psi_{n-1}}$, where $\arb_{n-1} = \tuple{p(\bar{X,\bar{Y}})\cdot \astkbc{}} \cdot \arb$.
	Then we can perform the following abstract step from 
	$\arb_{n-1} \abstractSep \psi_{n-1}$ using the abstract semantic rule $(2)$ and a fresh instance of the abstract program rule number $\mi{rn}$:
	\[(2)~\semrule
	{
		m(\bar{X'},\bar{Y'})  \rulearrow g^\alpha ~\abstractSep 
		~ \bca_1,\ldots,\bca_k  \in \palpha \mi{~fresh}
		~~~~~~
		\overline{X = X'} \wedge\psi_{n-1}{\wedge}g^\alpha \not\models {\it false}
	}
	{
		\tuple{m(\bar{X},\bar{Y}) \cdot \bsa} \cdot \arb \abstractSep \psi_{n-1}
		\rrabsderiv^{(2)\cdot \mi{rn}} \\ 
		\tuple{\bca_1 \cdots \bca_k}^{\overline{Y = Y'}} \cdot \tuple{\bsa} \cdot \arb \abstractSep \overline{X = X'} \wedge\psi_{n-1}\wedge g^\alpha 
	}\]
	We need to prove that:
	\begin{enumerate}
	\item $(\overline{X = X'} \wedge g^\alpha \wedge \psi_{n-1}) \wedge (\atv_1 \wedge \atv_2 \wedge \widetilde{\psi_{n-1}}) = \psi_{n} \wedge \widetilde{\psi_{n}} \nmodels \false$, from the soundness theorem.
	\item $(\overline{X = X'} \wedge g^\alpha \wedge \psi_{n-1}) \nmodels \false$, i.e., the abstract step is valid.
	\end{enumerate}
	We will focus only on the first statement as it implies the second \cbstart one \cbend by Prop.~\ref{teo:remove_contraints}. 
	\cbstart
	By Prop.~\ref{teo:satisf_param_passing} we have that 
	$(\overline{X = X'} \wedge \psi_{n-1}) \wedge (\atv_1 \wedge \widetilde{\psi_{n-1}}) \nmodels \false$, because:
	\begin{itemize}
	\item $\bigwedge_{x \in \overline{x}, T \in \typednorms(x)} \{X_T = \typednorm{\tv(x)}{T}\} \in \widetilde{\psi_{n-1}}$, with $\typednorm{\tv(x)}{T} \in \mathbb{Z}$, by definition of configuration transformation, and
	\item $\atv_1 = [\overline{x' \mapsto \tv(x)}]^\alpha = \bigwedge_{x' \in \overline{x'}, T \in \typednorms(x')} \{X'_T = \typednorm{\tv(x)}{T}\}$
	\item $\typednorms(x) = \typednorms(x')$
	\end{itemize}   
	\cbend
	Then by the guard evaluation $\fevalg(g,\tv_1) = \tv_2$ and Lemma~\ref{teo:satisfy_guard_eval} we have that: 
	\[
	\begin{array}{ll}
	  & (\overline{X = X'} \wedge \psi_{n-1} \wedge \widetilde{\psi_{n-1}}) \wedge g^\alpha \wedge \atv_1 \wedge \atv_2 \nmodels \false \\
	  = & (\overline{X = X'} \wedge g^\alpha \wedge \psi_{n-1}) \wedge (\atv_1 \wedge \atv_2 \wedge \widetilde{\psi_{n-1}}) \nmodels \false\\
	  = & \psi_{n} \wedge \widetilde{\psi_{n}} \nmodels \false\\
	 \end{array}
	\]
	Similarly to the previous case $\trstepsnolink(\irstate_0 \rrderiv^{n-1} \irstate_{n-1} \rrderiv^{(2)\cdot \mi{rn}} \irstate_n) = \trstepsnolink(\irstate_0^\alpha \rrabsderiv^{n-1}$
	$ \arb_{n-1}\abstractSep \psi_{n-1} \rrabsderiv^{(2)\cdot \mi{rn}} \tuple{\bca_1 \cdots \bca_m}^{\overline{Y = Y'}} \cdot \tuple{\bsa} \cdot \arb \abstractSep \overline{X = X'} \wedge\psi_{n-1}\wedge g^\alpha )$.

	
	\item $\irstate_{n-1} \equiv \tuple{m,\epsilon,\tv_1}{\cdot}\tuple{p[\overline{y' \sim y}],\stkbc,\tv}{\cdot} \ar$ . Then we can only apply rule $(3)$:
	\[(3)~\semrule
	{}
	{
		\irstate_{n-1} \equiv \tuple{m,\epsilon,\tv_0}{\cdot}\tuple{p[\overline{y' \sim y}],\stkbc,\tv}{\cdot} \ar 
		\rrderiv^{(3)\cdot \epsilon}
		\tuple{p,\stkbc,\tv[\overline{y \mapsto \tv_0(y')}]}{\cdot} \ar \equiv \irstate_{n}
	}\]
	By definition of configuration abstraction (Def.~\ref{def:state_abstr}) $\irstate_{n-1}^\alpha = \arb_{n-1} \abstractSep \widetilde{\psi_{n-1}}$, where
	$\arb_{n-1} = \tuple{\epsilon}^{\overline{Y = Y'}} \cdot \tuple{\bsa} \cdot \arb$.
	Then we can perform the following abstract step from 
	$\arb_{n-1} \abstractSep \psi_{n-1}$:
	\[(3) \semrule
	{\psi_{n-1}{\wedge}\overline{Y = Y'} \not\models {\it false}
	}
	{
		\tuple{\epsilon}^{\overline{Y = Y'}} \cdot \tuple{\bsa} \cdot \arb \abstractSep \psi_{n-1}
		\rrabsderiv^{(3)\cdot \epsilon}
		\tuple{\bsa} \cdot \arb \abstractSep \psi_{n-1}{\wedge}\overline{Y=Y'} 
	}\]
	Let $\tv_1$, $\tv_2$, \ldots, $\tv_k$ the variable mappings in the activation records in $\ar$. The constraints $\atv_0$ contains $\overline{Y'_i = z_i}$ for some $z_i \in \mathbb{Z}$, therefore by definition of configuration abstraction (Def.~\ref{def:state_abstr}) we have that $\widetilde{\psi_{n-1}} \equiv \atv_0 \land \atv \land \atv_1 \land \atv_2 \land \ldots \land \atv_k$ and $\widetilde{\psi_n} \equiv \atv \land \overline{Y_i = z_i} \land \atv_1 \land \atv_2 \land \ldots \land \atv_k $. As before, we need to prove that:
	\begin{enumerate}
	\item $\psi_{n} \wedge \widetilde{\psi_{n}} = (\psi_{n-1} \land \overline{Y=Y'}) \land (\atv \land \overline{Y_i = z_i} \land \atv_1 \land \atv_2 \land \ldots \land \atv_k) \nmodels \false$, from the soundness theorem.
	\item $\psi_{n-1} \land \overline{Y=Y'}\nmodels \false$, i.e., the abstract step is valid.
	\end{enumerate}
	
	The first statement implies the second \cbstart one \cbend by Prop.~\ref{teo:remove_contraints}, so we focus only on the first one. By IH we have $\psi_{n-1} \wedge \widetilde{\psi_{n-1}} \nmodels \false$, so we can apply Prop.~\ref{teo:satisf_param_passing} and obtain that $\psi_{n-1} \land \overline{Y_i = Y_i'} \wedge \widetilde{\psi_{n-1}} \land \overline{Y_i = z_i}\nmodels \false$, i.e.,  $\psi_{n-1} \land \overline{Y_i = Y_i'} \wedge \atv_0 \land \atv \land \atv_1 \land \atv_2 \land \ldots \land \atv_k \land \overline{Y_i = z_i}\nmodels \false$. By Prop.~\ref{teo:remove_contraints} we can remove the constraints $\atv_0$ and the set remains satisfiable, therefore $(\psi_{n-1} \land \overline{Y=Y'}) \land (\atv \land \overline{Y_i = z_i} \land \atv_1 \land \atv_2 \land \ldots \land \atv_k) \nmodels \false$.
	As before $\trstepsnolink(\irstate_0 \rrderiv^{n-1} \irstate_{n-1} \rrderiv^{(3)\cdot \epsilon} \irstate_n) = \trstepsnolink(\irstate_0^\alpha \rrabsderiv^{n-1}$
	$ \arb_{n-1}\abstractSep \psi_{n-1} \rrabsderiv^{(3)\cdot \epsilon} \tuple{\bsa} \cdot \arb \abstractSep \psi_{n-1}{\wedge}\overline{Y=Y'} )$.
\end{itemize}
\cbend
%
%
\end{proof}

\subsection{Results in Section~\ref{sec:inference_formalization}}

\cbstart
\begin{lemma}\label{teo:poset}
$(\prmaps{\rbrprog{}}, \sqsubseteq)$ is a partially ordered set.
\end{lemma}
\begin{proof}
\begin{itemize}
\item reflexivity: $\prmap \sqsubseteq \prmap$ because $\prmap(i) \sqsubseteq \prmap(i)$ for every rule $i$, since for all $x \in \dom(\prmap(i))$ we have that $\prmap(i)(x) \subseteq \prmap(i)(x)$.
\item transitivity: if $\prmap_1 \sqsubseteq \prmap_2$ then for all rule $i$ and $x \in \dom(\prmap_1(i))$, $\prmap_1(i)(x) \subseteq \prmap_2(i)(x)$. Similarly, if $\prmap_2 \sqsubseteq \prmap_3$ then for all rule $i$ and $x \in \dom(\prmap_2(i))$, $\prmap_3(i)(x) \subseteq \prmap_3(i)(x)$. Therefore, for every rule $i$ and $x \in \dom(\prmap_1(i))$, $\prmap_1(i)(x) \subseteq \prmap_2(i)(x) \subseteq \prmap_3(i)(x)$, so $\prmap_1 \sqsubseteq \prmap_3$.
\item anti-symmetry: if $\prmap_1 \sqsubseteq \prmap_2$ then for all rule $i$ and $x \in \dom(\prmap_1(i))$, $\prmap_1(i)(x) \subseteq \prmap_2(i)(x)$. Similarly, if $\prmap_2 \sqsubseteq \prmap_1$ then for all rule $i$ and $x \in \dom(\prmap_2(i))$, $\prmap_2(i)(x) \subseteq \prmap_1(i)(x)$. Therefore for all rule $i$ $\dom(\prmap_1(i)) = \dom(\prmap_2(i))$ and $\prmap_1(i)(x) = \prmap_2(i)(x)$ for every variable, so $\prmap_1 = \prmap_2$.
\end{itemize}
\end{proof}

\begin{lemma}\label{teo:leastub}
In the partially ordered set $(\prmaps{\rbrprog{}}, \sqsubseteq)$ every subset $C \subseteq \prmaps{\rbrprog{}}$ has a least upper bound $\prmap' = \baddr_{\prmap \in C} \prmap$.
\end{lemma}
\begin{proof}
We consider that $|P|= n$ and $C = \{\prmap_1, \prmap_2, \ldots, \prmap_k\}$, where $\prmap_i = \langle \rulemap_i^1, \rulemap_i^2, \ldots, \rulemap_i^n \rangle$. By definition of $\baddr$ we have $\prmap' = \langle \rulemap_1^1 \addr \rulemap_2^1 \ldots \addr \rulemap_k^1, \rulemap_1^2 \addr \rulemap_2^2 \ldots \addr \rulemap_k^2, \ldots, \rulemap_1^n \addr \rulemap_2^n \ldots \addr \rulemap_k^n\rangle$. 

We proceed by reduction to the absurd: suppose $\prmap''$ is a least upper bound of $C$ but $\prmap' \not \sqsubseteq \prmap''$, then for some rule $i$ and variable $x$ we have $\prmap'(i)(x) \not \subseteq \prmap''(i)(x)$. If $x \notin \dom(\prmap''(i))$ then $\prmap''$ cannot be an upper bound of $C$, because $\dom(\prmap'(i)) = \bigcup_{j=1}^k \dom(\rulemap_j^i))$ so for some $\rulemap_j^i$ we will have that $\rulemap_j^i(x) \not \subseteq \prmap'(i)$. On the other hand, if $x \in \dom(\prmap''(i))$ then $\prmap'(i)(x) \not \subseteq \prmap''(i)(x)$ because $T \in \prmap'(i)(x)$ but $T \not \in \prmap''(i)(x)$ for some type $T$. By definition of $\baddr$ we have that $\prmap'(i)(x) = \bigcup_{j=1}^k \rulemap_j^i(x)$, so $T \in \rulemap_j^i(x)$ for some $\rulemap_j^i(x)$. Therefore $\prmap''$ cannot be an upper bound of $C$ because $\prmap''(i)(x) \not \subseteq \rulemap_j^i(x)$, so $\prmap'' \not \sqsubseteq \prmap_j$.
\end{proof}

\noindent\textbf{Theorem~\ref{teo:lfp}.}
Consider a program $\rbrprog{}$ such that $\nrules{\rbrprog{}} = n$. 
Then $\mi{lfp}(\genf{\rbrprog{}})$ exists and is the supremum of the ascending Kleene chain starting from $\bot_\rbrprog{} = \langle \epsilon_1, \ldots, \epsilon_n \rangle$.

\begin{proof}
$(\prmaps{\rbrprog{}}, \sqsubseteq)$ is a partially ordered set because by Lemma~\ref{teo:poset} $\sqsubseteq$ is reflexive, transitive, and anti-symmetric. By Lemma~\ref{teo:leastub} every subset of \prmaps{\rbrprog{}} has a least upper bound (obtained by \addr) so $(\prmaps{\rbrprog{}}, \sqsubseteq)$ is also a complete lattice by Lemma A.2~\cite{DBLP:books/daglib/0098888}. \prmaps{\rbrprog{}} satisfies the Ascending Chain condition trivially because it is finite, and $\genf{\rbrprog{}}$ is monotone as by definition (Fig.~\ref{fig:infer}) it extends the program mapping $\prmap$ passed as argument. Then $\mi{lfp}(\genf{\rbrprog{}}) = \genf{\rbrprog{}(n)}(\bot_\rbrprog{})$ for some $n \ge 0$~\cite{DBLP:books/daglib/0098888,tarski1955lattice}.
\end{proof}
\cbend 

\subsection{Results in Section~\ref{sec:inference_soundness}}

In this section we need to track the concrete program rule associated to each activation record. Therefore we will assume that activation records contain in their first element the rule number used to create them. For example $\tuple{p^i, \bc{\cdot}\stkbc,\tv}$ is an activation record generated by a call to procedure $p$ using the i-th rule of the program.

We will also use the notion of a variable $z$ being dependent on a variable $x$ w.r.t.\ a type $T$ in a step $\ar \rrderiv \ar'$---written \valdep{x}{T}{z}. This relation tracks the dependence between variables in a step, so that a change in a component of type $T$ in the value of $x$ in the original configuration will have an impact on the components of type $T$ of variable $z$ in the destination configuration. Formally:

\begin{definition}[Dependent variables in a $\rrderiv$-step]\label{def:dependent_vars}
Consider a step $\tuple{p^i,\stkbc,\tv} \cdot \ar'_0 \rrderiv \tuple{q^j,\stkbc',\tv'} \cdot \ar'_1$. The set $D$ of dependent variables w.r.t.\ $T$ is defined as:
\begin{itemize} 
	
	\item If the step evaluates an assignment \rrassigns{x}{t} using rule (1) of Fig.~\ref{fig:rrsem} then $D = \{ \valdep{y}{T}{x} \mid y \in \vars(t),~ T \preceq \type{y} \}$, i.e., $x$ is dependent on all the variables $y \in \vars(t)$.
	
	\item If the step evaluates a procedure call $m(\bar{x},\bar{y})$---rule (2) of Fig.~\ref{fig:rrsem}---using the program rule $m(\bar{x'},\bar{y'}) \rulearrow g, \bc_1 \cdot \cdot \cdot \bc_k \in \rbrprog{}$ then $D = \depvars{T}{D_0, g}$ where $D_0 = \{ \valdep{x_k}{T}{x'_k} \mid T \preceq \type{x_k} \}$ and
	{\small
	\[
	\hspace*{-0.9cm}
	\depvars{T}{D,g} = 
	\left \{ 
	\begin{array}{lll} 
	\depvars{T}{\depvars{T}{D,g_1},g_2}  && \mbox{ if } g = g_1 \land g_2\\
	D \cup \{ \valdep{x}{T}{y} \mid y \in \vars(p),~ T \preceq \type{y} \} \cup && \mbox{ if } g = \match(x,p)\\
	~~\{ \valdep{z}{T}{y} \mid \valdep{z}{T}{x} \in D,~y \in \vars(p),~ T \preceq \type{y} \} \\
	D &~~~& \mbox{ otherwise }\\
	\end{array}
	\right. \enspace
	\]
	}
	Note that, in the case of $\match(x,p)$ guards, the variables $y \in \vars(p)$---which are fresh---are dependent on the variable $x$. Recursively, if $x$ depends on any variable $z$, then any variable $y \in \vars(p)$ is also dependent on $z$.
	
	\item Finally, if the step uses rule (3) of Fig.~\ref{fig:rrsem} to evaluate $\tuple{p^i,\epsilon,\tv_1}{\cdot}\tuple{q^j[\overline{y' \sim y}],\stkbc,\tv}{\cdot} \ar$ then we have $D = \{\valdep{y'_k}{T}{y_k} \mid y'_k,~ T \preceq \type{y'_k}\}$, i.e., the output arguments of a call ($y_k$) are dependent on the output variables of the fresh instance ($y'_k$) used to evaluate a call.
\end{itemize}	
We use the notation $\valdep{x}{*T}{y}$ in a trace $\trace = \ar_0 \rrderiv \ar_1 \rrderiv \ldots \rrderiv \ar_n$ to refer to a chain of $k \leq n$ dependences between ordered (but not necessarily consecutive) steps. Formally, we have a sequence of $k$ positions $p_1 < p_2 < \ldots < p_{k-1} < p_{k}$ such that $0 \leq p_i < n$ and a chain of $k$ dependencies:
\begin{enumerate}
\item[1.] $\valdep{x}{T}{x_1}$ in some step $\ar_{p_1} \rrderiv \ar_{p_1+1}$, 
\item[2.] $\valdep{x_1}{T}{x_2}$ in some subsequent step $\ar_{p_2} \rrderiv \ar_{p_2+1}$,\\
\vdots
\item [(k-1).]  $\valdep{x_{k-2}}{T}{x_{k-1}}$ in some subsequent step $\ar_{p_{k-1}} \rrderiv \ar_{p_{k-1}+1}$,
\item [k.] $\valdep{x_{k-1}}{T}{y}$ in step $\ar_{p_{k}} \rrderiv \ar_{p_{k}+1}$ with $p_{k} > p_{k-1}$
\end{enumerate}

\end{definition}

Finally, we will use an extended notion of \emph{value variations} for variable mappings and configurations. We will say that $\tv'$ is a \emph{variation} of $\tv$, written $\variation{\tv}{\tv'}{T}$, if $\tv' = \tv[\overline{x_i \mapsto v'_i}]$ for some variables $x_i \in \dom(\tv)$ and $\variation{\tv(x_i)}{v'_i}{T}$. This notion can be extended to configurations, written $\variation{\ar}{\ar'}{T}$, if $\ar'$ results from $\ar$ by replacing some of its variable mappings $\tv$ by $\tv'$ such that $\variation{\tv}{\tv'}{T}$.

\begin{lemma}\label{lemma:typednormsStep}
Let $\prmap$ be the result of the typed-norms inference of a program $\rbrprog$. If $\tuple{p^i, \stkbc,\tv} \cdot \ar_0 \rrderiv \tuple{q^j, \stkbc',\tv'} \cdot \ar'_0$, $T \in \prmap(j)(z)$ and $\valdep{x}{T}{z}$, 
then $T \in \prmap(i)(x)$.
\end{lemma}
\begin{proof}
By case distinction on the semantic rule from Fig.~\ref{fig:rrsem} used to perform the step:

\begin{itemize} 
\item If the step evaluates an assignment using rule (1) we have:
\[
\begin{array}{lc}
(1) & \semrule
{\bc \equiv \rrassigns{x}{t}~~ \fevalt(t,\tv)=v}
{
	\tuple{p^i,\bc{\cdot}\stkbc,\tv}{\cdot} \ar 
	\rrderiv
	\tuple{p^i,\stkbc,\tv[x\mapsto v]}{\cdot} \ar
}
\end{array}
\]
In this case the step is in the same activation record of $p^i$, and the only dependent variable is $x$, which depends on every $y \in \vars(t)$: $\valdep{y}{T}{x}$. Then by the first equation of $\genS{\rbrprog{}}_i$ any typed-norm of $x$---$T \in \prmap(i)(x)$---will be a typed-norm of $y$---$T \in \prmap(i)(y)$.

\item When using rule (2) we have:
\[
\begin{array}{lc}
(2) & \semrule
{
	\bc \equiv m(\bar{x},\bar{y})~~~~ m(\bar{x'},\bar{y'}) \rulearrow g, \bc_1 \cdot \cdot \cdot \bc_k \in \rbrprog{} ~\mi{fresh}\\
	\tv_1 \equiv [\overline{x' \mapsto \tv(x)}] ~~ \fevalg(g,\tv_1) = \tv_2
	
}
{
	\tuple{p^i,\bc{\cdot}\stkbc,\tv}{\cdot} \ar 
	\rrderiv
	\tuple{m^j,\bc_1 \cdot \cdot \cdot \bc_k,\tv_1 \uplus \tv_2} {\cdot} \tuple{p^i[\overline{y' \sim y}],\stkbc,\tv}{\cdot} \ar
}
\end{array}
\]
In this case the dependence between \cbstart variables \cbend is more complex: there are dependences from the parameter passing ($\valdep{x_k}{T}{x'_k}$) and also from the sequence of guards ($\valdep{x_k}{T}{z_m}$ for those variables $z_m$ appearing in the right-hand side of $\match$ guards). 
\begin{enumerate}
	\item \label{dep:1} For the parameter passing, we have $T \in \prmap(i)(x_k)$ directly by the second equation of $\genS{\rbrprog{}}_i$ (set $A$). The typed-norms of any input variable for any rule of $m$ (including rule number $j$) will be propagated to the arguments $\bar{x}_k$ of the call $m(\bar{x}, \bar{y})$ in rule $i$. 
	\item Assume a variable $z_m$ in some guard such that $\valdep{x_k}{T}{z_m}$ and $T \in \prmap(i)(z_m)$. Thus there is a sequence of \match{} guards 
	\[
	\match(x'_k,p_1) \wedge \match(z_1,p_2) \land \ldots \land \match(z_{m-1},p_m)
	\]
	such that $z_i \in \mi{vars}(p_i)$---note that we can safely ignore $e_1~\mi{op}~e_2$ as they do not define new variables. Then by definition of $\genG{\rbrprog{}}_i$ ($2^{nd}$ and $4^{th}$ rules) we know that $T \in \prmap(i)(z_{m-1})$, $T \in \prmap(i)(z_{m-2})$, \ldots, $T \in \prmap(i)(z_2)$, $T \in \prmap(i)(z_1)$, $T \in \prmap(i)(x'_k)$. Therefore $T \in \prmap(i)(x_k)$ using the same reasoning as in case~\ref{dep:1}.
\end{enumerate}

\item Finally, if we use rule (3):
\[
\begin{array}{lc}
  (3) & \semrule
{}
{
	\tuple{p^i,\epsilon,\tv_1}{\cdot}\tuple{q^j[\overline{y' \sim y}],\stkbc,\tv}{\cdot} \ar 
	\rrderiv
	\tuple{q^j,\stkbc,\tv[\overline{y \mapsto \tv_1(y')}]}{\cdot} \ar
}
\end{array}
\]
In this case we have that $\valdep{y'_k}{T}{y_k}$, so if $T \in \prmap(j)(y_k)$ then $T \in \prmap(i)(y'_k)$ because of the second equation of $\genS{\rbrprog{}}_i$: rule $q^j$ will contain a call $q(\bar{x}, \bar{y})$, so using the set $B$ the typed-norm $T$ will be propagated to any output variable of any rule of procedure $q$, in particular those $y'_k \in \rvars{i}$ such that $\valdep{y'_k}{T}{y_k}$. Therefore $\prmap(i)(y'_k)$.
\end{itemize}
\end{proof}

\noindent\textbf{Theorem~\ref{teo:soundnessInference} (Soundness)}. 
If $\useful{i}{T}{x}$ then $T \in \typednorms_i(x)$. 

\begin{proof}
Let $\prmap$ be the result of the typed-norms inference of the program. By definition of $\useful{i}{T}{x}$ there are two configurations such that:
\begin{enumerate}
	\item $\ar_0 = \tuple{p,\bc{\cdot}\stkbc,\tv} \cdot \ar'$ containing statements from the i-th rule of $\rbrprog{}$
	\item $x \in \dom(\tv)$
	\item $\ar'_0 = \tuple{p,\bc{\cdot}\stkbc,\tv'} \cdot \ar'$ with $\tv' = \tv[x \mapsto v]$ and $\variation{\tv(x)}{v}{T}$
	\item \label{p4}$\traces(\ar_0) \neq \traces(\ar'_0)$
\end{enumerate}
From point~\ref{p4} we know that $\traces(\ar_0) \not\subseteq \traces(\ar'_0)$ or $\traces(\ar_0) \not\supseteq \traces(\ar'_0)$. We will focus only on the first case, as the second one is similar. Since $\traces(\ar_0) \not\subseteq \traces(\ar'_0)$ we know that there is a trace $\trace \in \traces(\ar_0)$ such that $\trace \not\in \traces(\ar'_0)$. However, $\traces(\ar'_0)$ will contain a trace $\trace'$ which starts with a (possibly empty) prefix of $\trace$. Therefore:

\[
\begin{array}{l}
\trace = \ar_0 \rrderiv^{r_1} \ar_1 \rrderiv^{r_2}  \ldots \rrderiv^{r_n} \ar_n \rrderiv^{r_{n+1}} \ar_{n+1}\\
\trace' = \ar_0' \rrderiv^{r_1} \ar'_1 \rrderiv^{r_2} \ldots \rrderiv^{r_n} \ar'_n 
\end{array}
\]
but $\ar_n' \not \rrderiv^{r_{n+1}} \ar_{*}'$. Furthermore, from point 3 we have that $\variation{\ar_0}{\ar_0'}{T}$.
Assuming $\ar_n = \tuple{q^j,\stkbc',\tv_n} \cdot \ar''_n$ then 
by Lemma~\ref{teo:usefulTrace} we know that there is a variable $z \in \dom(\tv_n)$ such that
\begin{enumerate}
	\item $\useful{j}{T}{z}$
	\item $\valdep{x}{*T}{z}$
	\item $T \in \prmap(j)(z)$
\end{enumerate}
Therefore by Lemma~\ref{teo:typednormsTraces} we have that $T \in \prmap(i)(x)$.
\end{proof}	

\begin{lemma}\label{teo:usefulTrace}
Consider a configuration $\ar_0 = \tuple{p^i,\bc{\cdot}\stkbc,\tv} \cdot \ar$ containing statements from the i-th rule of $\rbrprog{}$, and a variant configuration $\variation{\ar_0 \equiv \tuple{p^i,\bc{\cdot}\stkbc,\tv} \cdot \ar}{\tuple{p^i,\bc{\cdot}\stkbc,\tv'} \cdot \ar' \equiv \ar'_0}{T}$. Let $\{\bar{x}\} \subseteq \dom(\tv)$ be those variables whose value has changed from $\ar_0$ to $\ar'_0$, and $\prmap$ be the result of the typed-norms inference of the program.
If there are traces

\[
\begin{array}{l}
\trace = \ar_0 \rrderiv^{r_1} \ar_1 \rrderiv^{r_2}  \ldots \rrderiv^{r_n} \ar_n \\ 
\trace' = \ar_0' \rrderiv^{r_1} \ar'_1 \rrderiv^{r_2} \ldots \rrderiv^{r_n} \ar'_n \\ 
\end{array}
\]
such that $\ar_{n} \rrderiv^{r_{n+1}} \ar_{n+1}$ but $\ar'_{n} \not\rrderiv^{r_{n+1}} \ar'_{*}$, 
where $\ar_n = \tuple{q^j, \slot,\tv_n} \cdot \ar''_n$, then $\variation{\ar_n}{\ar'_n}{T}$ and there is a variable $z \in \dom(\tv_n)$ such that: 
\begin{enumerate}
	\item $\useful{j}{T}{z}$
	\item $\valdep{x_k}{*T}{z}$ for some $x_k \in \{\bar{x}\}$
	\item $T \in \prmap(j)(z)$
\end{enumerate}
\end{lemma}

\begin{proof}
By induction on the length $n$ of the traces.

\begin{itemize}
	\item \underline{Base Case:} $n = 0$\\
	In this case $\ar_0 \rrderiv^{r_1} \ar_1$ but $\ar_0' \not \rrderiv^{r_1} \ar'_1$, where
	$\variation{\tv}{\tv'}{T}$.
	According to the rules in Fig.~\ref{fig:rrsem} the only rule that can prevent such a step is (2):
	
	\[
	\begin{array}{lc}
	(2) & \semrule
	{
		\bc \equiv m(\bar{x},\bar{y})~~~~ m(\bar{x'},\bar{y'}) \rulearrow g, \bc_1 \cdot \cdot \cdot \bc_k \in \rbrprog{} ~\mi{fresh}\\
		\tv_1 \equiv [\overline{x' \mapsto \tv(x)}] ~~ \fevalg(g,\tv_1) = \tv_2
		
	}
	{
		\tuple{p^i,\bc{\cdot}\stkbc,\tv}{\cdot} \ar 
		\rrderiv
		\tuple{m^j,\bc_1 \cdot \cdot \cdot \bc_k,\tv_1 \uplus \tv_2} {\cdot} \tuple{p^i[\overline{y' \sim y}],\stkbc,\tv}{\cdot} \ar
	}
	\end{array}
	\]
	
	$\ar'_0$ cannot be reduced because $\fevalg(g, [\overline{x' \mapsto \tv'(x)}])$ is undefined, but on the other hand $\fevalg(g, [\overline{x' \mapsto \tv(x)}])$ returns a variable mapping. Then it is clear that the value of some variable $z \in \{\overline{x'}\}$ must differ from $\tv$ to $\tv'$. If there is only differences in one variable $z$ then $\useful{i}{T}{z}$ (in this case $i = j$) because clearly $\traces(\tuple{p^i,\bc{\cdot}\stkbc,\tv} \cdot \ar) \neq \traces(\tuple{p^i,\bc{\cdot}\stkbc,\tv'} \cdot \ar)$. If there are several variables ${\overline{z}} \subseteq \dom(\tv)$ such that $\variation{\tv(z)}{\tv'(z)}{T}$ then it is clear that we can revert the values of all of them but one, which would prevent the guard evaluation, so again $\useful{i}{T}{z}$. In this case $\valdep{z}{*T}{z}$ trivially, because there is no step involved. Finally, by definition of $\genG{\rbrprog{}}_i$ and $\genS{\rbrprog{}}_i$ ($2^{nd}$ rule, set A) the typed-norm $T$ will be propagated from the guard to the input variable, and then to $z$, therefore $T \in \prmap(j)(z)$. Note that  $\variation{\ar_0}{\ar'_0}{T}$ holds trivially in this case.
	
	\item \underline{Inductive Step:} $n > 0$\\
	We know that the traces from $\ar_1$ and $\ar_1'$ are different, so $\ar_1 \neq \ar_1'$ and  by Prop.~\ref{teo:variation_step} we have that $\variation{\ar_1}{\ar'_1}{T}$ with $\ar_1 = \tuple{q^k, \slot, \tv_1} \cdot \ar$ and $\ar_1 = \tuple{q^k, \slot, \tv'_1} \cdot \ar'$. Then by IH we know $\variation{\ar_n}{\ar'_n}{T}$ and (1) $\useful{j}{T}{z}$, (2) $\valdep{y}{*T}{z}$ for some $y \in \dom(\tv_1)$ and (3) $T \in \prmap(j)(z)$. 
	
	If $y \in \dom(\tv)$ then $\variation{\tv(y)}{\tv'(y)}{T}$ and the proof is finished. Otherwise, consider the set of variables $\{\bar{x}\} \subseteq \dom(\tv)$ whose value has changed from $\ar_0$ to $\ar'_0$, i.e., $\variation{\tv(x_k)}{\tv'(x_k)}{T}$. Then by a case distinction on the definition of dependent variables (Def.~\ref{def:dependent_vars}) we have that there is at least one variable 
	$x_k \in \dom(\tv)$ 
	such that $\valdep{x_k}{T}{y}$, therefore $\valdep{x_k}{*T}{z}$.

\end{itemize}	
\end{proof}

\begin{lemma}\label{teo:typednormsTraces}
	Let $\prmap$ be the result of the typed-norms inference of the program.
Consider configurations $\ar_0 = \tuple{p^i,\stkbc,\tv} \cdot \ar'$ and $\ar_n = \tuple{q^j,\stkbc',\tv_n} \cdot \ar''$. 
If $\ar_0 \rrderiv^{r_1} \ar_1 \rrderiv^{r_2}  \ldots \rrderiv^{r_n} \ar_n$, $\valdep{x}{*T}{z}$ for some $x \in \dom(\tv)$, and $T \in \prmap(j)(z)$, then $T \in \prmap(i)(x)$.
\end{lemma}

\begin{proof}
By induction on the length $n$ of the trace:

\begin{itemize}
	\item \underline{Base Case:} $n = 0$\\
	In this case $\ar_0 = \ar_n$, $q^j = p^i$, and $\valdep{x}{*T}{x}$. Therefore $T \in \prmap(i)(x) = \prmap(j)(x)$.
	\item \underline{Inductive Step} $n > 0$\\
	We assume that $\ar_1 = \tuple{r^k,\stkbc'',\tv''} \cdot \ar''$. The dependence $\valdep{x}{*T}{z}$ appears in the complete trace, so there are two cases:
	\begin{itemize}
	\item $x \in \dom(lv'')$. Then by the Induction Hypothesis we have that $T \in \prmap(k)(x)$.
	  If $x \in \dom(lv'')$ and $x \in \dom(lv)$ then $p^i$ and $r^k$ are the same activation record, so $T \in \prmap(i)(x)$.
		
    \item $x \not\in \dom(lv'')$, so $\valdep{x}{T}{z'}$ in the first step and $\valdep{z'}{*T}{z}$ in the rest of the trace, for some $z' \in \dom(\tv'')$. Then by the Induction Hypothesis we have that $T \in \prmap(k)(z')$.
    From $T \in \prmap(k)(z')$, $\valdep{x}{T}{z'}$ and Lemma~\ref{lemma:typednormsStep}, 
		we obtain directly that $T \in \prmap(i)(x)$.	
	\end{itemize} 
\end{itemize}
\end{proof}

\begin{prop}\label{teo:variation_step}
If $\ar_0 \rrderiv^{r} \ar_1$, $\ar'_0 \rrderiv^{r} \ar'_1$, $\variation{\ar_0}{\ar'_1}{T}$, and $\ar'_0 \neq \ar'_1$ then $\variation{\ar_1}{\ar'_1}{T}$.
\end{prop}
\begin{proof}
By case distinction on the rule of Fig.~\ref{fig:rrsem} used. Note the importance of the premise $\ar'_0 \neq \ar'_1$, as rule (3) of Fig.~\ref{fig:rrsem} drops an activation record from the configuration and could remove those variables storing variations, therefore making $\ar'_0$ equal to $\ar'_1$.
\end{proof}

\end{document}